\documentclass{siamart1116}

\usepackage{mathtools}
\usepackage{graphicx}
\usepackage{amsfonts}
\usepackage{amsmath}
\usepackage{tikz}
\usepackage{pgf}
\usetikzlibrary{decorations.pathreplacing,matrix}
\usetikzlibrary{arrows.meta,bending,matrix,positioning}
\usepackage{amsmath}
\usepackage{amssymb}
\usepackage{epstopdf}
\usepackage{geometry} 
\usepackage{float}
\usepackage{algorithm}
\usepackage[noend]{algpseudocode}
\usepackage{color}
\usepackage{todonotes}
\usepackage{array}
\usepackage{listings}
\usepackage{rotating}

\ifpdf
  \DeclareGraphicsExtensions{.eps,.pdf,.png,.jpg}
\else
  \DeclareGraphicsExtensions{.eps}
\fi

\newtheorem{remark}[theorem]{Remark}
\newcommand{\norm}[1]{\left\lVert#1\right\rVert}
\definecolor{donkergroen}{rgb}{0.0, 0.5, 0.0}

\newcommand{\TheTitle}{The communication-hiding Conjugate Gradient method with deep pipelines} 
\newcommand{\TheAuthors}{Jeffrey Cornelis, Siegfried Cools and Wim Vanroose}

\headers{Conjugate Gradient method with deep pipelines}{Cornelis, Cools and Vanroose}

\title{{\TheTitle}\thanks{Submitted to the editors on \today.
\funding{This work was funded by the Research Council of the University of Antwerp under the University Research Fund (BOF) (J.\,Cornelis) and the Research Foundation Flanders (FWO) under grant 12H4617N (S.\,Cools).}}}

\author{
  Jeffrey Cornelis\thanks{Applied Mathematics Group, Department of Mathematics and Computer Science, University of Antwerp, Building G, Middelheimlaan 1, 2020 Antwerp, Belgium.} \and
  Siegfried Cools$^{\dagger}$\and
  Wim Vanroose$^{\dagger}$
}

\usepackage{amsopn}

\ifpdf
\hypersetup{
  pdftitle={\TheTitle},
  pdfauthor={\TheAuthors}
}
\fi

\begin{document}

\maketitle

\begin{abstract}
  Krylov subspace methods are among the most efficient solvers for large scale linear algebra problems. 
	Nevertheless, classic Krylov subspace algorithms do not scale well on massively parallel hardware due 
	to synchronization bottlenecks. Communication-hiding pipelined Krylov subspace methods offer increased 
	parallel scalability by overlapping the time-consuming global communication phase with computations such 
	as \textsc{spmv}s, hence reducing the impact of the global synchronization and avoiding processor idling. 
	One of the first published methods in this class is the pipelined Conjugate Gradient method (p-CG). 
	However, on large numbers of processors the communication phase may take much longer than the computation 
	of a single \textsc{spmv}. This work extends the pipelined CG method to deeper pipelines, denoted as p($l$)-CG, 
	which allows further scaling when the global communication phase is the dominant time-consuming factor. 
	By overlapping the global all-to-all reduction phase in each CG iteration with the next $l$ \textsc{spmv}s 
	(deep pipelining), the method hides communication latency behind additional computational work. 
	The p($l$)-CG algorithm is derived from similar principles as the existing p($l$)-GMRES method and by exploiting 
	operator symmetry. The p($l$)-CG method is also compared 
	to other Krylov subspace methods, including the closely related classic CG and D-Lanczos methods and the 
	pipelined CG method by Ghysels et al.. By analyzing the maximal accuracy attainable by the p($l$)-CG 
	method it is shown that the pipelining technique induces a trade-off between performance and numerical stability.
	A preconditioned version of the algorithm is also proposed and storage requirements and performance estimates 
	are discussed. Experimental results demonstrate the possible performance gains and the 
	attainable accuracy of deeper pipelined CG for solving large scale symmetric linear systems. 
\end{abstract}


\vspace{-0.3cm}


\section{Introduction} \label{sec:intro}

Krylov subspace methods \cite{greenbaum1997iterative,liesen2012krylov,meurant1999computer,saad2003iterative,van2003iterative} are well-known as efficient iterative solvers for large scale linear systems of the general form $Ax =b$, where $A$ is an $n \times n$ matrix and $b \in \mathbb{R}^n$. 
These iterative algorithms construct a sequence of approximate solutions $\{x_i\}_i$ with $x_i \in x_0 + \text{span}\{r_0,Ar_0,A^2r_0,\ldots,A^{i-1}r_0\}$, where $r_0 = b-Ax_0$ is the initial residual. The Conjugate Gradient (CG) method \cite{hestenes1952methods}, which allows for the solution of systems with symmetric and positive definite (SPD) matrices $A$, is generally considered as the first Krylov subspace method. Driven by the ongoing transition of hardware towards the exascale regime, research on the scalability of Krylov subspace methods on massively parallel architectures has recently (re)gained attention in the scientific computing community \cite{dongarra2011international,
dongarra2013toward,dongarra2015hpcg,fuller2011future}. Since for many applications the system operator $A$ is sparse (e.g.~given by a local stencil) and thus rather inexpensive to apply in terms of computational and communication cost, the main bottleneck for efficient parallel execution is typically not the sparse matrix-vector product (\textsc{spmv}), but the communication overhead due to global reductions in dot product computations and the related global synchronization bottleneck. 
A dot product of two distributed vectors requires the local computation of the dot product contributions, followed by a global reduction tree of height $\mathcal{O}(\log(N))$, where $N$ is the number of nodes. 
The explicit synchronization of processes involved in this reduction makes the dot product one of the most time consuming operations in the Krylov subspace algorithm on large parallel hardware.

Over the past decades there have been a variety of scientific efforts to reduce or eliminate the synchronization bottleneck in Krylov subspace methods. The earliest papers on synchronization reduction date back to the late 1980's and 1990's \cite{strakovs1987effectivity,d1992reducing,demmel1993parallel,erhel1995parallel}. A notable reduction of the number of global synchronization points was introduced by the so-called $s$-step methods \cite{chronopoulos1989s,chronopoulos1996parallel,chronopoulos2010block,carson2013avoiding,carson2014residual,imberti2017varying}. Other scalable approaches to Krylov subspace methods include hierarchical \cite{mcinnes2014hierarchical}, enlarged \cite{grigori2016enlarged} and an iteration fusing \cite{zhuang2017iteration} Krylov subspace methods. In addition to avoiding communication, research on hiding global communication by overlapping communication with computations can be found in the literature \cite{demmel1993parallel,de1995reducing,ghysels2013hiding,yamazaki2017improving,sanan2016pipelined,eller2016scalable}. The current work is situated in the latter branch of research.

Introduced in 2014, the so-called ``pipelined'' CG method (p-CG) \cite{ghysels2014hiding} aims at hiding global synchronization latency by overlapping the communication phase in the Krylov subspace algorithm by the application of the \textsc{spmv}. 
Hence, idle core time is reduced by simultaneous execution of the time-consuming synchronization phase and independent compute-bound calculations. The reorganization of the algorithm introduces auxiliary variables, resulting in several additional \textsc{axpy} ($y \leftarrow \alpha x +y$) operations required to recursively compute updates for these variables. Since these are local operations, the extra recursions have no impact on the communication flow of the algorithm. However, they may influence numerical stability, as analyzed by the authors in \cite{cools2018analyzing}. Additionally, other error effects (e.g.\,delayed convergence due to loss of basis orthogonality in finite precision \cite{greenbaum1992predicting,gutknecht2000accuracy,strakovs2002error,gergelits2014composite,carson2016numerical}, hard faults and soft errors \cite{agullo2016hard}
, etc.) may affect the convergence of pipelined CG. 

The original pipelined CG method performs well on problem/hardware setups where the time to compute the \textsc{spmv} (plus preconditioner application, when applicable) roughly equals the time spent communicating in the global reduction phase, such that a good overlap can be achieved. In this optimal scenario the time to solution can be reduced by a factor $3\times$ compared to classic CG on sufficiently large numbers of processors, see \cite{ghysels2014hiding}, Table 1 and Section 5. However, in situations where the global reduction takes significantly longer than one \textsc{spmv}, the performance gain of p-CG over classic CG may be less pronounced. In heavily communication-bound scenarios, which are typically encountered when solving large-scale problems on HPC hardware, a \emph{deeper pipeline} would be required to overlap the global reduction phase with the computational work of multiple \textsc{spmv}s. 

The concept of deep pipelines was introduced by Ghysels et al.\,\cite{ghysels2013hiding} for the Generalized Minimal Residual (GMRES) method. In the current work we extend the notion of deep pipelining to the CG method. Assuming the symmetry of the system matrix $A$, we establish theoretical properties that allow to derive the algorithm starting from the pipelined Arnoldi process in the original p($l$)-GMRES algorithm \cite{ghysels2013hiding}. Subsequently, it is shown that the p($l$)-CG algorithm has several interesting properties compared to p($l$)-GMRES, including shorter recurrences for the basis vectors and the Hessenberg elements and significantly reduced storage requirements. We also indicate limitations of the p($l$)-CG method, in particular in comparison to the (length-one) pipelined CG method \cite{ghysels2014hiding} which was derived using a notably different framework to the one presented here, see also \cite{cools2017communication}.

Reorganizing a Krylov subspace algorithm into a communication reducing variant typically introduces issues with the numerical stability of the algorithm. In exact arithmetic the pipelined CG method produces a series of iterates identical to the classic CG method. However, in finite precision arithmetic their behavior can differ significantly as local rounding errors may decrease attainable accuracy and induce delayed convergence. The impact of round-off errors on numerical stability of classic CG has been extensively studied 
\cite{greenbaum1989behavior,greenbaum1992predicting,greenbaum1997estimating,gutknecht2000accuracy,strakovs2002error,strakovs2005error,meurant2006lanczos,gergelits2014composite}. Similar observations have been made for other classes of communication reducing methods, see e.g.~\cite{chronopoulos1989s,carson2014residual} for the influence of the $s$-step parameter on the numerical stability of communication avoiding methods, and \cite{cools2018analyzing,carson2016numerical} for an overview of the stability analysis of the 
pipelined Conjugate Gradient method proposed in \cite{ghysels2014hiding}.

The remainder of this work is structured as follows. In Section \ref{sec:deriv} we introduce the mathematical context and notations of this paper by revisiting the p($l$)-GMRES method.
We then show how 
the p($l$)-Arnoldi process simplifies in the case of a symmetric system matrix $A$ 
and derive the p($l$)-CG algorithm with pipelines of general length $l$. We also comment on a variant of the algorithm that includes preconditioning. Section \ref{sec:imple} gives an overview of some crucial implementation issues and corresponding solutions related to the p($l$)-CG algorithm. This section contributes to a better understanding of  key aspects of the method from a performance point of view. 
Section \ref{sec:analysis} analyzes the behavior of local rounding errors that stem from the multi-term recurrence relations in 
p($\ell$)-CG.
We characterize the propagation of local rounding errors throughout the algorithm and discuss the influence of the pipelined length $l$ and the choice of the auxiliary Krylov basis on the maximal accuracy attainable by p($l$)-CG. 
The numerical analysis is limited to the effect of local rounding errors on attainable accuracy; a discussion of the loss of orthogonality \cite{greenbaum1992predicting,gutknecht2000accuracy} and consequential delay of convergence in pipelined CG is beyond the scope of this work.
Numerical experiments validating the p($l$)-CG method are provided in Section \ref{sec:experiments}. These illustrate the attainable speed-up of deeper pipelines on a distributed multicore hardware setup, but also comment on the possibly reduced attainable accuracy when longer pipelines are used. 
The paper concludes by presenting a summary of the current work and a short discussion on future research directions in Section \ref{sec:conclusions}.

\section{From $l$-length pipelined GMRES to $l$-length pipelined CG} \label{sec:deriv}

We use the classic notation $V_k = [v_0,\ldots,v_{k-1}]$ for the orthonormal basis of the $k$-th Krylov subspace $\mathcal{K}_k(A,v_0)$. Here the index $k$ denotes the number of basis vectors $v_j$ in $V_k$, with indices $j$ ranging from $0$ up to $k-1$. The length of each basis vector $v_j$ is the column dimension of the system matrix $A$. 

\subsection{Brief recapitulation of p($l$)-GMRES}
Let $V_{i-l+1}:=[v_{0},v_{1},\ldots,v_{i-l}]$ be the orthonormal basis for the Krylov subspace $\mathcal{K}_{i-l+1}(A,v_{0})$. These vectors satisfy the Arnoldi relation $AV_j = V_{j+1} H_{j+1,j}$ for $1 \leq j \leq i-l$, where $H_{j+1,j}$ is the $(j+1) \times j$ upper Hessenberg matrix. This translates in vector notation to:
\begin{equation}\label{plbetr}
v_{j}=\frac{Av_{j-1} - \sum_{k=0}^{j-1} h_{k,j-1} v_{k}}{h_{j,j-1}}, \qquad 1 \leq j \leq i - l.
\end{equation}
We define the auxiliary vectors $Z_{i+1}:=[z_{0},z_{1},\ldots,z_{i-l},z_{i-l+1},\ldots,z_{i}]$ as 
\begin{equation} \label{eq:defz}
z_{j}:= \left\{ \begin{matrix} 
v_{0}, & j=0, \\ 
P_{j}(A)v_{0}, & 0<j\leq l, \\ 
P_{l}(A)v_{j-l}, & j>l, 
\end{matrix} \right. 
\qquad \text{with} \qquad P_{i}(t) := \prod_{j=0}^{i-1} (t-\sigma_{j}), \qquad  \text{for} ~ i\leq l,
\end{equation}
where the polynomials $P_{i}(t)$ are defined 
with shifts $\sigma_{j}\in \mathbb{R}$ that will be specified later. The basis $Z_{i+1}$ can alternatively be defined using three-term recurrences, where we refer to Remark \ref{remark:remark4} and \cite{ghysels2013hiding} (Section 4.3) for more details on choosing the basis. Note that for any $i \geq 0$ the bases $V_{i+1}$ and $Z_{i+1}$ span the same Krylov subspace. 
One has the following recurrence relations for successive $z_{j}$:
\begin{equation} \label{zid2}
z_{j} = \left\{ \begin{matrix}  
(A-\sigma_{j-1}I)z_{j-1}, & 0 < j \leq l, \\
(Az_{j-1} - \sum_{k=0}^{j-l-1} h_{k,j-l-1}z_{k+l})/h_{j-l,j-l-1}, &  l < j \leq i. 
\end{matrix} \right. 
\end{equation}
These recursive relations for the vectors $z_{j}$ can be summarized in the Arnoldi-type matrix identity
\begin{equation} \label{matid}
AZ_{i}=Z_{i+1}B_{i+1,i}, \qquad \text{with} \qquad 
B_{i+1,i}
= 
{\small
\left(\begin{array}{ccc|ccc} 
\sigma_{0} & & & & & \\
1 & \ddots & & & & \\
& \ddots & \sigma_{l-1} & & & \\ \hline
& & 1 & & & \\
& & & & H_{i-l+1,i-l} & \\
& & & & & 
\end{array}\right).
}
\end{equation}
\noindent


\begin{theorem}\label{zisvg}
\emph{[Ghysels et al.\,\cite{ghysels2013hiding}]} Suppose $k > l$ and let $V_{k}$ be an orthonormal basis for the $k$-th Krylov subspace $\mathcal{K}_{k}(A,v_{0})$. Let $Z_{k}$ be a set of vectors defined by $\eqref{eq:defz}$. Then the identity $Z_{k}=V_{k}G_{k}$ holds 
with $G_{k}$ an upper triangular $k\times k$ matrix. The entries of the last column of $G_{k}$, i.e.\,$g_{j,k-1} = (z_{k-1},v_{j})$ (with $j=0,\ldots,k-1$), can be computed using the elements of $G_{k-1}$ and the dot products $(z_{k-1},v_{j})$ that are available for $j\leq k-l-1$ and $(z_{k-1},z_{j})$ for $k-l-1 < j \leq k-1$: 
\begin{equation} \label{eq:gjk}
g_{j,k-1} = \frac{(z_{k-1},z_{j})-\sum_{m=0}^{j-1}g_{m,j}g_{m,k-1}}{g_{j,j}} \quad \text{and} \quad
g_{k-1,k-1}=\sqrt{(z_{k-1},z_{k-1})-\sum_{m=0}^{k-2}g_{m,k-1}^{2}}. 
\end{equation}
\end{theorem}

\noindent Given the set of vectors $Z_{i-l+2}$, the basis $V_{i-l+1}$ can 
be extended to $V_{i-l+2}$ by applying Theorem \ref{zisvg} with $k=i-l+2$. As soon as the dot products $(z_{i-l+1},v_{j})$ for $0 \leq j \leq i-2l+1$ and $(z_{i-l+1},z_{j})$ for $i-2l+1< j \leq i-l+1$ are calculated, the vector $v_{i-l+1}$ can be computed recursively as
\begin{equation} \label{eq:vil1}
v_{i-l+1} = \frac{z_{i-l+1}-\sum_{j=0}^{i-l}g_{j,i-l+1}v_{j}}{g_{i-l+1,i-l+1}},
\end{equation}
i.e.~using the identity $Z_{i-l+2}=V_{i-l+2}G_{i-l+2}$.
For $k > 0$, the Hessenberg matrix $H_{k+1,k}$ can be computed from the matrices $H_{k,k-1}$, $G_{k+1}$ and $B_{k+1,k}$ in a column-wise fashion as follows.

\begin{theorem} \label{hessenberg}
\emph{[Ghysels et al.\,\cite{ghysels2013hiding}]} Assume $k > 0$. Let $G_{k+1}$ be the upper triangular basis transformation matrix for which $Z_{k+1} = V_{k+1} G_{k+1}$ and let $B_{k+1,k}$ be the upper Hessenberg matrix that connects the vectors in $Z_k$ via $AZ_{k}=Z_{k+1}B_{k+1,k}$. Then the Hessenberg matrix for the basis $V_{k+1}$ can be constructed column by column as
\begin{equation} \label{hessup}
H_{k+1,k} = \begin{bmatrix} H_{k,k-1} & (G_{k}b_{:,k-1}+g_{:,k}b_{k,k-1}-H_{k,k-1}g_{:,k-1})g^{-1}_{k-1,k-1} \\ 0 & g_{k,k}b_{k,k-1}g^{-1}_{k-1,k-1}\end{bmatrix}.
\end{equation}
\end{theorem}
\begin{proof}
This follows directly from $H_{k+1,k} = G_{k+1}B_{k+1,k}G^{-1}_{k}$, which is derived using the identity $AV_{k} = V_{k+1}H_{k+1,k}$, the transformation $Z_{k+1} = V_{k+1} G_{k+1}$ and the relation $AZ_{k}=Z_{k+1}B_{k+1,k}$.
\end{proof}

\begin{algorithm}[t]
{\small
\caption{Pipelined GMRES method ($p(l)-$GMRES) \hfill \textbf{Input:} $A$, $b$, $x_0$, $l$, $m$}\label{algo:plGMRES}
\begin{algorithmic}[1]
\State $r_{0}:=b-Ax_{0};$ 
\State $v_{0}:= r_{0}/\|r_{0}\|_{2};$
\State $z_{0}:=v_{0}; ~ g_{0,0}:=1;$
\For {$i=0,\ldots,m+l$}
\State $z_{i+1}:=\left\{ \begin{matrix}(A-\sigma_{i}I)z_{i}, & i<l \\ Az_{i}, & i \geq l \end{matrix}\right.$
\If {$i\geq l$}
\State $g_{j,i-l+1} := (g_{j,i-l+1}-\sum_{k=0}^{j-1}g_{k,j}g_{k,i-l+1})/g_{j,j}; \qquad j=i-2l+2,\ldots,i-l$  
\State $g_{i-l+1,i-l+1} :=\sqrt{g_{i-l+1,i-l+1}-\sum_{k=0}^{i-l}g_{k,i-l+1}^2};$
\State \# Check for breakdown and restart if required
\If {$i<2l$}
\State $h_{j,i-l}:=(g_{j,i-l+1}+\sigma_{i-l}g_{j,i-l}-\sum_{k=0}^{i-l-1}h_{j,k}g_{k,i-l})/g_{i-l,i-l}; \qquad j=0,\ldots,i-l$
\State $h_{i-l+1,i-l}:=g_{i-l+1,i-l+1}/g_{i-l,i-l}$;
\Else
\State $h_{j,i-l}:=(\sum_{k=0}^{i-2l+1}g_{j,k+l}h_{k,i-2l}-\sum_{k=j-1}^{i-l-1}h_{j,k}g_{k,i-l})/g_{i-l,i-l}; \qquad j=0,\ldots,i-l$
\State $h_{i-l+1,i-l}:=(g_{i-l+1,i-l+1}h_{i-2l+1,i-2l})/g_{i-l,i-l};$
\EndIf
\State \textbf{end if}
\State $v_{i-l+1} := (z_{i-l+1} - \sum_{j=0}^{i-l} g_{j,i-l+1} v_{j})/g_{i-l+1,i-l+1};$
\State $z_{i+1} := (z_{i+1} - \sum_{j=0}^{i-l} h_{j,i-l} z_{j+l})/h_{i-l+1,i-l};$
\EndIf
\State \textbf{end if}
\State $g_{j,i+1}:=\left\{ \begin{matrix}(z_{i+1},v_{j}); & j=0,\ldots,i-l+1 \\ (z_{i+1},z_{j});  &j=i-l+2,\ldots,i+1 \end{matrix}\right.$
\EndFor 
\State \textbf{end for}
\State $y_{m}:= \text{argmin} \|H_{m+1,m}y_{m}-\norm{r_{0}}_{2}e_{1}\|_{2};$
\State $x_{m} := x_{0}+V_{m}y_{m};$
\end{algorithmic}
}
\end{algorithm}

\noindent Once the Hessenberg matrix $H_{i-l+2,i-l+1}$ has been computed, the basis $Z_{i+1}$ can be extended to $Z_{i+2}$ by adding the vector $z_{i+1}$ that can be computed using the expression \eqref{zid2}. The above considerations lead to the p($l$)-GMRES algorithm shown in Alg.\,\ref{algo:plGMRES}.

\begin{remark} \textbf{Krylov basis choice.} \label{remark:remark4}
The choice of appropriate shift values $\sigma_{i}$ is vital to ensure numerical stability for the p($l$)-GMRES algorithm. The monomial basis $\left[v_{0},Av_{0},\ldots, A^{k} v_{0} \right]$ for $\mathcal{K}_{k+1}(v_{0},A)$ may become ill-conditioned very quickly, which can be resolved by choosing an alternative basis of the form $\left[v_{0},P_{1}(A)v_{0},\ldots, P_{k}(A) v_{0} \right]$ such as the Newton basis $P_{i}(A)=\prod_{j=0}^{i-1}(A-\sigma_{j}I)$. Different shift choices $\sigma_{i}$ are possible, e.g.\,the Ritz values of $A$ or the zeros of the degree-$l$ Chebyshev polynomial. If all eigenvalues are located in $[\lambda_{min},\lambda_{max}]$ (excl.~zero), the Chebyshev shifts are 
\begin{equation} \label{eq:cheb_shifts}
  \sigma_{i} = \frac{\lambda_{\max}+\lambda_{\min}}{2}+\frac{\lambda_{\max}-\lambda_{\min}}{2} \cos\left(\frac{(2i+1)\pi}{2l}\right), \qquad i=0,\ldots,l-1.
\end{equation}
More information 
can be found in the works by Hoemmen \cite{hoemmen2010communication} and 
Ghysels et al.\cite{ghysels2013hiding}.
\end{remark}

\subsection{Deriving p($l$)-CG from p($l$)-GMRES}

We now derive the p($l$)-CG algorithm starting from the Arnoldi procedure in p($l$)-GMRES, Alg.\,\ref{algo:plGMRES}, based on arguments that are similar to the ones used in the classical derivation of CG from GMRES, see e.g.\,\cite{saad2003iterative,van2003iterative,liesen2012krylov}.

\subsubsection{Exploiting the symmetry: the Hessenberg matrix} \label{sec:hess}

Application of p($l$)-GMRES to a symmetric matrix $A$ induces tridiagonalization of the Hessenberg matrix $H_{i-l+2,i-l+1}$.
Thus only three Hessenberg elements need to be computed in each iteration $i$ in lines 11/12 and 14/15 of Alg.\,\ref{algo:plGMRES}, namely $h_{i-l-1,i-l}$, $h_{i-l,i-l}$ and $h_{i-l+1,i-l}$. Due to symmetry $h_{i-l-1,i-l}$ equals $h_{i-l,i-l-1}$, which was already computed in iteration $i-1$. Furthermore, the ranges of the sums in the right-hand side of the expressions for $h_{i-l,i-l}$ and $h_{i-l+1,i-l}$, see Alg.\,\ref{algo:plGMRES} line 11-15, 
are reduced significantly.

\begin{corollary} \label{corollary}
Let $A$ be a symmetric matrix, let $k > 0$ and let the matrices $V_{k+1}$, $Z_{k+1}$, $G_{k+1}$, $B_{k+1,k}$ and $H_{k,k-1}$ as defined in Theorem \ref{hessenberg} be available. Then the tridiagonal Hessenberg matrix $H_{k+1,k}$ for the basis $V_{k+1}$ can be constructed from $H_{k,k-1}$ using the following expressions:
\vspace{-0.1cm}
\begin{equation} \label{eq:hkk1}
h_{k-1,k-1} =  
\left\{
  \begin{aligned} 
    &\sigma_{k-1}+(g_{k-1,k}-g_{k-2,k-1}h_{k-1,k-2})/g_{k-1,k-1}, & k \leq l, \\ 
    &h_{k-l-1,k-l-1}+(g_{k-1,k}h_{k-l,k-l-1} 
		-g_{k-2,k-1}h_{k-1,k-2})/g_{k-1,k-1}, & k > l,
  \end{aligned} 
	\right.
\end{equation} \vspace{-0.3cm}
\begin{equation} \label{eq:hkk2}
h_{k,k-1} = \left\{ 
  \begin{matrix} 
    g_{k,k}/g_{k-1,k-1}, & k \leq l, \\ 
    (g_{k,k}h_{k-l,k-l-1})/g_{k-1,k-1}, & k > l. 
  \end{matrix} 
	\right.
\end{equation}
Note that for the cases $k = 1$ and $k = l+1$ the values $h_{0,-1}$ and $h_{-1,0}$ 
should be considered zero.
\end{corollary}

\begin{proof}
It suffices to compute the entries $h_{k-1,k-1}$ and $h_{k,k-1}$ of $H_{k+1,k}$, since $h_{k-2,k-1} = h_{k-1,k-2}$ due to symmetry and $h_{j,k-1} = 0$ for $0 \leq j \leq k-3$ since $H_{k+1,k}$ is tridiagonal. The expressions \eqref{eq:hkk1} and \eqref{eq:hkk2} are derived directly from Theorem \ref{hessenberg} using the fact that $H_{k+1,k}$ is tridiagonal. 
\end{proof}

\noindent Applying the above theorem with $k = i-l+1$ allows us to compute the Hessenberg matrix $H_{i-l+2,i-l+1}$ in iteration $i$. 
The element $h_{i-l,i-l}$ in $H_{i-l+2,i-l+1}$ is characterized by expression \eqref{eq:hkk1}:
\begin{equation} \label{eq:haa1}
h_{i-l,i-l} =  
\left\{
  \begin{aligned} 
    &(g_{i-l,i-l+1}+\sigma_{i-l}g_{i-l,i-l}-g_{i-l-1,i-l}h_{i-l,i-l-1})/g_{i-l,i-l}, & \hspace{-1cm} l \leq i < 2l, \\ 
    &(g_{i-l,i-l}h_{i-2l,i-2l}+g_{i-l,i-l+1}h_{i-2l+1,i-2l}-g_{i-l-1,i-l}h_{i-l,i-l-1})/g_{i-l,i-l}, & i \geq 2l.
  \end{aligned} 
	\right.
\end{equation}
Note that for $i = l$ the term $-g_{i-l-1,i-l}h_{i-l,i-l-1}/g_{i-l,i-l}$ drops, while for $i = 2l$ the term $g_{i-l,i-l-1}h_{i-2l-1,i-2l}/g_{i-l,i-l}$ should be omitted.
The update for $h_{i-l+1,i-l}$ follows from \eqref{eq:hkk2} by setting $k = i-l+1$, i.e.:
\begin{equation} \label{eq:haa3}
h_{i-l+1,i-l} = \left\{ 
  \begin{matrix} 
    g_{i-l+1,i-l+1}/g_{i-l,i-l}, & l \leq i < 2l, \\ 
    (g_{i-l+1,i-l+1}h_{i-2l+1,i-2l})/g_{i-l,i-l}, & i \geq 2l,
  \end{matrix} 
	\right.
\end{equation}
which is identical to the expressions on lines 12 and 15 in Alg.\,\ref{algo:plGMRES}.
In addition, the recurrence for the auxiliary basis vector $z_{i+1}$, given by \eqref{zid2} for the p($l$)-GMRES method, can due to the symmetry of $A$ be simplified to
a three-term recurrence relation:
\begin{equation} \label{zid3}
z_{i+1} = (A z_{i} - h_{i-l,i-l} z_{i}- h_{i-l-1,i-l} z_{i-1})/h_{i-l+1,i-l}.
\end{equation}

\subsubsection{The band structure of $G_{i+1}$} \label{sec:band}

The symmetry of the matrix $A$ induces a particular band structure and symmetry for the upper triangular basis transformation matrix $G_{i+1}$. 

\begin{lemma}\label{lemma:lemmaband}
Let $A$ be a symmetric matrix, assume $i \geq l$ and let $V_{i+1} = [v_0,v_1,\ldots,v_{i}]$ be the orthogonal basis for $\mathcal{K}_{i+1}(A,v_0)$. Let $Z_{i+1} = [z_0, z_1, \ldots, z_i]$ be the auxiliary basis with 
$z_{i} = P_{l}(A)v_{i-l}$ as defined by \eqref{eq:defz}, then for all $j \notin \{i-2l,i-2l+1,\ldots,i-1,i\}$ it holds that $g_{j,i} = (z_{i},v_j) =  0$.
\end{lemma}  

\begin{proof}
For symmetric $A$ the matrix $G_{i+1}$ is symmetric around its $l$-th upper diagonal, since
\begin{equation} \label{eq:symmetry_G}
	g_{j,i} = (z_{i},v_{j}) = (P_{l}(A)v_{i-l},v_{j}) =  (v_{i-l},P_{l}(A)v_{j})= (v_{i-l},z_{j+l}) = g_{i-l,j+l}.
\end{equation}
As $G_{i+1}$ is an upper triangular matrix, only the elements $g_{i-2l,i}, g_{i-2l+1,i},\ldots, g_{i-1,i}, g_{i,i}$ 
differ from zero and $G_{i+1}$ thus has a band structure with a band width of at most $2l+1$ non-zeros.
\end{proof}

\noindent For completeness we note that the alternative characterization of the matrix $G_{i+1}$ (as a function of the matrix $T_{i+1}$) derived in Appendix \ref{sec:lemma} may serve as an equivalent proof of Lemma \ref{lemma:lemmaband}.

The band structure of the matrix $G_{i+1}$ further simplifies the algorithm. 
In the symmetric case the expression for $g_{j,i-l+1}$ (see Alg.\,\ref{algo:plGMRES}, line 7) that is derived from expression \eqref{eq:gjk} reads:
\begin{equation}
g_{j,i-l+1} = \frac{g_{j,i-l+1}-\sum_{k=i-3l+1}^{j-1}g_{k,j}g_{k,i-l+1}}{g_{j,j}}, \qquad j=i-2l+2, \ldots i-l,
\end{equation}
where the sum includes only the non-zero elements $g_{k,i-l+1}$ for which $i-3l+1 \leq k \leq i-l+1$.
The computation of $g_{i-l+1,i-l+1}$ (Alg.\,\ref{algo:plGMRES}, line 8), see \eqref{eq:gjk}, is treated in a similar way, yielding:
\begin{equation}
g_{i-l+1,i-l+1} =\sqrt{g_{i-l+1,i-l+1}-\sum_{k=i-3l+1}^{i-l}g_{k,i-l+1}^2}.
\end{equation}
Furthermore, by exploiting the band structure of $G_{i+1}$ for symmetric matrices $A$, the recurrence for $v_{i-l+1}$ (Alg.\,\ref{algo:plGMRES}, line 17), given in general by \eqref{eq:vil1}, is rewritten as:
\begin{equation} \label{eq:recur_v}
v_{i-l+1} = \frac{z_{i-l+1} - \sum_{j=i-3l+1}^{i-l} g_{j,i-l+1} v_{j}}{g_{i-l+1,i-l+1}},
\end{equation}
such that the recurrence for $v_{i-l+1}$ is only based on the $2l$ previous basis vectors $v_{i-3l+1},\ldots,v_{i-l}$.

Finally, since it follows from Lemma \ref{lemma:lemmaband} that many entries of $G_{i+1}$ effectively equal zero, 
only the dot products $(z_{i+1},z_j)$ and $(z_{i+1},v_j)$ (see Alg.\,\ref{algo:plGMRES}, line 20) for which $j=\max(0,i-2l+1), \ldots, i+1$ 
need to be computed in iteration $i$ of the algorithm.
We obtain
\begin{equation} \label{eq:gdotpr2}
g_{j,i+1}=\left\{ 
  \begin{matrix}
	  (z_{i+1},v_{j}), & j=\max(0,i-2l+1),\ldots,i-l+1, & \text{(only evaluated if $i \geq l-1$)} \\ 
	  (z_{i+1},z_{j}), & j=i-l+2,\ldots,i+1.
	\end{matrix}
\right. 
\end{equation}

\subsubsection{Towards p($l$)-CG} \label{sec:towards}

To compute the solution $x_m$ after the p($l$)-GMRES iteration has finished and the Krylov subspace basis $V_m$ has been constructed, p($l$)-GMRES minimizes the Euclidean norm of the residual ${\|b-Ax_{m}\|}_{2}$ over the Krylov subspace $\mathcal{K}_{m}(A,r_{0})$ as follows:
\begin{equation} \label{eq:minimi}
  \min_{x_0 + V_m y_m \in \mathbb{R}^n} {\| b-Ax_m \|}_2 = \min_{y_m \in \mathbb{R}^m} {\| r_0 - V_{m+1} H_{m+1,m} y_m\|}_2 = \min_{y_m \in \mathbb{R}^m} {\|{\|r_0\|}_2 e_1 - H_{m+1,m} y_m\|}_2,
\end{equation}
where $y_m$ is a column vector of length $m$.
This leads to a least squares 
minimization problem 
with an $(m+1)\times m$ Hessenberg matrix, see Alg.\,\ref{algo:plGMRES}, line 22. The construction of the solution $x_m$  requires the \emph{entire} Krylov subspace basis $V_m$ in p($l$)-GMRES, which gives rise to a high storage overhead.


In contrast, in the Conjugate Gradient method the Ritz-Galerkin condition $r_{m} = b-Ax_m \, \bot \, V_{m}$ together with the Lanczos relation $AV_{m} = V_{m+1}T_{m+1,m}$ are used to find an expression for the approximate solution over the Krylov subspace $x_{0} + K_{m}(A,r_{0})$. More specifically, this implies:
\begin{equation} \label{eq:lanczos}
  0 = V_m^T r_m  = V_m^T (r_0 - A V_m y_m) = V^T_m \left(V_m\left({\|r_0\|}_2e_1 - T_m y_m\right)\right) = {\|r_0\|}_2e_1 - T_m y_m,
\end{equation}
resulting in $y_m = T_{m}^{-1} {\|r_0\|}_2 e_1$ with a symmetric tridiagonal matrix $T_m$. 

\begin{remark} \textbf{Relation to FOM and MINRES.}
By modifying the p($l$)-GMRES algorithm by changing line $22$ in Alg.\,\ref{algo:plGMRES} to $y_{m} = H^{-1}_{m,m}{\|r_{0}\|}_{2}e_{1}$ one immediately obtains a deep pipelined version of the so-called \textit{Full Orthogonalization Method} (FOM), cf.\,\cite{saad2003iterative}. In this section the p($l$)-CG method is derived as a symmetric variant of the FOM algorithm. As a side-note
we also remark that by only exploiting the symmetry of the matrix but constructing the solution using the minimization procedure \eqref{eq:minimi} like in Alg.\,\ref{algo:plGMRES} one could derive a pipelined version of the Minimal Residual method (MINRES) \cite{paige1975solution,saad2003iterative}, which can be applied to symmetric and indefinite systems. However, we focus on deriving a deep pipelined variant of the more widely used CG method in this work.
\end{remark}

\noindent For notational convenience the elements of the matrix $T_{i-l+2,i-l+1}$ 
are renamed in the symmetric setting. Denote $\gamma_{k} := h_{k,k}$ and $\delta_{k} := h_{k+1,k}$ for any $0 \leq k \leq i-l$. 
Then $T_{i-l+2,i-l+1}$ is completely defined by the two arrays ${(\gamma_{k})}_{k=0,\ldots,i-l}$ and ${(\delta_{k})}_{k=0,\ldots,i-l}$.
The square 
symmetric tridiagonal matrix $T_{i-l+1}$ is obtained by omitting the last row of the Hessenberg matrix $T_{i-l+2,i-l+1}$.
Assume that the LU-factorization of the tridiagonal matrix $T_{i-l+1} = L_{i-l+1} U_{i-l+1}$ is given by
\begin{equation} \label{eq:LU} 
\begin{pmatrix} 
\gamma_{0} & \delta_{0} & & & \\ 
\delta_{0} & \gamma_{1} & \delta_{1} & & \\ 
& \delta_{1} & \gamma_{2} & \ddots & \\ 
& & \ddots & \ddots & \delta_{i-l-1} \\ 
& & & \delta_{i-l-1} & \gamma_{i-l}  
\end{pmatrix} 
= 
\begin{pmatrix} 
1 & & & & \\ 
\lambda_{1} & 1 & & &  \\ 
& \lambda_{2} & 1 & & \\ 
& & \ddots & \ddots &  \\ 
& & & \lambda_{i-l} & 1  
\end{pmatrix}
\begin{pmatrix} 
\eta_{0} & \delta_{0} & & & \\ 
& \eta_{1} & \delta_{1} & & \\ 
& & \eta_{2} & \ddots & \\ 
& & \ & \ddots & \delta_{i-l-1} \\ 
& & & & \eta_{i-l} 
\end{pmatrix}.
\end{equation}
Following the procedure outlined in \cite{van2003iterative,liesen2012krylov}, and notably the derivation of D-Lanczos in \cite{saad2003iterative}, Sec.~6.7.1 (see also Remark \ref{remark:remark7}), we now replace the minimization procedure \eqref{eq:minimi} (Alg.\,\ref{algo:plGMRES}, line 22) by an iterative update of the solution $x_{i-l+1}$ based on a search direction $p_{i-l}$ as defined below. 
Note that $\gamma_{0}=\eta_{0}$ and it follows from \eqref{eq:LU} that $\delta_{k-1}= \lambda_{k}\eta_{k-1}$ and that $\gamma_k = \lambda_{k}\beta_{k}+\eta_{k}$, or equivalently 
\begin{equation}
\lambda_{k}=\delta_{k-1}/\eta_{k-1} \qquad \text{and} \qquad \eta_{k}= \gamma_{k}-\lambda_{k}\delta_{k-1}, \qquad 1 \leq k \leq i-l. 
\end{equation}
These expressions allow to compute the elements of the lower/upper triangular matrices $L_{i-l+1}$ and $U_{i-l+1}$. 
From \eqref{eq:lanczos} it follows that the approximate solution $x_{i-l+1}$ is given by 
\begin{equation} \label{eq:x_a}
x_{i-l+1} = x_0 + V_{i-l+1} y_{i-l+1} =  x_{0} + V_{i-l+1} U_{i-l+1}^{-1}L_{i-l+1}^{-1}\norm{r_{0}}_{2}e_{1} = x_{0}+P_{i-l+1} q_{i-l+1},
\end{equation}
where the search directions are defined as $P_{i-l+1} := V_{i-l+1} U_{i-l+1}^{-1}$ and $q_{i-l+1} := L_{i-l+1}^{-1}\norm{r_{0}}_{2}e_{1}$. Note that $p_0 = v_0/\eta_0$. The columns $p_k$ (for $1 \leq k \leq i-l$) of $P_{i-l+1}$ can easily be computed recursively. Indeed, since $P_{i-l+1} U_{i-l+1} = V_{i-l+1}$, it follows that $v_k = \delta_{k-1} p_{k-1}+\eta_k p_k$ for any $1 \leq k \leq i-l$, yielding the recurrence for the search directions $p_k$:
\begin{equation}\label{eq:pdir}
p_k = \eta_k^{-1}(v_k-\delta_{k-1} p_{k-1}), \qquad 1 \leq k \leq i-l.
\end{equation}
Denoting the elements of the vector $q_{i-l+1}$ by $\left[\zeta_0,\ldots,\zeta_{i-l}\right]^T$, it follows from $L_{i-l+1} q_{i-l+1}=\norm{r_{0}}_{2}e_{1}$ that $\zeta_0 = {\|r_0\|}_2$ and $\lambda_k\zeta_{k-1}+\zeta_{k}=0$ for $1\leq k \leq i-l$. Hence, the scalar $\zeta_k$ is computed in each iteration using the recursion:
\begin{equation}
\zeta_{k} = -\lambda_{k}\zeta_{k-1}, \qquad 1 \leq k \leq i-l.
\end{equation}
Using the search direction $p_{i-l}$ and the scalar $\zeta_{i-l}$, which are both updated recursively in each iteration, the approximate solution $x_{i-l+1}$ is updated using the recurrence relation: 
\begin{equation} \label{eq:x_a2}
x_{i-l+1} =x_{i-l} + \zeta_{i-l}p_{i-l}.
\end{equation}
By merging this recursive update for the solution with the simplifications suggested in Sections \ref{sec:hess} and \ref{sec:band}, we obtain a new iterative scheme which we will denote as $l$-length pipelined CG, or p($l$)-CG for short. The corresponding algorithm is shown in Alg.\,\ref{algo:plCG}.

\begin{algorithm}[t]
{\small
\caption{Pipelined Conjugate Gradient method (p($l$)-CG) \hfill \textbf{Input:} $A$, $b$, $x_0$, $l$, $m$}\label{algo:plCG}
\begin{algorithmic}[1]
\State $r_{0}:=b-Ax_{0};$ 
\State $v_{0}:= r_{0}/{\|r_{0}\|}_2;$
\State $z_{0}:=v_{0}; ~ g_{0,0}:=1;$
\For {$i=0,\ldots, m+l$}
\State $ z_{i+1}:=\left\{ \begin{matrix}(A-\sigma_{i}I)z_{i}, & i<l \\ Az_{i}, & i \geq l \end{matrix}\right.$
\If {$i\geq l$}
\State $g_{j,i-l+1} := (g_{j,i-l+1}-\sum_{k=i-3l+1}^{j-1}g_{k,j}g_{k,i-l+1})/g_{j,j}; \qquad j=i-2l+2,\ldots,i-l$  
\State $g_{i-l+1,i-l+1}:= \sqrt{g_{i-l+1,i-l+1}-\sum_{k=i-3l+1}^{i-l}g_{k,i-l+1}^2};$
\State \# Check for breakdown and restart if required
\If {$i<2l$}
\State $\gamma_{i-l}:=(g_{i-l,i-l+1}+\sigma_{i-l}g_{i-l,i-l}-g_{i-l-1,i-l}\delta_{i-l-1})/g_{i-l,i-l};$
\State $\delta_{i-l}:=g_{i-l+1,i-l+1}/g_{i-l,i-l};$
\Else
\State $\gamma_{i-l}:=(g_{i-l,i-l}\gamma_{i-2l}+g_{i-l,i-l+1}\delta_{i-2l}-g_{i-l-1,i-l}\delta_{i-l-1})/g_{i-l,i-l};$
\State $\delta_{i-l}:=(g_{i-l+1,i-l+1}\delta_{i-2l})/g_{i-l,i-l};$
\EndIf
\State \textbf{end if}
\State $v_{i-l+1} := (z_{i-l+1} - \sum_{j=i-3l+1}^{i-l} g_{j,i-l+1} v_{j})/g_{i-l+1,i-l+1};$
\State $z_{i+1} := (z_{i+1} - \gamma_{i-l} z_{i}- \delta_{i-l-1} z_{i-1})/\delta_{i-l};$
\EndIf
\State \textbf{end if}
\State $g_{j,i+1}:=\left\{ \begin{matrix}(z_{i+1},v_{j}); & j=\max(0,i-2l+1),\ldots,i-l+1 \\ (z_{i+1},z_{j});  &j=i-l+2,\ldots,i+1 \end{matrix}\right.$
\State \textbf{end if}
\If {$i=l$}
\State $\eta_{0}:=\gamma_{0};$
\State $\zeta_{0}:={\|r_{0}\|}_2;$
\State $p_{0}:=v_0/\eta_0;$
\Else \, \textbf{if} {$i\geq l+1$} \textbf{then}
\State $\lambda_{i-l}:=\delta_{i-l-1}/\eta_{i-l-1};$ 
\State $\eta_{i-l}:=\gamma_{i-l}-\lambda_{i-l}\delta_{i-l-1};$
\State $\zeta_{i-l}:=-\lambda_{i-l}\zeta_{i-l-1};$
\State $p_{i-l}:=(v_{i-l}-\delta_{i-l-1}p_{i-l-1})/\eta_{i-l};$
\State $x_{i-l}:=x_{i-l-1}+\zeta_{i-l-1}p_{i-l-1};$
\EndIf
\State \textbf{end if}
\EndFor 
\State \textbf{end for}
\end{algorithmic}
}
\end{algorithm}

\begin{remark} \textbf{CG vs.\,D-Lanczos.} \label{remark:remark7}
An important remark on nomenclature should be made here. As indicated earlier, Alg.\,\ref{algo:plCG} is mathematically equivalent (i.e.\,in exact arithmetic) to the direct Lanczos (or \emph{D-Lanczos}) method \cite{saad2003iterative}, rather than the CG method. Indeed, Alg.\,\ref{algo:plCG} could alternatively be called `\emph{p($l$)-D-Lanczos}'. The difference between the two methods is subtle. Unlike classic CG the D-Lanczos algorithm may break down even in exact arithmetic 
due to a possible division by zero in the recurrence relation \eqref{eq:pdir}, which stems from implicit Gaussian elimination without pivoting. Apart from this possible (yet rarely occurring) instability, Alg.\,\ref{algo:plCG} is mathematically equivalent to CG and their convergence histories coincide (for any choice of $l$), see Fig.\,\ref{fig:residuals}. Since Alg.\,\ref{algo:plCG} is intrinsically based on residual orthogonality and search direction $A$-orthogonality (i.e.~the key properties of the CG method), we denote Alg.\,\ref{algo:plCG} as p($l$)-CG.
Alternative formulations of the CG algorithm could be used to derive other pipelined 
variants that are mathematically equivalent to CG. For example, three-term Conjugate Gradients \cite{saad2003iterative} can be rewritten into a pipelined variant that produces iterates identical to those of CG in exact arithmetic; we do not expound on the details of this method here.
\end{remark}

\begin{remark} \textbf{Square root breakdown.} \label{remark:sqrt_breakdown}
Unlike other more common variants of the CG algorithm \cite{saad1984practical,meurant1987multitasking,chronopoulos1989s,ghysels2014hiding}, 
the p($l$)-CG method computes a square root to calculate $g_{i-l+1,i-l+1}$ (Alg.\,\ref{algo:plCG}, line 8) and may break down when the  
root argument becomes negative or zero, just like p($l$)-GMRES. When $g_{i-l+1,i-l+1}-\sum_{k=i-3l+1}^{i-l}g_{k,i-l+1}^2 = 0$ (or sufficiently close to zero in finite precision arithmetic) a happy breakdown occurs, implying the solution has been found. A value 
smaller than zero signals loss of basis orthogonality and results in a hard breakdown caused by numerical rounding errors in finite precision arithmetic. In this case the algorithm has not converged and a restart 
or a re-orthogonalization of the Krylov subspace basis is required. 
In this work we opt for an explicit restart when a square root breakdown occurs, 
using the last computed solution $x_k$ as the new initial guess. Other restart or re-orthogonalization strategies are possible 
\cite{bai2000templates,saad2003iterative}, but are beyond the scope of this work. 
We stress that the square root breakdown is intrinsic to the pipelining procedure and  
that it is unrelated to the possible instability 
in the D-Lanczos method pointed out in Remark \ref{remark:remark7}.
\end{remark}

\subsubsection{The residual norm in p($l$)-CG}

The derivation of the p($l$)-CG method follows the classic procedure from \cite{saad2003iterative,van2003iterative,liesen2012krylov} but does not include a recurrence relation for the residual, similarly to the p($l$)-GMRES algorithm. This is in contrast with most (communication reducing) variants of CG, where the residual is typically computed recursively in each iteration to update the solution, see e.g.\,\cite{ghysels2014hiding,carson2015communication} and Remark \ref{remark:remark9}.
However, the residual norm is a useful measure of deviation from the 
solution that allows (among others) to formulate 
stopping criteria. 
Adding an extra \textsc{spmv} to explicitly compute the residual in each iteration is not advisable, since it would increase the algorithm's computational cost.  The following property allows to compute the residual norm in each iteration of the p($l$)-CG algorithm without the explicit computation of the residual vector. 

\begin{theorem} \label{th:resnorm}
Let the search directions $P_{i-l+1}$ of the p($l$)-CG method be defined by $p_0 = v_0/\eta_0$ and the recurrence \eqref{eq:pdir}, i.e.
$p_k = \eta_k^{-1}(v_k-\beta_k p_{k-1})$, for $1 \leq k \leq i-l$. Let the solution $x_{i-l+1}$ be given by $x_{i-l+1} = x_{0}+P_{i-l+1} q_{i-l+1}$,
where $q_{i-l+1} = (\zeta_0,\ldots,\zeta_{i-l})^T$ is characterized by the entries $\zeta_0 = {\|r_0\|}_2$ and $\zeta_{k} = -\lambda_{k}\zeta_{k-1}$. 
Then it holds that $|\zeta_k| = {\|r_k\|}_{2}$ in any iteration $0 \leq k \leq i-l$.
\end{theorem}

\begin{proof}
The first part of the proof follows the classic argumentation of Saad in \cite{saad2003iterative}, p.160. The residual $r_k$ with $0\leq k \leq i-l$ is given by:
\begin{eqnarray*}
r_k &=& b-Ax_{k} = b-A(x_{0}+V_{k}y_{k}) = r_{0} - AV_{k}y_{k} \\
&=&  {\|r_0\|}_2 v_{0} - V_{k}T_{k}y_{k} - \delta_{k-1} (e_{k}^T y_{k}) v_{k}
= -\delta_{k-1} (e_{k}^{T}y_{k}) v_{k}, \qquad 0\leq k \leq i-l,
\end{eqnarray*}
where the last equality holds since ${\|r_0\|}_2 v_{0} - V_{k}T_{k}y_{k} = 0$, see \eqref{eq:lanczos}.
Since $v_{k}$ is normalized, i.e.\,${\|v_{k}\|}_2 = 1$, it holds that ${\|r_k\|}_2 = |{-}\delta_{k-1} (e_{k}^{T}y_{k})|$. 
From the definition $y_{k} = U_{k}^{-1} q_{k}$ it follows that the last element of $y_{k}$ is $\zeta_{k-1} / \eta_{k-1}$.
Hence ${\|r_k\|}_2 = |{-}\zeta_{k-1} \delta_{k-1}  / \eta_{k-1}| = |{-}\zeta_{k-1} \lambda_{k}| = |\zeta_{k}|$.
\end{proof}
The equality $|\zeta_k| = {\|r_k\|}_{2}$ holds in exact arithmetic, but 
rounding errors may contaminate 
$\zeta_{k}$ in a practical 
implementation in each iteration, leading to deviations from the actual residual norm. 
We expound on the numerical behavior of the p($l$)-CG method 
in finite precision in Section \ref{sec:analysis}.

\begin{figure}
\begin{center}
\includegraphics[width=0.47\textwidth]{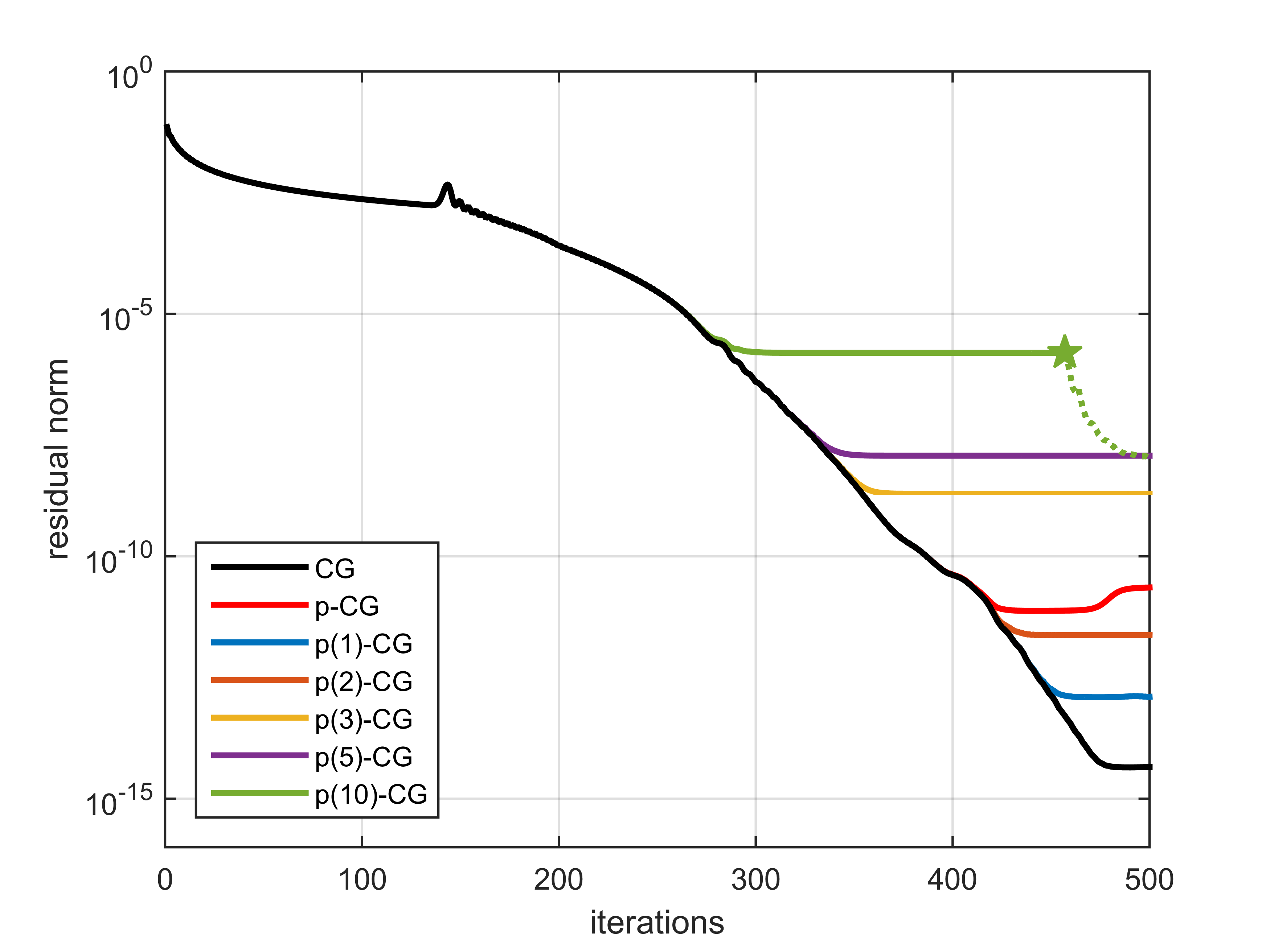} 
\includegraphics[width=0.47\textwidth]{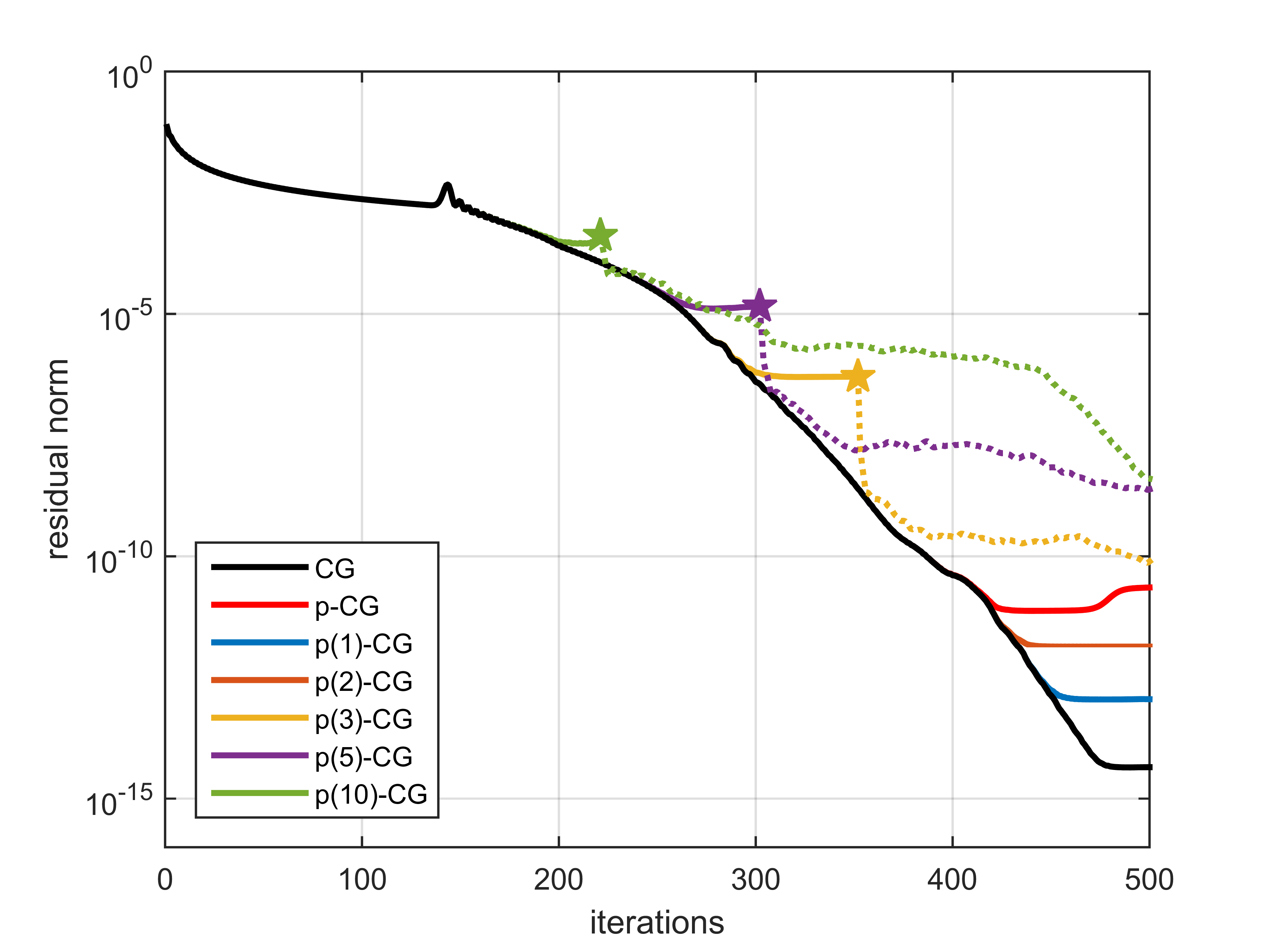}
\end{center}
\caption{Comparison of the residual norm history $\|b-Ax_j\|$ for different CG variants for a 2D Poisson problem with 200 $\times$ 200 unknowns. The stabilizing shifts $\sigma_i$ for p($\ell$)-CG are based on the degree $l$ Chebyshev polynomial on the interval [0,8] (left) (optimal shift choices) and the perturbed interval [0,8*1.005] (right) (slightly sub-optimal shifts). Square root breakdowns in p($\ell$)-CG are indicated by a $\bigstar$ symbol (followed by explicit iteration restart).}
\label{fig:residuals}
\end{figure}

\begin{remark} \textbf{Comparison to p-CG from Ghysels et al.\,\cite{ghysels2014hiding}.} \label{remark:remark9}
Introducing a recurrence relation for the residual and rewriting the corresponding expression 
to achieve an overlap between global communication and \textsc{spmv} computation would lead to the pipelining approach proposed in \cite{ghysels2014hiding} to derive the p-CG and p-CR methods (which are limited to a pipeline length $l = 1$). Although the p-CG and p($l$)-CG variants are both denoted as `pipelined CG methods', the approach to pipelining proposed in this work fundamentally differs from the procedure in \cite{ghysels2014hiding} and the resulting algorithms are quite different from a numerical perspective.
\end{remark}

\begin{remark} \textbf{Stopping criterion.} \label{remark:stopping}
The characterization of the residual norm ${\|r_{i-l}\|}_2 = |\zeta_{i-l}|$ allows to add a stopping criterion to Alg.\,\ref{algo:plCG}. Given a relative residual tolerance $\tau$ (input variable), the following classic stopping criterion can for example be added after line 31 in Alg.\,\ref{algo:plCG}:
\begin{equation*}
  \text{~~\textbf{\emph{if}}~~~} |\zeta_{i-l}|/{\|b\|}_2 \leq \tau \text{~~\textbf{\emph{then}}~ \emph{BREAK} ~\textbf{\emph{end if}}.}
\end{equation*}
Note that the location of the stopping criterion in the algorithm is important. Indeed, the above check could be performed immediately after $\zeta_{i-l}$ has been computed in Alg.\,\ref{algo:plCG} on line 29; however, the update of the solution $x_{i-l}$ that corresponds to the residual $|\zeta_{i-l}|$ is only computed on line 31.
\end{remark}

\begin{remark} \textbf{Solution update.} \label{remark:solution}
In light of Remark \ref{remark:stopping} and the discussion in Section \ref{sec:towards}, 
note that one could already compute the next solution $x_{i-l+1} = x_{i-l} + \zeta_{i-l} p_{i-l}$, 
see \eqref{eq:x_a}-\eqref{eq:x_a2}, on line 31 of Alg.\,\ref{algo:plCG}, since 
$p_{i-l}$ and 
$\zeta_{i-l}$ are both computed in iteration $i$. 
However, the corresponding residual norm $|\zeta_{i-l+1}|$ is then available only after executing line 29 in iteration $i+1$ (where the stopping criterion could be checked). 
To keep the solution and residual norm within a single iteration in sync 
and retain the analogy with Alg.\,\ref{algo:plGMRES} 
we opt to compute the solution $x_{i-l}$ in iteration $i$ in Alg.\,\ref{algo:plCG}. 
\end{remark}

\subsection{Preconditioned p($l$)-CG} \label{sec:preconditioned}

Since preconditioning is 
a crucial aspect for the efficient solution of large scale linear systems, we discuss the extension of the p($l$)-CG algorithm to include a preconditioner. The methodology follows the derivation of the preconditioned CG and p-CG methods outlined in e.g.~\cite{ghysels2014hiding}, aiming to iteratively solve the system $M^{-1}A x =M^{-1}b$ where both $A$ and $M$ are symmetric positive definite matrices. The approximate solutions $x_i$ $(i \geq 0)$ lie in the subspaces $x_0 + \mathcal{K}_i(M^{-1}A,r_0)$. However, the symmetry of $A$ and $M$ in general does not imply that the preconditioned system is symmetric. To preserve symmetry we use the observation that $M^{-1}A$ is self-adjoint with respect to the $M$ inner product $(x,y)_M = (Mx,y) = (x,My)$.

Let $V_{i-l+1} = [v_0,\ldots,v_{i-l}]$ again be the orthonormal basis for the $(i-l+1)$-th Krylov subspace $\mathcal{K}_{i-l+1}(M^{-1}A,r_0)$. 
Note that $r_i$ denotes the preconditioned residual $M^{-1}(b-Ax_i)$ in this section.
We define the auxiliary basis $Z_{i+1} = [z_0,z_1,\ldots,z_i]$ similarly to \eqref{eq:defz}:
\begin{equation} 
z_{j}:= \left\{ \begin{matrix} 
v_{0}, & j=0, \\ 
P_{j}(M^{-1}A)v_{0}, & 0<j\leq l, \\ 
P_{l}(M^{-1}A)v_{j-l}, & j>l, 
\end{matrix} \right. 
\qquad \text{with} \qquad P_{i}(t) := \prod_{j=0}^{i-1} (t-\sigma_{j}), \qquad  \text{for} ~ i\leq l,
\end{equation}
The recurrence relations \eqref{zid2} and \eqref{zid3} 
can in the preconditioned case be summarized as:
\begin{equation} \label{eq:prerecc}
z_{i+1} =\left\{ \begin{matrix} (M^{-1}A-\sigma_{i}I)z_{i} & i<l \\ ( M^{-1}Az_{i} - \gamma_{i-l} z_{i}- \delta_{i-l-1}z_{i-1})/\delta_{i-l} & i \geq l. \end{matrix}\right.
\end{equation}
In addition to the basis $Z_{i+1}$, the \emph{unpreconditioned} auxiliary basis vectors $\hat{Z}_{i+1} = [\hat{z}_0,\hat{z}_1,\ldots,\hat{z}_i]$ are defined as 
$\hat{z}_j = M z_j$
such that 
\begin{equation}
z_j = M^{-1} \hat{z}_j, \qquad j \leq i.
\end{equation}
The matrix $M$ (preconditioner inverse) is generally not explicitly available; however, for $i=0$ the first auxiliary vectors $\hat{z}_0$ and $z_0$ can be computed as 
$\hat{z}_{0} = M r_{0}/{\|r_{0}\|}_M  =  (b-Ax_{0})/{\|r_{0}\|}_M$ and $z_0 = M^{-1} \hat{z}_{0}$.
By multiplying both sides in the recurrence relations \eqref{eq:prerecc} for $z_{i+1}$ by $M$ one readily derives recurrence relations for the unpreconditioned basis vector $\hat{z}_{i+1}$:
\begin{equation} \label{eq:zhatid3}
\hat{z}_{i+1} 
= \left\{ \begin{matrix}Az_{i} -\sigma_{i}\hat{z}_{i}  & i<l \\ (Az_{i} - \gamma_a \hat{z}_{i}- \delta_{a-1} \hat{z}_{i-1})/\delta_{a} & i \geq l.  \end{matrix}\right. 
\end{equation}
The preconditioned auxiliary basis vector $z_{i+1}$ can be computed after the \textsc{spmv} $Az_i$ has been computed by applying the preconditioner to $Az_i$ and using expression \eqref{eq:prerecc}.
Given the basis vectors $V_{i-l+1}$, $\hat{Z}_{i+1}$ and $Z_{i+1}$, we replace the usual Euclidean dot product in Alg.\,\ref{algo:plCG} with the $M$ dot product. 
The dot products $g_{j,i+1}$ for $0 \leq j \leq i+1$ are then computed in analogy to \eqref{eq:gdotpr2} as follows: 
\begin{equation} \label{eq:gdotpr2prec}
g_{j,i+1}=\left\{ 
  \begin{matrix}
	  (z_{i+1},v_{j})_M = (\hat{z}_{i+1},v_{j}), & j=\max(0,i-2l+1),\ldots,i-l+1, \\ 
	  (z_{i+1},z_{j})_M = (\hat{z}_{i+1},v_{j}), & j=i-l+2,\ldots,i+1.
	\end{matrix}
\right. 
\end{equation}
With the above definitions, Lemma \ref{lemma:lemmaband} holds for the preconditioned pipelined CG method with the adapted definition $g_{j,i} = (z_i,v_j)_M = (\hat{z}_i,v_j)$.  
Consequently, the preconditioned p($l$)-CG algorithm is a direct extension of Alg.\,\ref{algo:plCG}, with reformulated dot products and the addition of the recurrence relation \eqref{eq:zhatid3} for the unpreconditioned auxiliary variable $\hat{z}_{i+1}$. 

Note that Theorem \ref{th:resnorm} still holds for the preconditioned system when the Euclidean 2-norm of $r_k$ in the formulation of the theorem is replaced by the $M$-norm of $r_k$. That is: for any iteration $k$ in preconditioned p($l$)-CG 
it holds that $|\zeta_k| = {\|r_k\|}_M = \sqrt{(b-Ax_k, M^{-1}(b-Ax_k))}$. The preconditioned algorithm thus intrinsically computes the $M$-norm of the residual in each iteration.

\begin{remark} \textbf{Newton basis shifts.}  \label{remark:shifts}
The preconditioned linear system also allows for the use of a shifted Newton-type basis for the polynomials $P_i(M^{-1}A)$ that are used to define $\hat{z}_{i+1}$ and $z_{i+1}$ as illustrated by \eqref{eq:prerecc} and \eqref{eq:zhatid3}. However, since the preconditioner application also overlaps with global communication, it may not be required to use deep pipelines in practice when the preconditioner application is sufficiently computationally expensive with respect to the global reduction phase. 
\end{remark}

\section{Implementation considerations} \label{sec:imple}

Section \ref{sec:deriv} gave an overview of the mathematical properties of the p($l$)-CG method, ultimately leading to Alg.\,\ref{algo:plCG}. 
In this section we comment on several important technical aspects concerning the implementation of the p($l$)-CG algorithm.

\subsection{Hiding communication in p($l$)-CG} \label{sec:hiding}

Alg.\,\ref{algo:plCG} gives the classic algebraic formulation of the p($l$)-CG method. However, it may not be directly apparent from this formulation where the overlap of global communication with computational work occurs throughout the algorithm. We therefore introduce a schematic kernel-based representation of the p($l$)-CG algorithm in this section. The following computational kernels are defined in iteration $i$ in Alg.\,\ref{algo:plCG}:\\

\begin{center}
  \begin{tabular}{| c | c | l | c | }
    \hline
    kernel \# & kernel type & kernel description & Alg.\,\ref{algo:plCG} lines  \\ \hline \hline
    (K1) & \textsc{spmv}   & apply $A$ and $M^{-1}$ to compute $z_{i+1}$ and $\hat{z}_{i+1}$      & 5 \\ \hline
    (K2) & \textsc{scalar} & update basis transformation matrix $G_{i-l+2}$ elements              & 7-8 \\ \hline
		(K3) & \textsc{scalar} & update Hessenberg matrix $H_{i-l+2,i-l+1}$ elements                  & 10-16 \\ \hline
		(K4) & \textsc{axpy}   & recursive update of $v_{i-l+1}$, $z_{i+1}$ and $\hat{z}_{i+1}$       & 17-18 \\ \hline
		(K5) & \textsc{dotpr}  & compute dot products $(\hat{z}_{i+1},z_j)$ and $(\hat{z}_{i+1},v_j)$ & 20 \\ \hline
		(K6) & \textsc{axpy}   & update solution $x_{i-l}$ and residual norm $|\zeta_{i-l}|$          & 22-32 \\ \hline
  \end{tabular}
\end{center}\ \\

\noindent The \textsc{spmv} kernel (K1) is considered to be the most computationally intensive part of the algorithm, and hence should be overlapped with the global reduction phase in (K5) to hide communication latency and idle core time. Kernels (K2), (K3), (K4) and (K6) represent purely local scalar and vector operations which are assumed to be executed very fast on multi-node hardware. These operations are also overlapped with the global reduction phase; however, due to their low arithmetic complexity the overlap is not expected to yield any major performance improvement. In (K5) all local contributions to the dot products are first computed by each worker. Subsequently a global reduction phase is performed in which local contributions are added pairwise via a $\log_2(N)$ length reduction tree, where $N$ represents the number of workers. Once the scalar result of each dot product has been collected on a single worker, a global broadcasting phase redistributes the resulting scalars back to all workers for local use in the next iteration. 
The preconditioned p($l$)-CG algorithm can be summarized schematically using these kernel definitions as displayed in Alg.\,\ref{algo:schema}. 

\begin{algorithm}[t]
{\small
\caption{Schematic representation of p($l$)-CG \hfill \textbf{Input:} $A$, $M^{-1}$, $b$, $x_0$, $l$, $m$}\label{algo:schema}
\begin{algorithmic}[1]
\State \textsc{initialization} ;
\For {$i=0,\ldots, m+l$}
\State (K1) \textsc{spmv} ;
\If {$i\geq l$}
\State \texttt{MPI\_Wait(req(i-l), \ldots)} ;
\State (K2) \textsc{scalar} ;
\State (K3) \textsc{scalar} ;
\State (K4) \textsc{axpy} ;
\EndIf
\State \textbf{end if}
\State (K5) \textsc{dotpr} ;
\State \texttt{MPI\_Iallreduce(\ldots, G(max(0,i-2l+1):i+1,i+1), \ldots, req(i))} ;
\State (K6) \textsc{axpy} ;
\EndFor 
\State \textbf{end for}
\end{algorithmic}
}
\end{algorithm}

Our implementation uses MPI with the MPI-3 standard as the communication library. The MPICH-3 library used in our experiments, see Section \ref{sec:experiments}, allows for asynchronous progress in the global reduction by setting the following environment variables:
\vspace{0.2cm}
\begin{itemize}  
\item[] \texttt{MPICH\_ASYNC\_PROGRESS=1};
\item[] \texttt{MPICH\_MAX\_THREAD\_SAFETY=multiple};
\vspace{0.2cm}
\end{itemize} 
Global communication is initiated by a call which starts a non-blocking reduction: 
\vspace{0.2cm} 
\begin{itemize}
\item[] \texttt{MPI\_Iallreduce(\ldots, G(i-2l+1:i+1,i+1), \ldots, req(i));}
\vspace{0.2cm}
\end{itemize}
The input argument \texttt{G(i-2l+1:i+1,i+1)} represents the $2l+1$ elements of the band structured matrix $G_{i+2}$ that are computed using the dot products in (K5) in iteration $i$, see \eqref{eq:gdotpr2prec}. The result of the corresponding global reduction phase is signaled to be due to arrive by the call to
\vspace{0.2cm}
\begin{itemize}
\item[] \texttt{MPI\_Wait(req(i), \ldots);}
\vspace{0.2cm}
\end{itemize}
The \texttt{MPI\_Request} array element \texttt{req(i)} that is passed as an input argument to \texttt{MPI\_Wait} keeps track of the iteration index in which the global reduction phase was initiated.
Since the p($l$)-CG method overlaps $l$ \textsc{spmv}'s with a single global reduction phase, the call to \texttt{MPI\_Wait(req(i), \ldots)} occurs effectively in iteration $i+l$, i.e.\,$l$ iterations after the call \texttt{MPI\_Iallreduce(\ldots, req(i))}.

\begin{table}[t]
\centering
\vspace{0.5cm}
\small
\begin{tabular}{| l | c | c | c | c | c |}
\hline 
							& \textsc{glred} & \textsc{spmv} &  \textsc{time} & \textsc{flops} & \textsc{memory} \\
							&  &  & (\textsc{glred} \& \textsc{spmv}) & (\textsc{axpy} \& \textsc{dotpr}) &  \\
\hline 
CG	 					& 2 & 1 		& 2 \textsc{glred} + 1 \textsc{spmv} & 10			& 3 \\
p-CG					& 1 & 1     & $\max$(\textsc{glred}, \textsc{spmv}) & 16 			& 6 \\
p($l$)-CG    	& 1 & 1     & $\max$(\textsc{glred}$/l$, \textsc{spmv}) & $6l+10$ & $3l+3$ 
\\
p($l$)-GMRES 	& 1 & 1     & $\max$(\textsc{glred}$/l$, \textsc{spmv}) & $6i-4l+8$ & $2i-l+4$ \\
\hline
\end{tabular}
\caption{Theoretical specifications comparing p($l$)-CG to related Krylov subspace methods. 
`CG' denotes classic CG; `p-CG' is the pipelined CG method from \cite{ghysels2014hiding};
`p($l$)-CG' (with $l \geq 1$) denotes Alg.\,\ref{algo:plCG};
`p($l$)-GMRES' is Alg.\,\ref{algo:plGMRES}.
\textsc{glred}: number of global all-reduce communication phases per iteration;
\textsc{spmv}: number of \textsc{spmv}s per iteration;
\textsc{time}: time spent per iteration in \textsc{glred}s and \textsc{spmv}s; 
\textsc{flops}: number of flops ($\times n$) per iteration for \textsc{axpy}s and dot products ($i$ = iteration index);
\textsc{memory}: total number of vectors in memory (excl.~$x_{i}$ and $b$) 
during execution.}
\label{tab:pipelcg}
\vspace{-0.5cm}
\end{table}
 
The schematic representation, Alg.\,\ref{algo:schema}, shows that the global reduction phase that is initiated by \texttt{MPI\_Iallreduce} with request \texttt{req(i)} in iteration $i$ overlaps with a total of $l$ \textsc{spmv}'s, namely the kernels (K1) in iterations $i+1$ up to $i+l$. The corresponding call to \texttt{MPI\_Wait} with request \texttt{req((i+l)-l)} = \texttt{req(i)} takes place in iteration $i+l$ before the computations of (K2) in which the dot product results are required, but after the \textsc{spmv} kernel (K1) has been executed. In each iteration the global reduction also overlaps with a number of less computationally intensive operations from (K2), (K3), (K4) and (K6). Hence, 
the global communication latency of the dot products in (K5) is `hidden' behind the computational work of $l$ p($l$)-CG iterations.

Table \ref{tab:pipelcg} summarizes key theoretical properties of the p($l$)-CG method in comparison to related algorithms (incl.~memory requirements -- see Section \ref{sec:storing}). The \textsc{flops} count reported for p($l$)-CG assumes that the symmetry of $G_{i+1}$, 
see Lemma \ref{lemma:lemmaband}, 
is exploited to reduce the number of dot products that are computed in Alg.\,\ref{algo:plCG}, line 20. The 
$l$ elements $g_{i-2l+1,i+1},\ldots,g_{i-l,i+1}$ have already been computed in previous iterations, 
see expression \eqref{eq:symmetry_G}. 
Note that the results for the p($l$)-GMRES algorithm exclude the computational and storage cost to execute line 22-23 in Alg.\,\ref{algo:plGMRES}.

\subsection{Storing the $V_k$ and $Z_k$ basis} \label{sec:storing}

A clear advantage of the pipelined p($l$)-CG method, Alg.\,\ref{algo:plCG}, in comparison with p($l$)-GMRES, Alg.\,\ref{algo:plGMRES}, is its reduced storage requirements. In p($l$)-GMRES the complete bases $V_{i-l+1}$ and $Z_{i+1}$ need to be built and stored during the entire run of the algorithm, 
since all basis vectors are needed in the recursions for the next basis vectors (see Alg.\,\ref{algo:plGMRES} lines 17, 18 and 20). 
In contrast, the symmetry of the matrix $A$ induces the symmetry 
of the matrix $T_k$, see Corollary \ref{corollary}, which in turn induces a band structure for $G_k$ as shown in Theorem \ref{lemma:lemmaband}.
In the $i$-th iteration of Alg.\,\ref{algo:plCG} 
the new vector $z_{i+1}$ is computed by an \textsc{spmv} with $z_i$ (line 5); subsequently $z_{i-l+1}$ and $v_{i-3l+1}$, $v_{i-3l+2}$, \ldots, $v_{i-l}$ are required to compute the next basis vector $v_{i-l+1}$, which itself also needs to be stored (line 17); next the last three auxiliary vectors $z_{i-1}$, $z_i$ and $z_{i+1}$ are used in the recurrence for $z_{i+1}$ (line 18); and finally the vectors $v_{i-l+1}$ and $z_{i-l+2}$, \ldots, $z_{i+1}$ are needed to compute the dot products (line 20). This implies that the $3l+2$ basis vectors 
$\{v_{i-3l+1},\ldots,v_{i-l+1}\} \in V_{i-l+2}$ and 
$\{z_{i-l+1},\ldots,z_{i+1}\} \in Z_{i+2}$
are required in iteration $i$. 
From iteration $i$ onward basis vectors $v_j$ with indices $j < i-3l+1$, and vectors $z_j$ with $j < i-l+1$ are not used in either the recursive vector updates or the dot products in Alg.\,\ref{algo:plCG}, and should thus no longer be stored. Hence, no more than $3l+2$ basis vectors need to be kept in memory in \emph{any} iteration of the p($l$)-CG algorithm. 
Fig.\,\ref{fig:basis_storage} schematically shows the storage requirements in iteration $i$ of Alg.\,\ref{algo:plCG}. In each iteration an auxiliary vector $z_{i+1}$ is added and a new sequence of dot products is initiated. The results of the dot product calculations arrive $l$ iterations later, see Section \ref{sec:hiding}. Dot products that were initiated $l$ iterations ago are then used to append a basis vector $v_{i-l+1}$. 

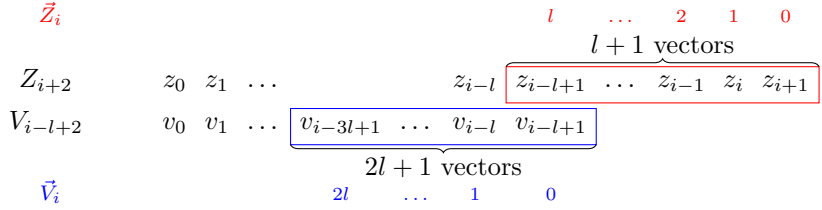
\begin{figure} \label{fig:basis_storage}
  \begin{center}
\begin{tikzpicture}[ node distance=1mm and 0mm, baseline]
\matrix (M1) [matrix of nodes]
 {
  \footnotesize{\color{red}$\vec{Z}_{i}$} ~~~~~~ & & & & & & & \scriptsize{\color{red}$l$} & \scriptsize{\color{red}\ldots} & \scriptsize{\color{red}2} & \scriptsize{\color{red}1} & \scriptsize{\color{red}0} \\[10pt]
  $Z_{i+2}$   ~~~~~~ & $z_0$ & $z_1$ & \ldots & & & $z_{i-l}$ & $z_{i-l+1}$  & \ldots & $z_{i-1}$ & $z_i$ & $z_{i+1}$ \\
  $V_{i-l+2}$ ~~~~~~ & $v_0$ & $v_1$ & \ldots & $v_{i-3l+1}$ & \ldots & $v_{i-l}$  &$v_{i-l+1}$ & & & & \\[10pt]
  \footnotesize{\color{blue}$\vec{V}_{i}$} ~~~~~~ & & & & \scriptsize{\color{blue}$2l$} & \scriptsize{\color{blue}\ldots} &  \scriptsize{\color{blue}1} & \scriptsize{\color{blue}0} & & & & \\
};
\draw[red,thin] 
(M1-2-8.north west) -| (M1-2-12.south east) -| (M1-2-8.north west);
\draw[blue,thin] 
(M1-3-5.north west) -| (M1-3-8.south east) -| (M1-3-5.north west);
\draw[decorate,decoration={brace,amplitude=3pt,mirror}](M1-3-5.south west)--node[below=1pt]{$2l+1$ vectors}(M1-3-8.south east);
\draw[decorate,decoration={brace,amplitude=3pt}](M1-2-8.north west)--node[above=1pt]{$l+1$ vectors}(M1-2-12.north east);
\end{tikzpicture}
\end{center}
\vspace{-0.5cm}
\caption{The sequence of auxiliary vectors $Z_{i+2}$ and basis vectors $V_{i-l+2}$ constructed in the $i$-th iteration of p($l$)-CG, Alg.\,\ref{algo:plCG}. Colored boxes represent vectors that are actively used in iteration $i$ and thus should be kept in memory up to this point, see Section \ref{sec:storing}. They are collected in the sliding windows $\vec{Z}_{i}$ and $\vec{V}_{i}$, see Appendix \ref{sec:sliding}. The super- and sub-indices indicate the position of the corresponding vectors in the sliding windows.}
\end{figure}

\begin{remark} \textbf{Storage in preconditioned p($l$)-CG.}    \label{remark:prec_storage}
A similar analysis of basis storage requirements can be performed for the preconditioned version of p($l$)-CG. 
In addition to the $3l+2$ vectors from the bases $V_{i-j+2}$ and $Z_{i+2}$ pointed out above, only the last three vectors $\hat{z}_{i-1}$, $\hat{z}_{i}$ and $\hat{z}_{i+1}$ in the auxiliary basis $\hat{Z}_i$ are required in the recursive update for $\hat{z}_{i+1}$. Vectors $\hat{z}_j$ with $j < i-1$ 
are not used in current or future iterations of the algorithm from iteration $i$ onward. 
Hence, the preconditioned algorithm stores a maximum of $3l+5$ basis vectors at any point during the algorithm, and thus a total of $3l+6$ vectors are kept in \textsc{memory}, cf.~Table \ref{tab:pipelcg}.
\end{remark}

In Appendix \ref{sec:sliding} we comment on an efficient way to implement the storage of the basis vectors throughout the algorithm by using the concept of \emph{`sliding windows'}.

\section{Analysis of the attainable accuracy of p($l$)-CG in finite precision arithmetic} \label{sec:analysis}

As suggested by Fig.\,\ref{fig:residuals}, replacing the classic CG algorithm by the pipelined p($l$)-CG variant introduces numerical issues. The numerical accuracy attainable by the p($\ell$)-CG method may be reduced drastically for larger pipeline lengths $l$. In exact arithmetic (and provided no square root breakdowns occur in p($l$)-CG, see Remark \ref{remark:sqrt_breakdown}), p($l$)-CG produces a series of iterates identical to the classic CG method. However, in finite precision arithmetic their behavior can differ significantly as local rounding errors may induce a decrease in attainable accuracy and a delay of convergence. The impact of finite precision round-off errors on the numerical stability of classic CG has been extensively studied 
\cite{greenbaum1989behavior,greenbaum1992predicting,greenbaum1997estimating,gutknecht2000accuracy,strakovs2002error,strakovs2005error,meurant2006lanczos,gergelits2014composite}. In communication reducing CG variants the effects of local rounding errors are significantly amplified; we refer to our manuscript \cite{cools2018analyzing} and the work by Carson et al.~\cite{carson2016numerical} for an 
overview of the stability analysis of the depth one pipelined Conjugate Gradient method 
from \cite{ghysels2014hiding}. 
In this section we analyze the behavior of local rounding errors that stem from the multi-term recurrence relations in the p($l$)-CG algorithm in a finite precision framework.
Preconditioning is omitted in this section for simplicity of notation but without loss of generality.\footnote{The extension of the local rounding error analysis 
to the preconditioned p($\ell$)-CG algorithm is trivial since the recurrences for the unpreconditioned variables are decoupled from their preconditioned counterparts. We refer the reader to Section \ref{sec:preconditioned} and our related work in \cite{cools2018analyzing} for more details.}

In this section we use a notation with bars to indicate variables that are 
computed in a finite precision setting.  Furthermore, for variables that are defined recursively in the algorithm, we differentiate between the recursively computed variable and the `actual' variable, i.e.~the variable that would be produced by exact computation using already computed inaccurate quantities. The latter is denoted by a bold typesetting. E.g.~the recursively computed residual in finite precision is denoted as $\bar{r}_j$, whereas the actual residual (which \emph{could} be computed, but typically isn't to reduce computational cost) is $\bar{\bold{r}}_j = b-A\bar{x}_j$. The primary aim of this section is to analyze the gap between the recursively computed variable $\bar{r}_j$ and its local 
`recurrence error-free' counterpart $\bar{\bold{r}}_j$.

We use the classic model for floating point arithmetic with machine precision $\epsilon$. The round-off error on scalar multiplication, vector summation, \textsc{spmv} application and dot product computation on an $n$-by-$n$ matrix $A$, length $n$ vectors $v$, $w$ and a scalar number $\alpha$ are respectively bounded by
\begin{align*}
	\| \alpha v - \text{fl}(\alpha v) \| &\leq \| \alpha v \| \, \epsilon =  |\alpha| \, \|v\| \, \epsilon, &\qquad
	\| v + w - \text{fl}(v + w) \| &\leq (\|v\| + \|w\|) \, \epsilon,\\
	\| Av - \text{fl}(Av) \| &\leq \mu\sqrt{n} \, \|A\| \, \|v\| \, \epsilon, &\qquad
	| \left( v,w \right) - \text{fl}(\,\left(v,w \right)\,) | &\leq n \, \|v\| \, \|w\| \epsilon,
\end{align*}
where $\text{fl}(\cdot)$ indicates the finite precision floating point representation, $\mu$ is the maximum number of nonzeros in any row of $A$, and the norm $\|\cdot\|$ represents the Euclidean 2-norm in this section.

\subsection{Local rounding error behavior in finite precision p($l$)-CG}

We give a very summarily overview of the behavior of local rounding errors in classic CG, 
see e.g.~\cite{greenbaum1992predicting,strakovs2002error,strakovs2005error,meurant2006lanczos}.
The recurrence relations for the approximate solution $\bar{x}_{j+1}$ and the residual $\bar{r}_{j+1}$ 
computed by the \emph{classic} CG algorithm in the finite precision framework are
\begin{equation} \label{eq:recs_xandr}
	\bar{x}_{j+1} = \bar{x}_j + \bar{\alpha}_j \bar{p}_j + \xi_{j+1}^{\bar{x}} , \qquad 
	\bar{r}_{j+1} = \bar{r}_j - \bar{\alpha}_j A\bar{p}_j + \xi_{j+1}^{\bar{r}} , \qquad
\end{equation}
where $\xi_{j+1}^{\bar{x}}$ and $\xi_{j+1}^{\bar{r}}$ 
represent local rounding errors, and where $\bar{s}_j = A \bar{p}_j$. We refer to the analysis in \cite{carson2016numerical,cools2018analyzing} for bounds on the norms of these local rounding errors.
It follows directly from \eqref{eq:recs_xandr} that in classic CG the residual gap is
\begin{equation} \label{eq:f_CG}
  \bar{\bold{r}}_{j+1} - \bar{r}_{j+1} = (b-A\bar{x}_{j+1}) - \bar{r}_{j+1} =  \bar{\bold{r}}_{j} - \bar{r}_j - A\xi_{j+1}^{\bar{x}} - \xi_{j+1}^{\bar{r}} = \bar{\bold{r}}_{0} - \bar{r}_{0} - \sum_{k=0}^{j} ( A  \xi_{k+1}^{\bar{x}} + \xi_{k+1}^{\bar{r}} ).
\end{equation}
In each iteration local rounding errors of the form $A\xi_{j+1}^{\bar{x}} + \xi_{j+1}^{\bar{r}}$ add to the gap on the residual.
Thus, local rounding errors are merely accumulated in the classic CG algorithm, and no amplification of rounding errors occurs. To avoid confusion we stress that \eqref{eq:recs_xandr}-\eqref{eq:f_CG} apply to classic CG only.

We now turn towards analyzing the p($l$)-CG method in a finite precision framework.
Consider the faulty variant of recurrence relation \eqref{eq:x_a2} for $\bar{x}_j$ in p($l$)-CG, Alg.\,\ref{algo:plCG}, in finite precision:
\begin{equation} \label{eq:exp_x}
	\bar{x}_j = \bar{x}_{j-1} + \bar{\zeta}_{j-1} \bar{p}_{j-1} + \xi^{\bar{x}}_j 
						= \bar{x}_0 + \bar{P}_{j} \bar{q}_j + \Theta_j^{\bar{x}} \, \boldsymbol{1},
\end{equation}
where $\bar{q}_{j} = (\bar{\zeta}_0,\ldots,\bar{\zeta}_{j-1})^T$ is characterized by the entries $\bar{\zeta}_j$ which are computed explicitly in Alg.\,\ref{algo:plCG}, 
$\Theta_j^{\bar{x}} = [\xi^{\bar{x}}_1,\xi^{\bar{x}}_1,\ldots,\xi^{\bar{x}}_{j}]$ with $\|\xi^{\bar{x}}_j\| \leq (\|\bar{x}_{j-1}\| + 2 \, |\bar{\zeta}_{j-1}| \, \|\bar{p}_{j-1}\|) \, \epsilon$ are local rounding errors, and $\boldsymbol{1} = [1,1,\ldots,1]^T$. 
Similarly, the finite precision recurrence relation for $\bar{p}_j$ in Alg.\,\ref{algo:plCG} is
\begin{equation} \label{eq:exp_p}
	\bar{p}_j = (\bar{v}_{j} - \bar{\delta}_{j-1} \bar{p}_{j-1})/\bar{\eta}_j + \xi^{\bar{p}}_j \qquad \Leftrightarrow \qquad \bar{V}_{j+1} = \bar{P}_{j+1} \bar{U}_{j+1} + \Theta^{\bar{p}}_{j+1},
\end{equation} 
where $\bar{U}_{j}$ is the upper triangular factor of $\bar{T}_{j,j} = \bar{L}_{j} \bar{U}_{j}$ and $ \Theta^{\bar{p}}_{j+1} = -[\bar{\eta}_0\xi^{\bar{p}}_0,\bar{\eta}_1\xi^{\bar{p}}_1,\ldots,\bar{\eta}_j\xi^{\bar{p}}_j]$ with $\|\xi^{\bar{p}}_j\| \leq (2 / \bar{\eta}_j \, \|\bar{v}_{j-1}\|  + 3 \, |\bar{\delta}_{j-1}|/ \bar{\eta}_j  \, \|\bar{p}_{j-1}\|) \, \epsilon$ are local rounding errors.

Substitution of expression \eqref{eq:exp_p}, i.e.~$\bar{P}_j = (\bar{V}_j - \Theta_j^{\bar{p}}) \, \bar{U}^{-1}_{j}$, 
 into equation \eqref{eq:exp_x} yields
\begin{equation} \label{eq:exp_x2}
	\bar{x}_j 
						= \bar{x}_0 + \bar{V}_j \, \bar{U}^{-1}_{j} \bar{q}_j - \Theta_j^{\bar{p}} \, \bar{U}^{-1}_{j} \bar{q}_j + \Theta_j^{\bar{x}} \, \boldsymbol{1}.
\end{equation}
Consequently, the actual residual $\bar{\bold{r}}_j$ can be written as
\begin{align} \label{eq:extra_res3}
	\bar{\bold{r}}_j = b-A\bar{x}_j =  \bar{\bold{r}}_0 - A \bar{V}_j \bar{U}^{-1}_{j} \bar{q}_j + A \Theta_j^{\bar{p}} \bar{U}^{-1}_{j} \bar{q}_j - A \Theta_j^{\bar{x}} \, \boldsymbol{1}.
\end{align}
In this expression the computed basis vectors $\bar{v}_{j}$ are calculated from the finite precision variant of the recurrence relation \eqref{eq:recur_v}, i.e.
\begin{equation} \label{eq:vbar_rec}
  \bar{v}_{j+1} = \left( \bar{z}_{j+1} - \sum_{k=j-2l+1}^{j} \bar{g}_{k,j+1} \bar{v}_{k} \right) / \bar{g}_{j+1,j+1} + \xi^{\bar{v}}_{j+1}, \qquad 0 \leq j < i-l,
\end{equation}
where the size of the local rounding errors $\xi^{\bar{v}}_{j+1}$ can be bounded in terms of the machine precision $\epsilon$ as $  \| \xi^{\bar{v}}_{j+1} \| \leq ( 2 \, \|\bar{z}_{j+1}\|/|\bar{g}_{j+1,j+1}| + 3 \sum_{k=j-2l+1}^{j} |\bar{g}_{k,j+1}|/|\bar{g}_{j+1,j+1}| \, \|\bar{v}_k\| ) \epsilon$.
On the other hand, for any $j \geq 0$ the actual basis vector $\bar{\bold{v}}_{j+1}$ satisfies the Lanczos relation exactly, that is, it is defined as
\begin{equation} \label{eq:v_arnoldi}
  \bar{\bold{v}}_{j+1}= ( A\bar{v}_{j} - \bar{\gamma}_j \bar{v}_j - \bar{\delta}_{j-1} \bar{v}_{j-1} ) / \bar{\delta}_j, \qquad 0 \leq j < i-l.
\end{equation}
For $j = 0$ it is assumed that $\bar{v}_{-1} = 0$.
By subtracting the computed basis vector $\bar{v}_{j+1}$ from both sides of the equation \eqref{eq:v_arnoldi}, it is easy to see that this relation alternatively translates to
\begin{equation*}
  A\bar{v}_{j} = \bar{\delta}_{j-1} \bar{v}_{j-1} + \bar{\gamma}_j \bar{v}_j +  \bar{\delta}_j \bar{v}_{j+1} + \bar{\delta}_j ( \bar{\bold{v}}_{j+1} - \bar{v}_{j+1} ) , \qquad 0 \leq j < i-l
\end{equation*}
or written in matrix notation:
\begin{equation} \label{eq:AV_BAR}
  A \bar{V}_j = \bar{V}_{j+1} \bar{T}_{j+1,j} + (\bar{\bold{V}}_{j+1}-\bar{V}_{j+1}) \bar{\Delta}_{j+1,j}, \qquad 1 \leq j \leq i-l,
\end{equation}
where $\bar{\Delta}_{j+1,j}$ is
\begin{equation*}
\bar{\Delta}_{j+1,j} =
\left(\begin{array}{cccc} 
0&&& \\
\bar{\delta}_0&0&& \\
&\bar{\delta}_1&0& \\
&&\ddots&0 \\
&&&\bar{\delta}_{j-1}
\end{array}\right).
\end{equation*}
We call $\bar{\bold{V}}_{j+1}-\bar{V}_{j+1}$ the `gaps' on the computed basis vectors, in analogy to the residual gaps.
Expression \eqref{eq:AV_BAR} enables to further work out the expression \eqref{eq:extra_res3} for the actual residual as follows:
\begin{align} \label{eq:extra_res5}
	\bar{\bold{r}}_j &= \bar{\bold{r}}_0 - \bar{V}_{j+1} \bar{T}_{j+1,j} \bar{U}^{-1}_{j} \bar{q}_j + (\bar{\bold{V}}_{j+1}-\bar{V}_{j+1}) \bar{\Delta}_{j+1,j} \bar{U}^{-1}_{j} \bar{q}_j + A \Theta_j^{\bar{p}} \bar{U}^{-1}_{j} \bar{q}_j - A \Theta_j^{\bar{x}} \, \boldsymbol{1} \notag \\
									 &= \bar{r}_j + (\bar{\bold{r}}_0 - \bar{r}_0) - (\bar{\bold{V}}_{j+1}-\bar{V}_{j+1}) \bar{\Delta}_{j+1,j} \bar{U}^{-1}_{j} \bar{q}_j + A \Theta_j^{\bar{p}} \bar{U}^{-1}_{j} \bar{q}_j - A \Theta_j^{\bar{x}} \, \boldsymbol{1} 
\end{align}
where the implicitly computed residual $\bar{r}_j$ corresponds to the residual norm $\|\bar{r}_j\| = |\bar{\zeta}_j|$, see Theorem \ref{th:resnorm}; it is defined as
\begin{equation} \label{eq:extra_res6}
	\bar{r}_j 
	= \bar{r}_0 - \bar{V}_{j+1} \bar{T}_{j+1,j} \bar{U}_{j}^{-1} \bar{q}_{j} 
	= -\bar{\delta}_{j-1} (e_j^T \bar{U}_{j}^{-1} \bar{q}_j) \, \bar{v}_j 
	=  \bar{\zeta}_j \, \bar{v}_j.  
\end{equation}
Expression \eqref{eq:extra_res5} indicates that the gap between $\bar{\bold{r}}_j$ and $\bar{r}_j$ critically depends on the basis vector gaps $\bar{\bold{V}}_{j+1} - \bar{V}_{j+1}$. We therefore proceed by analyzing the gap on the basis vectors in p($l$)-CG.

By setting $\Theta_j^{\bar{v}} = [\theta^{\bar{v}}_0,\theta^{\bar{v}}_1,\ldots,\theta^{\bar{v}}_{j-1}] := [0, \, \bar{g}_{1,1}\xi^{\bar{v}}_1, \, \ldots, \, \bar{g}_{j-1,j-1}\xi^{\bar{v}}_{j-1}]$, we obtain from \eqref{eq:vbar_rec} the matrix expression
\begin{equation} \label{eq:Z_BAR}
  \bar{Z}_j = \bar{V}_j \bar{G}_j + \Theta^{\bar{v}}_j, \qquad 1 \leq j \leq i-l.
\end{equation}
The computed auxiliary basis vector $\bar{z}_{j+1}$ satisfies a finite precision version of the recurrence relation \eqref{zid3}, which pours down to
\begin{equation} \label{eq:zbar_rec}
  \bar{z}_{j+1} =
	\left\{ 
		\begin{matrix}
			(A-\sigma_jI) \, \bar{z}_j + \xi^{\bar{z}}_{j+1}, & 0 \leq j < l \\ 
			( A\bar{z}_{j} - \bar{\gamma}_{j-l} \bar{z}_j - \bar{\delta}_{j-l-1} \bar{z}_{j-1} ) / \bar{\delta}_{j-l} + \xi^{\bar{z}}_{j+1}, & l \leq j < i,
		\end{matrix}
	\right. 
\end{equation}
where the local rounding errors $\xi^{\bar{z}}_{j+1}$ are bounded by
\begin{equation*}
\| \xi^{\bar{z}}_{j+1} \| \leq
	\left\{ 
		\begin{matrix}
			\mu \sqrt{n} \, \|A-\sigma_j I\| \, \|\bar{z}_j\| \, \epsilon, & 0 \leq j < l , \\
			 \left( (\mu \sqrt{n}+2) \frac{\|A\|}{|\bar{\delta}_{j-l}|} \|\bar{z}_j\| + 3 \frac{|\bar{\gamma}_{j-l}|}{|\bar{\delta}_{j-l}|} \|\bar{z}_j\| + 3 \frac{|\bar{\delta}_{j-l-1}|}{|\bar{\delta}_{j-l}|} \|\bar{z}_{j-1}\| \right) \epsilon, & l \leq j < i.
		\end{matrix}
	\right. 		
\end{equation*}
Here $n$ is the number of rows/columns in the matrix $A$ and $\mu$ is the maximum number of non-zeros over all rows of $A$. 
Expression \eqref{eq:zbar_rec} can be summarized in matrix notation as:
\begin{equation} \label{eq:AZ_BAR}
  A \bar{Z}_j = \bar{Z}_{j+1} \bar{B}_{j+1,j} + \Theta^{\bar{z}}_j, \qquad 1 \leq j \leq i,
\end{equation}
where $\Theta^{\bar{z}}_j = [\theta^{\bar{z}}_0,\theta^{\bar{z}}_1,\ldots,\theta^{\bar{z}}_{j-1}]$, with $\theta^{\bar{z}}_{k} = \xi^{\bar{z}}_{k+1}$ for $0 \leq k < l$ and $\theta^{\bar{z}}_{k} = \bar{\delta}_{k-l}\xi^{\bar{z}}_{k+1}$ for $l \leq k < i$.
Furthermore, the recursive definitions of the scalar coefficients $\bar{\gamma}_j$ and $\bar{\delta}_j$ in Alg.\,\ref{algo:plCG} imply that in iteration $i$ the following matrix relations hold:
\begin{equation} \label{eq:rec_coeff}
  \bar{G}_{j+1}  \bar{B}_{j+1,j} = \bar{T}_{j+1,j} \bar{G}_j, \qquad 1 \leq j \leq i-l.
\end{equation}
The gap $\bar{\bold{V}}_{j+1} - \bar{V}_{j+1}$ can now be computed by combining the above expressions. It holds that\footnote{Note that the Moore-Penrose (left) pseudo-inverse $\bar{\Delta}_{j+1,j}^{+} = (\bar{\Delta}_{j+1,j}^* \bar{\Delta}_{j+1,j})^{-1} \bar{\Delta}_{j+1,j}^*$ of the lower diagonal matrix $\bar{\Delta}_{j+1,j}$ in expression \eqref{eq:gap_plcg2} is an upper diagonal matrix, where $\bar{\Delta}_{j+1,j}^*$ is its Hermitian transpose.}
\begin{align} \label{eq:gap_plcg2}
\bar{\bold{V}}_{j+1} - \bar{V}_{j+1} 
				&\stackrel{\eqref{eq:AV_BAR}}{=}  (A \bar{V}_j - \bar{V}_{j+1} \bar{T}_{j+1,j}) \bar{\Delta}^{+}_{j+1,j} \notag \\
				&\stackrel{\eqref{eq:Z_BAR}}{=} (A \bar{Z}_j \bar{G}^{-1}_j - \bar{Z}_{j+1} \bar{G}^{-1}_{j+1} \bar{T}_{j+1,j} - A \Theta^{\bar{v}}_j \bar{G}^{-1}_j  + \Theta^{\bar{v}}_{j+1} \bar{G}^{-1}_{j+1} \bar{T}_{j+1,j}) \bar{\Delta}^{+}_{j+1,j} \notag \\
				&\stackrel{\eqref{eq:rec_coeff}}{=} (A \bar{Z}_j \bar{G}^{-1}_j - \bar{Z}_{j+1} \bar{B}_{j+1,j} \bar{G}^{-1}_j - A \Theta^{\bar{v}}_j \bar{G}^{-1}_j  + \Theta^{\bar{v}}_{j+1} \bar{B}_{j+1,j} \bar{G}^{-1}_j) \bar{\Delta}^{+}_{j+1,j} \notag \\
				&\stackrel{\eqref{eq:AZ_BAR}}{=} (\Theta^{\bar{z}}_j  - A \Theta^{\bar{v}}_j  + \Theta^{\bar{v}}_{j+1}  \bar{B}_{j+1,j}) \, \bar{G}^{-1}_j \, \bar{\Delta}^{+}_{j+1,j}.
\end{align}
Consequently, it is clear that in p($\ell$)-CG the local rounding errors in $\Theta^{\bar{z}}_j$, $A \Theta^{\bar{v}}_j$ and $\Theta^{\bar{v}}_{j+1} \bar{B}_{j+1,j}$ are possibly amplified by the entries of the matrix $\bar{G}^{-1}_j\bar{\Delta}^{+}_{j+1,j}$. 
This in turn 
implies an amplification of local rounding errors in expression \eqref{eq:extra_res5}, leading to reduced maximal attainable accuracy.
Since $\bar{\Delta}^{+}_{j+1,j}$ is a diagonal matrix, the propagation of local rounding errors in p($l$)-CG is primarily governed by the inverse of the (finite precision variant of the) basis transformation matrix $\bar{G}_j$.

\subsection{Bounding the basis vector gaps in finite precision p($l$)-CG} \label{sec:bounding}

As the matrix $\bar{G}^{-1}_j$ fulfills a crucial role in the propagation of local rounding errors in p($\ell$)-CG, see \eqref{eq:gap_plcg2}, we aim to establish 
bounds on the maximum norm of $\bar{G}^{-1}_j$ in this section. Consider the norm 
\begin{equation}
	\| \bar{G}^{-1}_j \|_{\max} = \max_{k,l} | \bar{G}^{-1}_j(k,l) |, 
\end{equation}
which characterizes the propagation of local rounding errors in the $\bar{V}_j$ basis in p($l$)-CG. When the maximum norm is larger than one local rounding errors may be amplified, see expression \eqref{eq:gap_plcg2}.

The inverse of the banded matrix $\bar{G}_j$ is an upper triangular matrix, which can be expressed as
\begin{equation} \label{eq:expan}
	\bar{G}^{-1}_j = \sum_{k=0}^{j-1} \left( - \bar{\Lambda}_j^{-1} \bar{G}_j^{\triangle} \right)^{k} \bar{\Lambda}_j^{-1},
\end{equation}
where $\bar{\Lambda}_j := [\delta_{mk} \bar{g}_{m,k}]$ contains the diagonal of $\bar{G}_j$ and $\bar{G}_j^{\triangle} := \bar{G}_j - \bar{\Lambda}_j$ is the strictly upper triangular part of $\bar{G}_j$.
\begin{lemma} \label{lemma:lemma_0}
	Assume $1 \leq j \leq i-l+1$ such that the basis $\bar{V}_j$ is orthonormal. Let the Krylov subspace basis transformation matrix $\bar{G}_j$ be defined by $ \bar{G}_j = \bar{V}^T_j \bar{Z}_j + \bar{V}^T_j \Theta_j^{\bar{v}}$ for $1 \leq j \leq i-l+1$ as in \eqref{eq:Z_BAR}. Then it holds that
	\begin{equation} \label{eq:est}
		\|\bar{G}_j\|_{\max} \leq  \|P_l(A)\| + \|\xi_k^{\bar{z}}\| + \|\theta_k^{\bar{v}}\|. 
	\end{equation}
\end{lemma}
\begin{proof}
	Since for any $0 \leq m \leq i-l$ the vector $\bar{v}_m$ is normalized, i.e.~$\|\bar{v}_m\| = 1$, the following bound on the entries of $\bar{G}_j$ holds:
	\begin{align} \label{eq:est2}
	  |\bar{g}_{m,k}| &= |(\bar{z}_k,\bar{v}_m) - (\theta_k^{\bar{v}},\bar{v}_m)| 
				\leq \|\bar{z}_k\| + \|\theta_k^{\bar{v}}\| 
				\leq \|P_l(A) \bar{v}_{k-l} + \xi_k^{\bar{z}}\| + \|\theta_k^{\bar{v}}\| \notag \\
				&\leq \|P_l(A)\| + \|\xi_k^{\bar{z}}\| + \|\theta_k^{\bar{v}}\|, 
				\qquad \qquad \qquad \qquad \qquad \qquad 0 \leq m \leq i-l, \quad 0 \leq k \leq i.
	\end{align}
\end{proof}

We remark that the matrix $\bar{G}_j$ is not necessarily diagonally dominant, since the norms $\|\bar{z}_k\|$ are not guaranteed to be monotonically decreasing in p($\ell$)-CG. Consequently, if the bound in \eqref{eq:est} is tight, Lemma \ref{lemma:lemma_0} suggests that when $\|P_l(A)\|$ is large, $\|\bar{G}^{-1}_j\|_{\max}$ may be (much) larger than one. This observation leads to some interesting 
insights. First, note that the norm $\|\bar{G}_j\|_{\max}$ increases monotonically with respect to the iteration index $j$, since each iteration adds a new column to $\bar{G}_j$. Furthermore, an increase in $j$ also implies an increase in the number of terms in the summation in \eqref{eq:expan}. The norm $\|\bar{G}^{-1}_j\|_{\max}$ thus increases as a function of $j$. Secondly, an increasing pipeline length $l$ increases the number of non-zero diagonals in $\bar{G}_j$ significantly, see Lemma \ref{lemma:lemmaband} and Appendix \ref{sec:lemma}, which impacts the norm of the right-hand side expression in \eqref{eq:expan}. The pipeline length $l$ also relates directly to the norm $\|P_l(A)\|$ which bounds the norm $\|\bar{G}_j\|_{\max}$ in \eqref{eq:est}. Therefore, it is expected that $\| \bar{G}^{-1}_j \|_{\max}$ grows as a function of $l$ and the p($l$)-CG method will attain a less accurate maximal attainable precision with increasing $l$, see Fig.\,\ref{fig:residuals}. Finally, remark that it was also illustrated by Fig.\,\ref{fig:residuals} that the choice of the shifts $\{\sigma_0, \ldots, \sigma_{l-1}\}$ has a significant impact on the numerical stability of the p($l$)-CG method. A sub-optimal choice for the shifts may lead to a significant increase in the norm of the polynomial, $\|P_l(A)\|$, which bounds the propagation of local rounding errors in p($l$)-CG, see Lemma \ref{lemma:lemma_0}. The `optimal' shifts $\sigma_i$ ($0 \leq i < l$) are the roots of the degree $l$ Chebyshev polynomial, see \eqref{eq:cheb_shifts}, which minimize the 2-norm of $P_l(A)$, see \cite{greenbaum1994gmres,faber2010chebyshev}. 

\section{Experimental results} \label{sec:experiments}

Parallel performance measurements in this section result from a PETSc \cite{petsc-web-page} implementation of p($l$)-CG on a distributed memory machine using the message passing paradigm (MPI). The p($l$)-CG method is validated on a two-dimensional Laplacian PDE model with homogeneous Dirichlet boundary conditions, discretized using second order finite differences on a uniform $n = n_x \times n_y$ point discretization of the unit square. The resulting 5-point stencil Poisson problem forms the basis for many HPC applications to which the pipelined CG method can be applied. The conditioning of these types of systems is typically bad for large problem sizes, implying that iterative solution using Krylov subspace methods is non-trivial. Note that vectors are distributed uniformly across the number of available workers and stored locally (distributed memory). This implies that the matrix is partitioned by contiguous chunks of rows across processors (which is the default way in which PETSc partitions matrix operators).

\subsection{Parallel performance}

\paragraph{\underline{Test setup 1}} The first parallel strong scaling experiment is performed on a small cluster with $20$ compute nodes, consisting of two $6$-core Intel Xeon X5660 Nehalem $2.80$ GHz processors each (12 cores per node). Nodes are connected by $4\,\times\,$QDR InfiniBand technology (32 Gb/s point-to-point bandwidth).
We use PETSc version 3.6.3 \cite{petsc-web-page}. The MPI library used for this experiment is MPICH-3.1.3\footnote{\url{http://www.mpich.org/}}. The PETSc environment variables 
\texttt{MPICH\_ASYNC\_PROGRESS=1} and \texttt{MPICH\_MAX\_THREAD\_SAFETY=multiple} are set to ensure optimal parallelism by allowing for non-blocking global communication, see Section \ref{sec:hiding}.
A 2D Poisson type linear system with exact solution $\hat{x}_j = 1$ and right-hand side $b = A\hat{x}$ is solved, and the initial guess is chosen to be $\bar{x}_0 = 0$.
This benchmark problem is available in the PETSc distribution as example $2$ in the Krylov subspace solvers (KSP) folder. 
The simulation domain is discretized using $1000\times1000$ grid points (1 million unknowns). No preconditioner is applied in this experiment. 
The tolerance imposed on the scaled recursive residual norm $\|\bar{r}_i\|_2 / \|b\|_2$ is $10^{-5}$. For p($l$)-CG stabilizing Chebyshev shifts are used based on the interval $[\lambda_{\min}, \lambda_{\max}] = [0,8]$, see Remark \ref{remark:remark4}.

Fig.\,\ref{fig:speedup1} shows the speedup of CG, p-CG and p($l$)-CG (for different values of $l$) over CG on one node. Timings reported are the most favorable results (in the sense of smallest overall run-time) over 5 individual runs of each method. All methods perform the same number of iterations to reach the preset tolerance. For small numbers of nodes pipelined methods are typically slower than classic CG due to the computational overhead in the initial $l$ iterations 
and the extra vector operations in the basis vector recurrences. For large numbers of nodes pipelined methods outperform classic CG, for which speedup stagnates from 4 nodes onward on this problem setup. The p-CG method \cite{ghysels2014hiding} tends to scale slightly better than the 
p($1$)-CG variant. Scaling of length one pipelined variants also inevitably stagnates from a certain number of nodes onward. The p($l$)-CG methods with deeper pipelines ($l = 2,3$) 
continue to scale beyond the stagnation point of p($1$)-CG. The attainable speedup of p($3$)-CG on 20 nodes over classic CG on 20 nodes is roughly $5 \times$, whereas the p($1$)-CG method on 20 nodes is approximately $3 \times$ faster than classic CG on the same number of nodes. 
The theoretical maximum speedup of p($l$)-CG over CG 
is $t(\text{\textsc{spmv}} + 2\,\text{\textsc{glred}}) / \max(t(\text{\textsc{spmv}}),t(\text{\textsc{glred}})/l)$ where t(\textsc{spmv}) is the time to apply the matrix (and preconditioner) and t(\textsc{glred}) is the time of one global reduction phase, see Table \ref{tab:pipelcg}. This model neglects the time spent in local operations and thus holds only on large numbers of nodes. In the ideal p($l$)-CG scenario when $t(\text{\textsc{glred}}) =  l \, t(\text{\textsc{spmv}})$, the model suggests that the maximal speedup of p($l$)-CG over CG is a factor $(2l+1) \times$.

\begin{figure}
\begin{minipage}{0.48\textwidth}
\centering
\includegraphics[width=1.0\textwidth]{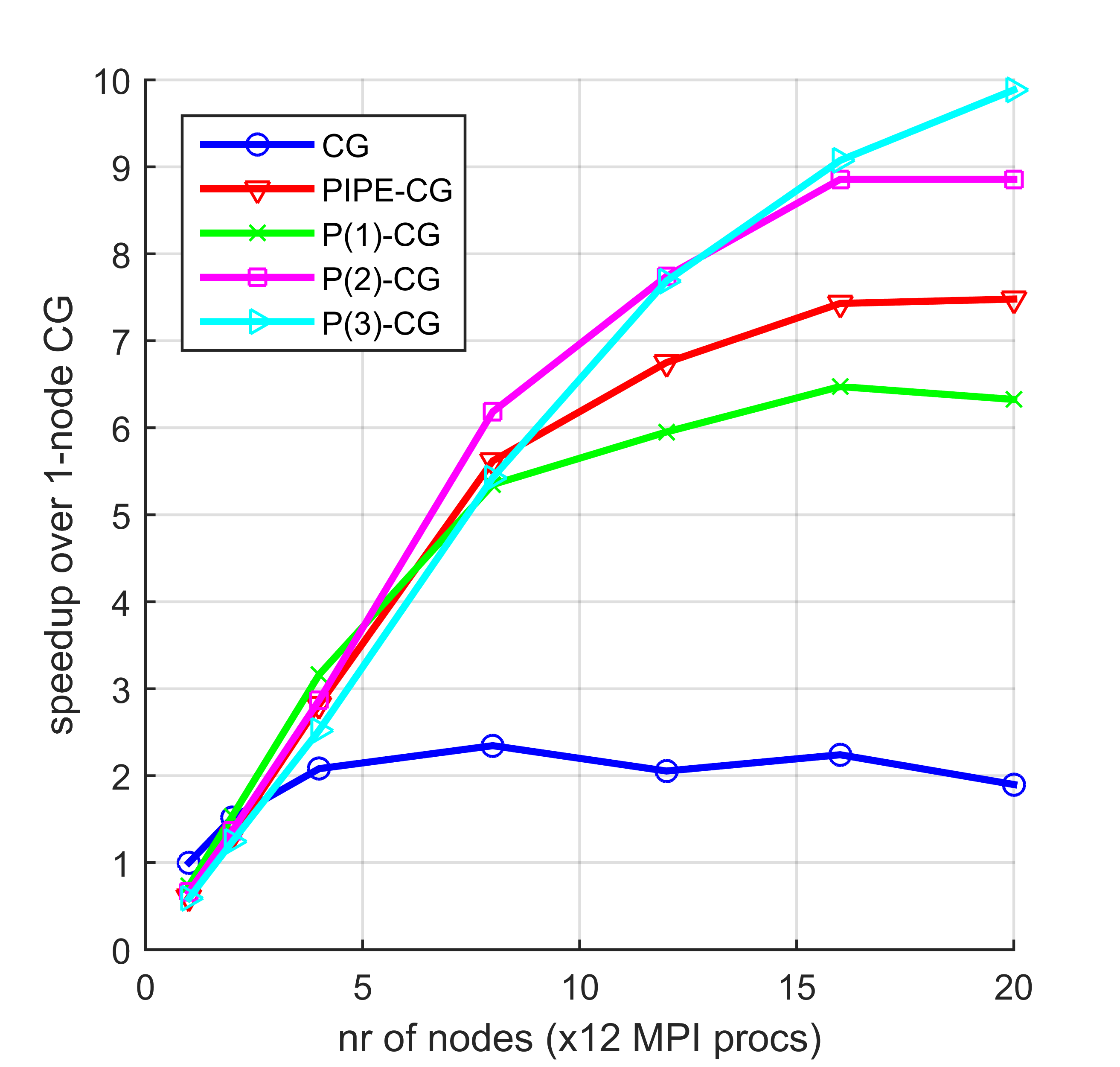}
\vspace{-0.5cm}
\caption{Strong scaling experiment on up to $20$ nodes ($240$ processes) for a 5-point stencil 2D Poisson problem with $1.000.000$ unknowns.
Speedup over single-node classic CG for various pipeline lengths. All methods converged to ${\|\bar{r}_i\|}_2/{\|b\|}_2 = 1.0\text{e-}5$ in 1342 iterations.}
\label{fig:speedup1}
\end{minipage}
\hfill
\begin{minipage}{0.48\textwidth}
\centering
\includegraphics[width=1.0\textwidth]{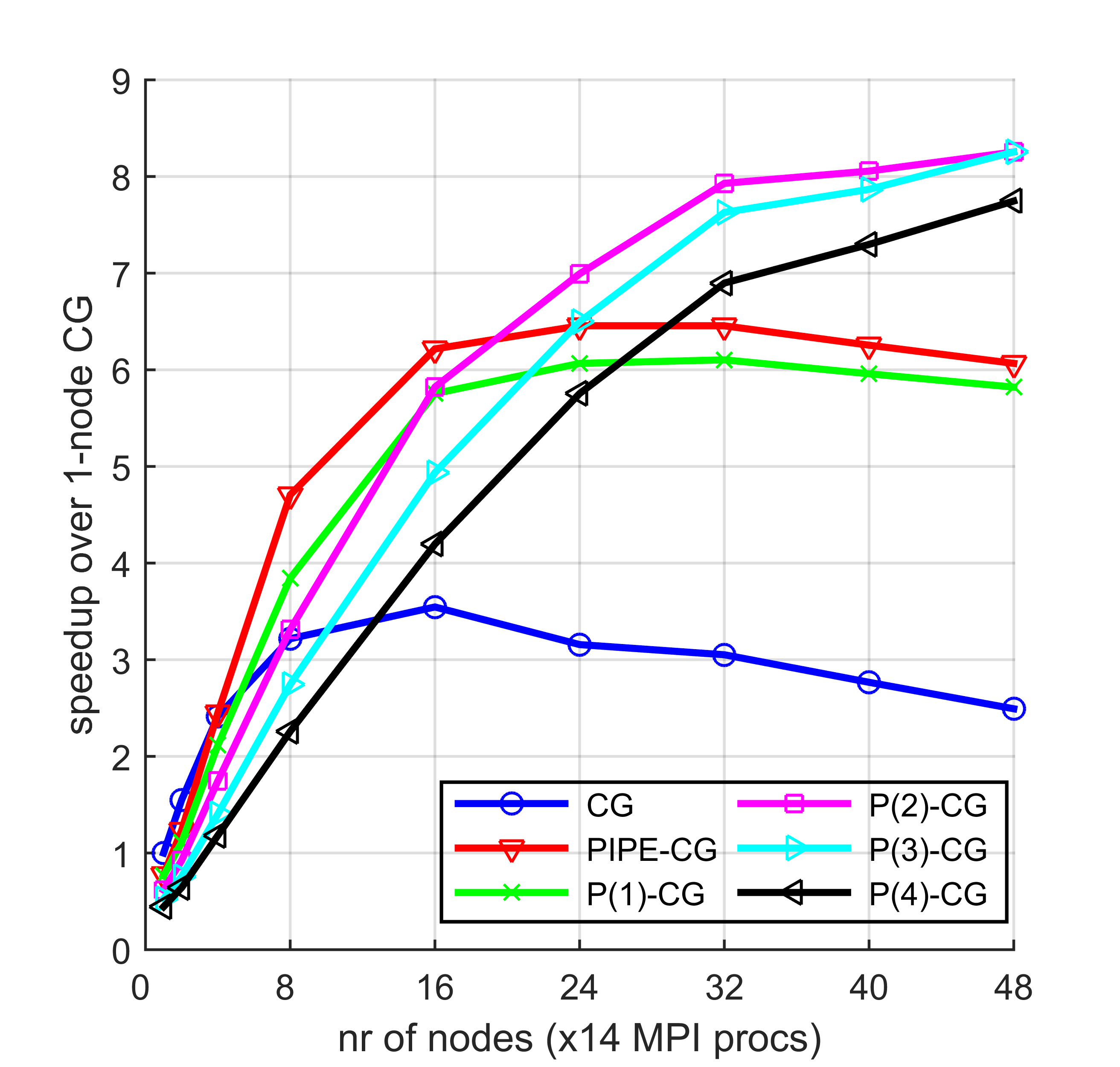}
\vspace{-0.5cm}
\caption{Strong scaling experiment on up to $48$ nodes ($672$ processes) for a 5-point stencil 2D Poisson problem with $3.062.500$ unknowns.
Speedup over single-node classic CG for various pipeline lengths. All methods performed $1500$ iterations with ${\|\bar{r}_i\|}_2/{\|b\|}_2 = 6.3\text{e-}4$.}
\label{fig:speedup2}
\end{minipage}
\end{figure}

\paragraph{\underline{Test setup 2}} A strong scaling experiment on a different hardware setup is shown in Fig.\,\ref{fig:speedup2}. Here a medium-sized cluster with $48$ compute nodes consisting of two $14$-core Intel E5-2680v4 Broadwell generation CPUs connected through an EDR InfiniBand network is used. PETSc version 3.7.6 and MPICH-3.3a2 are installed on the machine. 
A $1750\times 1750$ 2D Poisson type linear system with right-hand side $b = A\hat{x}$, where $\hat{x}_j = 1$, is solved on this system. 
The outcome of the 
 scaling experiment is shown in Fig.\,\ref{fig:speedup2}. The number of iterations was capped at 1500 for this problem, which is equivalent to a relative residual norm tolerance of $6.3$e-4. Similar observations as for Fig.\,\ref{fig:speedup1} can be made; the achievable parallel performance gain by using longer pipelines is apparent from the experiment. However, the balance between time spent in communication vs.~computation is clearly different from the experiment reported in Fig.\,\ref{fig:speedup1}. Note that on this problem for pipeline lengths $l \geq 4$ the computational overhead of the initial $l$ start-up iterations (in which the pipeline is filled) and the additional \textsc{axpy} operations required for the basis vector recurrences slow down the algorithm significantly. Hence, on up to 48 nodes the p($l$)-CG algorithm with $l = 2,3$ slightly outperforms the p($l$)-CG algorithm with $l \geq 4$. Deeper pipelined p($l$)-CG ($l \geq 4$) methods are expected to eventually scale further, achieving even better speedups beyond the number of nodes reported here.

\subsection{Preconditioning}

\begin{figure}
\begin{minipage}{0.48\textwidth}
\centering
\includegraphics[width=1.0\textwidth]{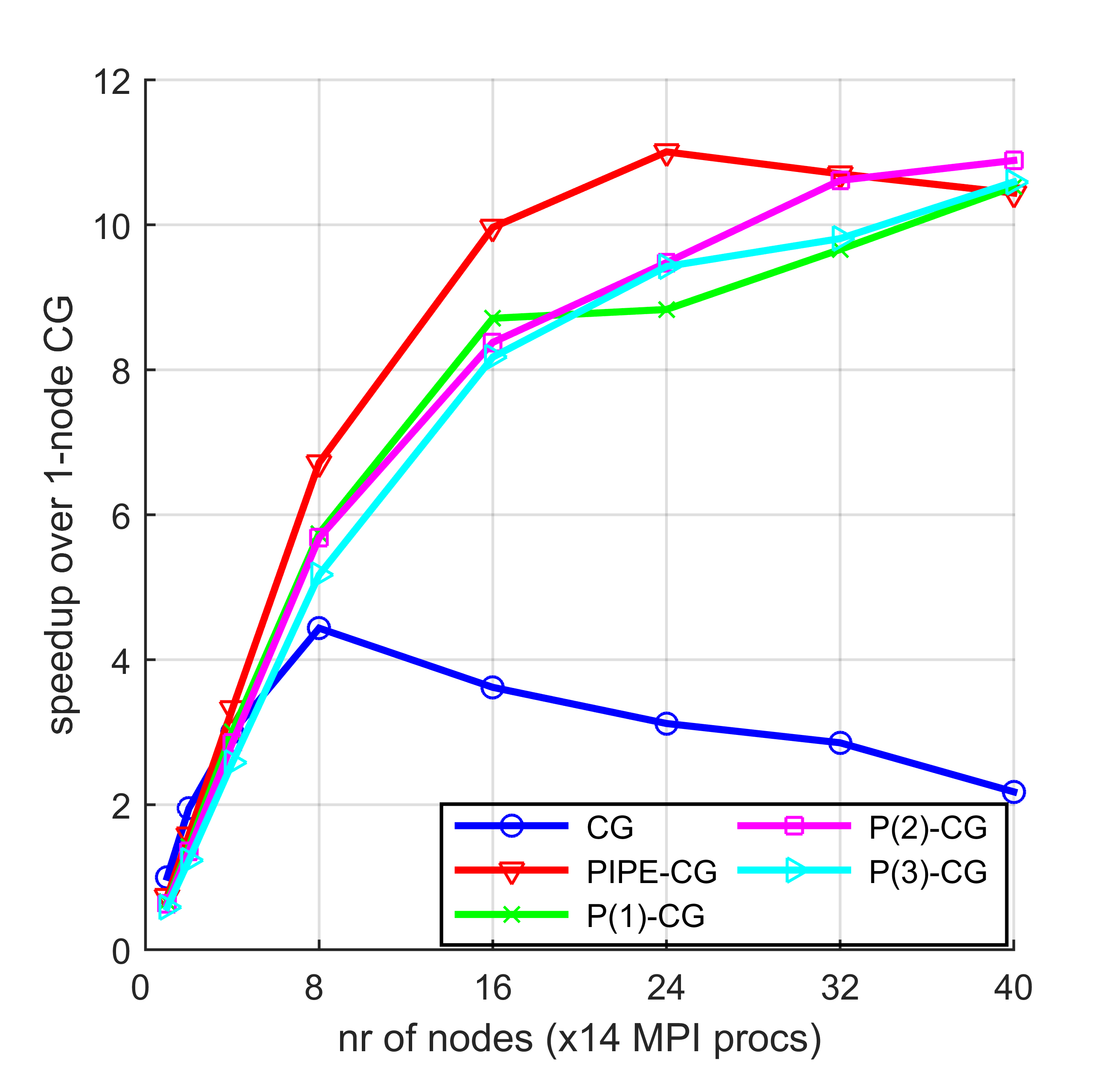}
\vspace{-0.5cm}
\caption{Strong scaling experiment on up to $40$ nodes ($560$ processes) for a block Jacobi preconditioned 2D Poisson problem with $3.062.500$ unknowns. All methods performed $600$ iterations with ${\|\bar{r}_i\|}_2/{\|b\|}_2 = 1.8\text{e-}4$ (on 1 node) and ${\|\bar{r}_i\|}_2/{\|b\|}_2 \leq 9.3\text{e-}4$ (on 40 nodes).}
\label{fig:speedup3}
\end{minipage}
\hfill
\begin{minipage}{0.48\textwidth}
\centering
\includegraphics[width=1.0\textwidth]{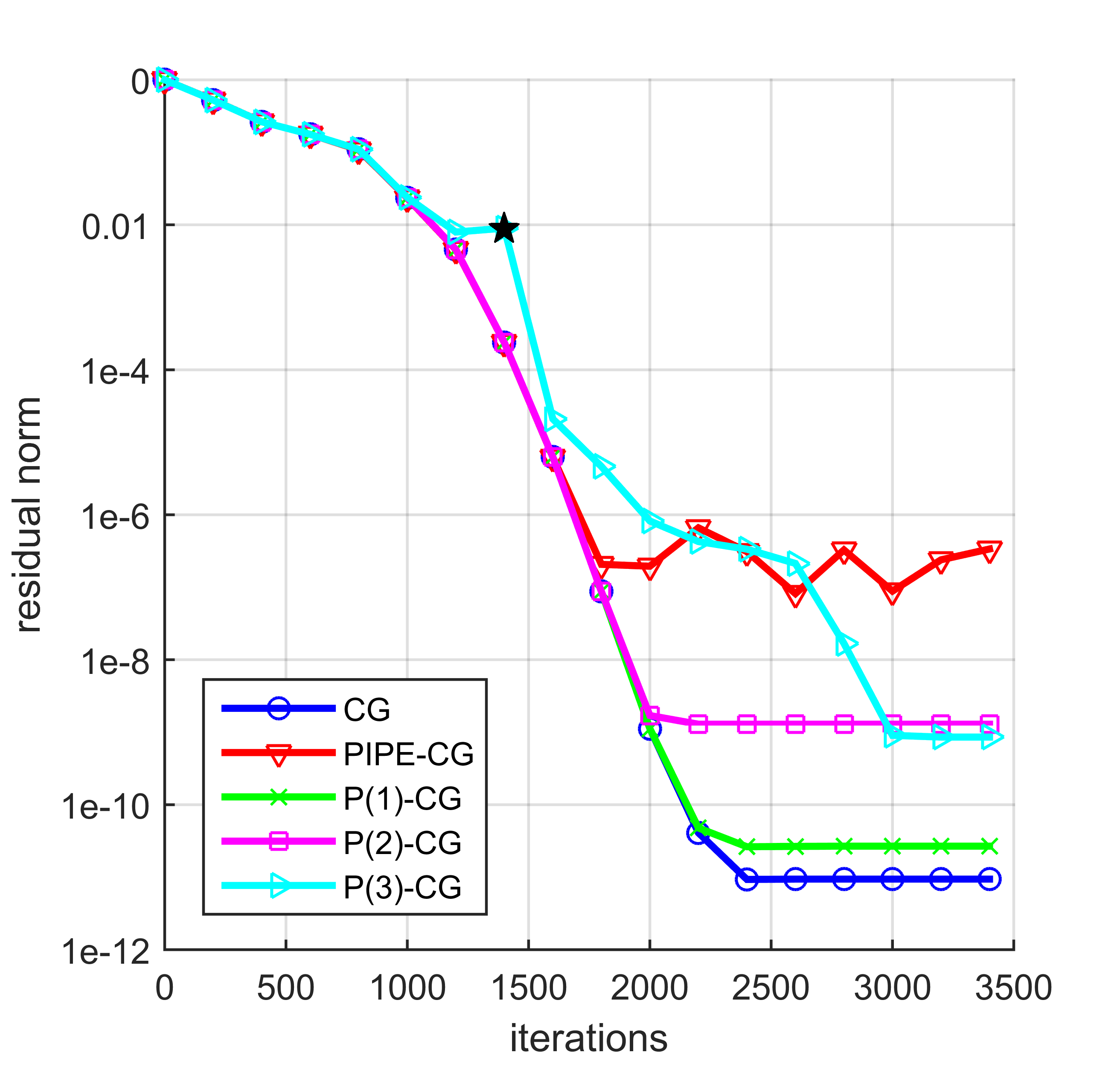}
\vspace{-0.5cm}
\caption{Accuracy experiment on 20 nodes (240 processes) for a 5-point stencil 2D Poisson problem with $1.000.000$ unknowns. 
Residual norm ${\|b-A\bar{x}_i\|}_2$ as a function of time spent by the algorithm. Minimal/maximal number of iterations is $200/3400$ for all methods.}
\label{fig:accuracy1}
\end{minipage}
\end{figure}

Fig.\,\ref{fig:speedup3} shows another parallel performance experiment in the setting of \emph{Test setup 2}. Contrary to Fig.\,\ref{fig:speedup2}, here a preconditioner is included and Alg.\,\ref{algo:plCG} is applied to solve the preconditioned system, see Section \ref{sec:preconditioned} for details. The preconditioner is a simple block Jacobi scheme, supplied to PETSc by the argument \texttt{-pc\_type bjacobi}, where the local blocks are approximately inverted using ILU. Its straightforward parallelism makes block Jacobi an ideal preconditioner for pipelined methods, although the convergence benefit of this preconditioner may deteriorate slightly as the number of nodes increases. For preconditioned p($l$)-CG Chebyshev shifts based on the interval $[\lambda_{\min}, \lambda_{\max}] = [0,1.5]$ are used, cf.\,Remark \ref{remark:remark4}.
After 600 iterations a relative residual accuracy ${\|\bar{r}_i\|}_2/{\|b\|}_2 \leq 1\text{e-}3$ is reached for all methods and node numbers. 
On 40 nodes the pipelined methods all show comparable speedups over classic CG. The p(2)-CG method outperforms p(1)-CG by a small fraction on 40 nodes. Compared to the unpreconditioned experiments reported in Fig.\,\ref{fig:speedup2} the performance gains of using longer pipelines are rather limited here. This can be understood by considering the balance between time spent in computations vs.\,communication. Figs.\,\ref{fig:breakdown1}-\ref{fig:breakdown2} give a detailed overview of the time spent in each phase for the CG and p(1)-CG methods respectively. 
In Fig.\,\ref{fig:breakdown2} the `\textsc{glred}' bar represents the part of the global communication phase that is not overlapped with the \textsc{spmv} and preconditioner. It is clear from Fig.\,\ref{fig:breakdown2} that on up to 32 nodes there is no more time to gain from overlapping the communication phase with more than one \textsc{spmv} and preconditioner application, and hence for this problem setup the additional speedup achievable by using pipelines with $l > 1$ compared to p(1)-CG is expected to be small, as illustrated in Fig.\,\ref{fig:speedup3}.

\subsection{Numerical accuracy}

Fig.\,\ref{fig:accuracy1} presents an accuracy experiment for \emph{Test setup 1}. The actual residual 2-norm ${\|b-A\bar{x}_i\|}_2$ is shown as a function of 
iterations. 
Fig.\,\ref{fig:accuracy1} indicates that the maximal attainable accuracy of the p($l$)-CG method decreases with growing pipeline lengths $l$. It was analyzed in \cite{cools2018analyzing,carson2016numerical} that the attainable accuracy of p-CG can be significantly worse compared to classic CG due to the propagation of local rounding errors in the vector recurrences. Whereas it is clear from the figure that the maximal attainable accuracy for p(2)-CG is worse than for p(1)-CG, the latter appears to be 
more 
robust to rounding errors 
compared to the 
length one pipelined p-CG method from \cite{ghysels2014hiding}. We again point out that the p-CG and p($l$)-CG methods are essentially different 
algorithms as stated in Remark \ref{remark:remark9}. Note that the p(3)-CG method encounters a square root breakdown in iteration 1393 (Fig.\,\ref{fig:accuracy1}, $\bigstar$) and consequently performs a restart. Due to the restart the accuracy 
attainable by p(3)-CG and p(2)-CG is comparable; however, the number of iterations required to reach it is considerably higher for p(3)-CG. 
Supplementary and insightful numerical experiments on maximal attainable accuracy for a variety of SPD matrices from the Matrix Market collection (\url{http://math.nist.gov/MatrixMarket/}) can be found in Table \ref{tab:matrix_market} in Appendix \ref{sec:supplementary}.

\begin{figure}
\begin{minipage}{0.48\textwidth}
\centering
\includegraphics[width=1.0\textwidth]{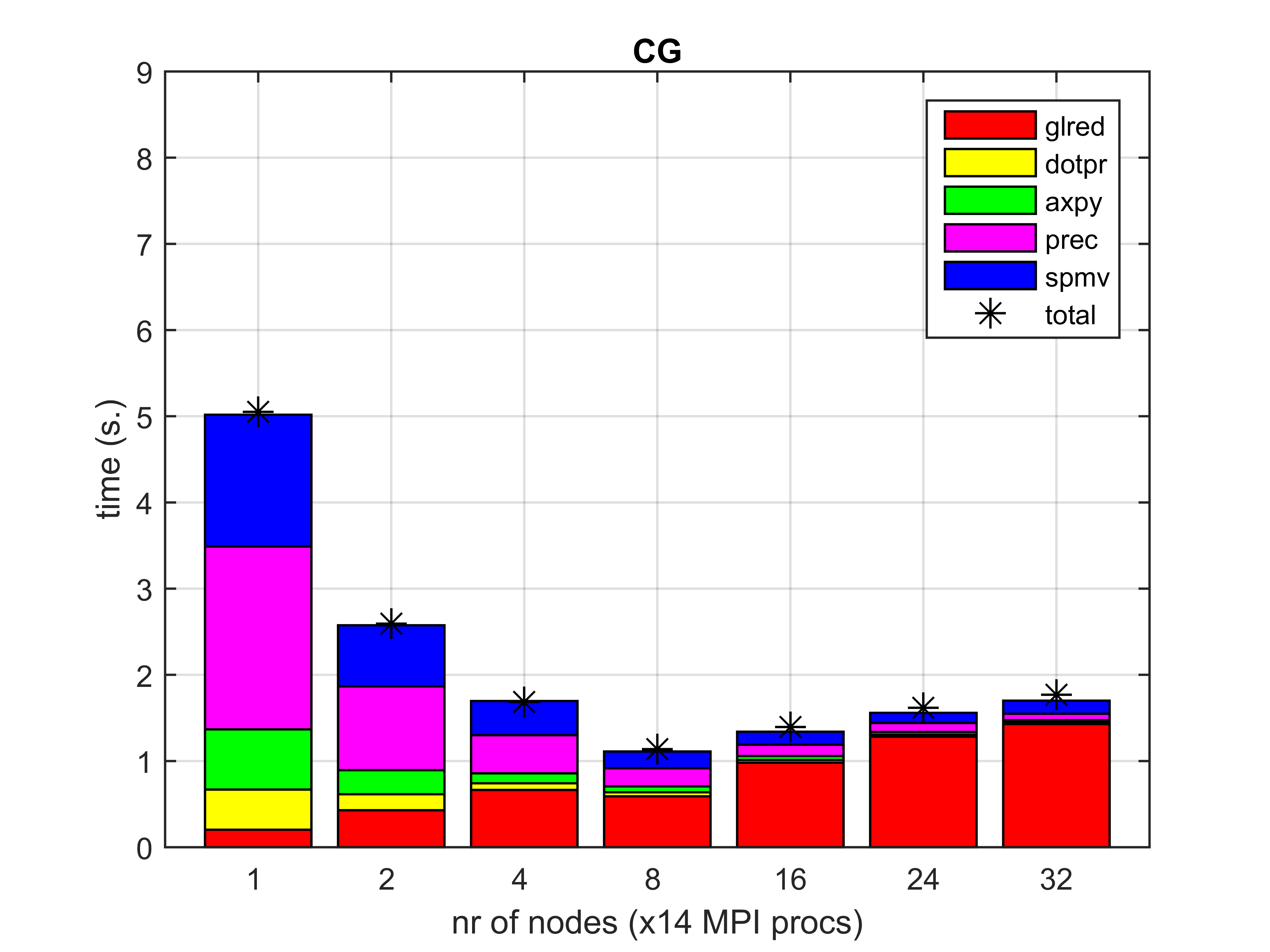}
\vspace{-0.5cm}
\caption{Detailed timing breakdown of the classic CG algorithm for the $3.062.500$ unknowns 2D Poisson strong scaling experiment in Fig.\,\ref{fig:speedup3} on up to $32$ nodes.}
\label{fig:breakdown1}
\end{minipage}
\hfill
\begin{minipage}{0.48\textwidth}
\centering
\includegraphics[width=1.0\textwidth]{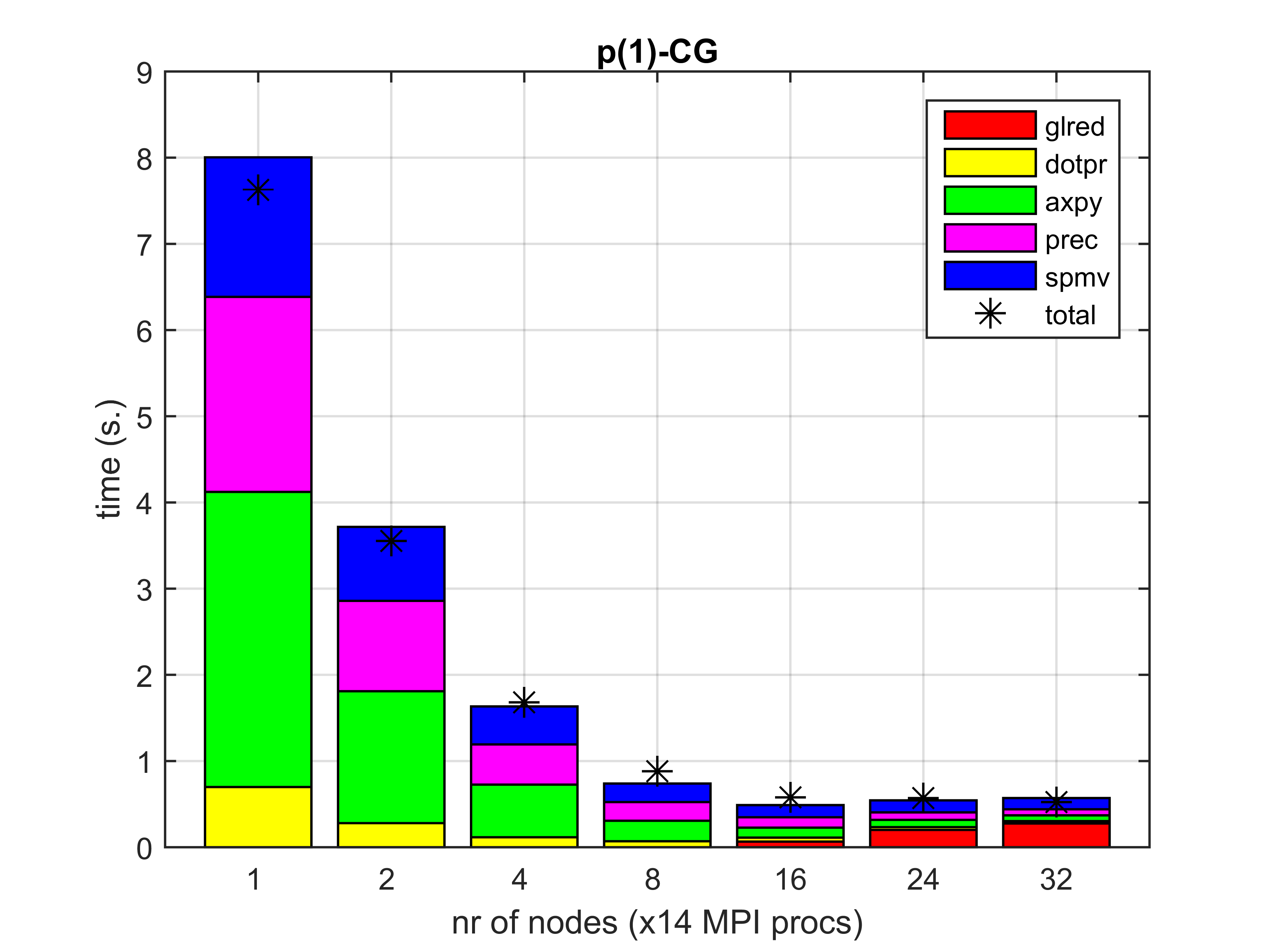}
\vspace{-0.5cm}
\caption{Detailed timing breakdown of the p(1)-CG algorithm for the $3.062.500$ unknowns 2D Poisson strong scaling experiment in Fig.\,\ref{fig:speedup3} on up to $32$ nodes.}
\label{fig:breakdown2}
\end{minipage}
\end{figure}

\begin{figure}
\begin{center}
\begin{tabular}{cc}
\includegraphics[width=0.47\textwidth]{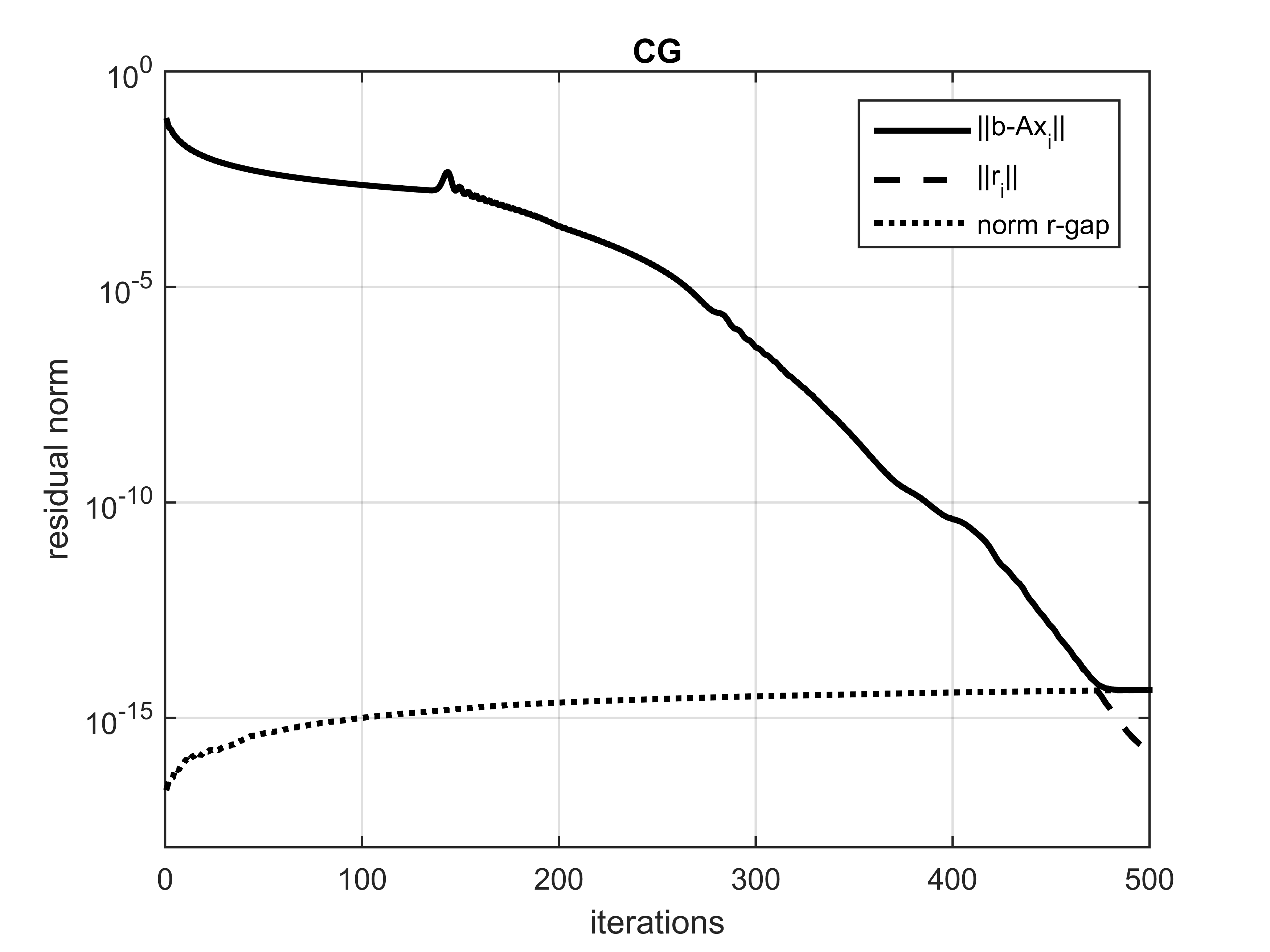} &
\includegraphics[width=0.47\textwidth]{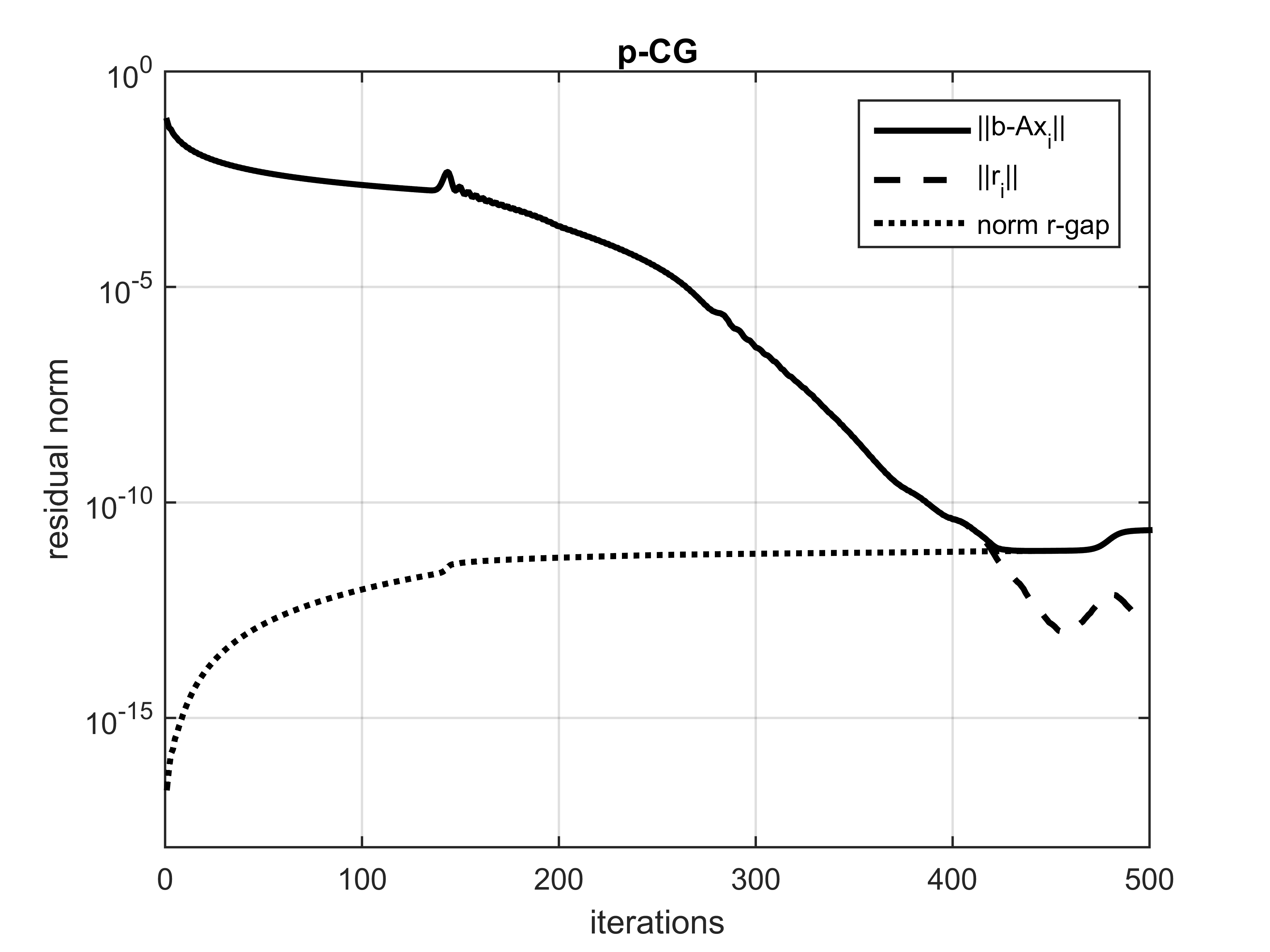} \\
\includegraphics[width=0.47\textwidth]{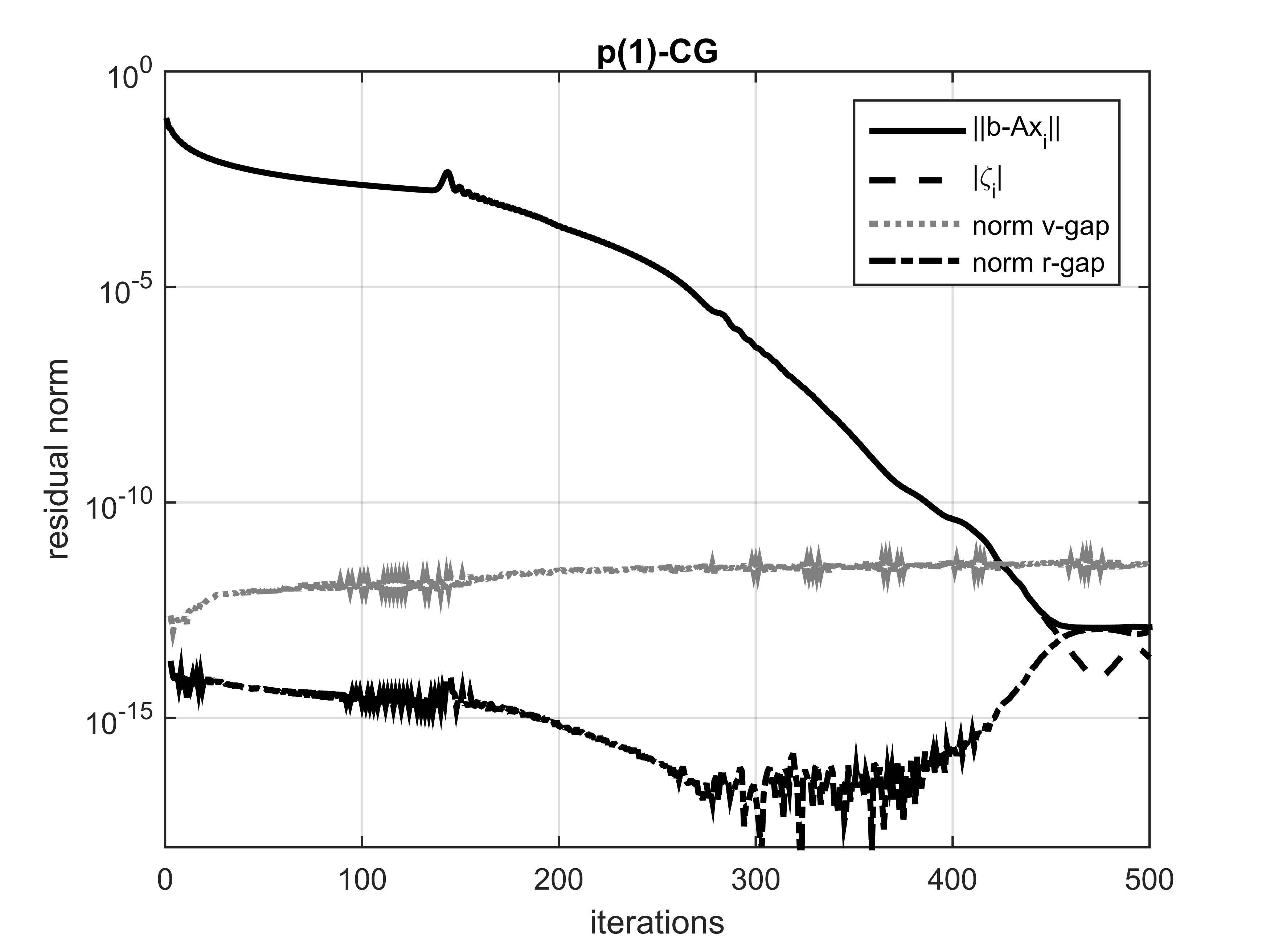} &
\includegraphics[width=0.47\textwidth]{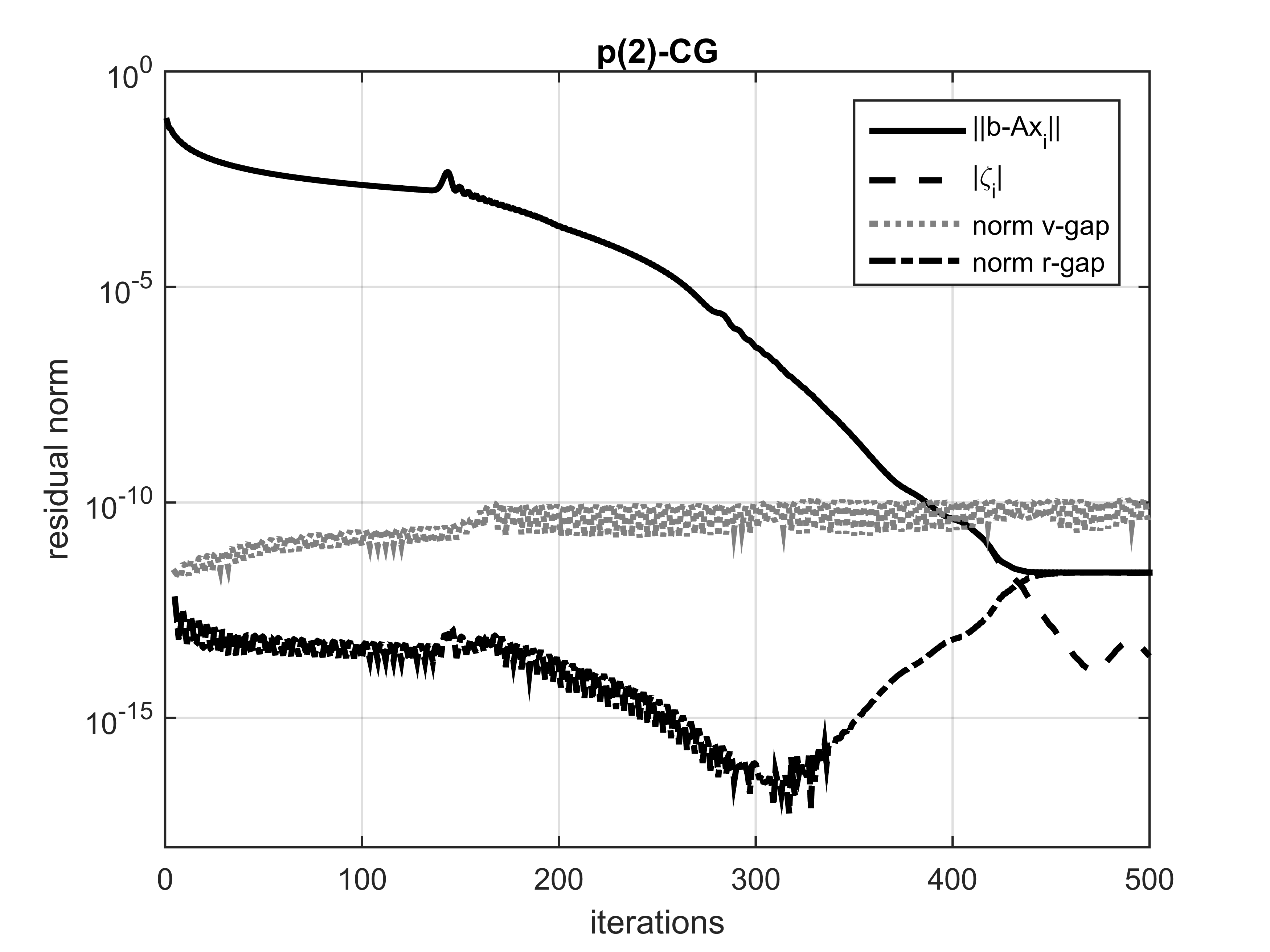} \\
\includegraphics[width=0.47\textwidth]{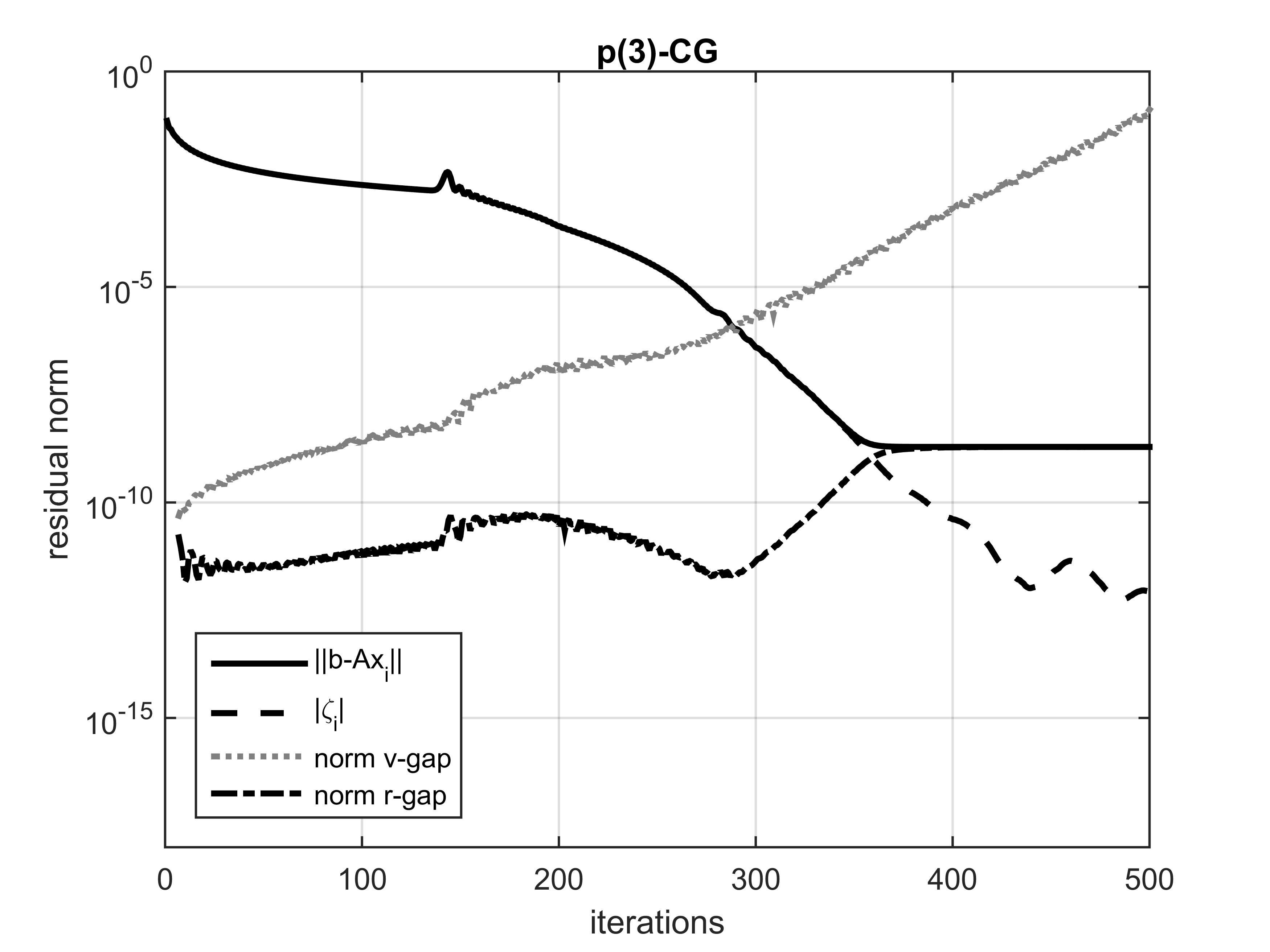} &
\includegraphics[width=0.47\textwidth]{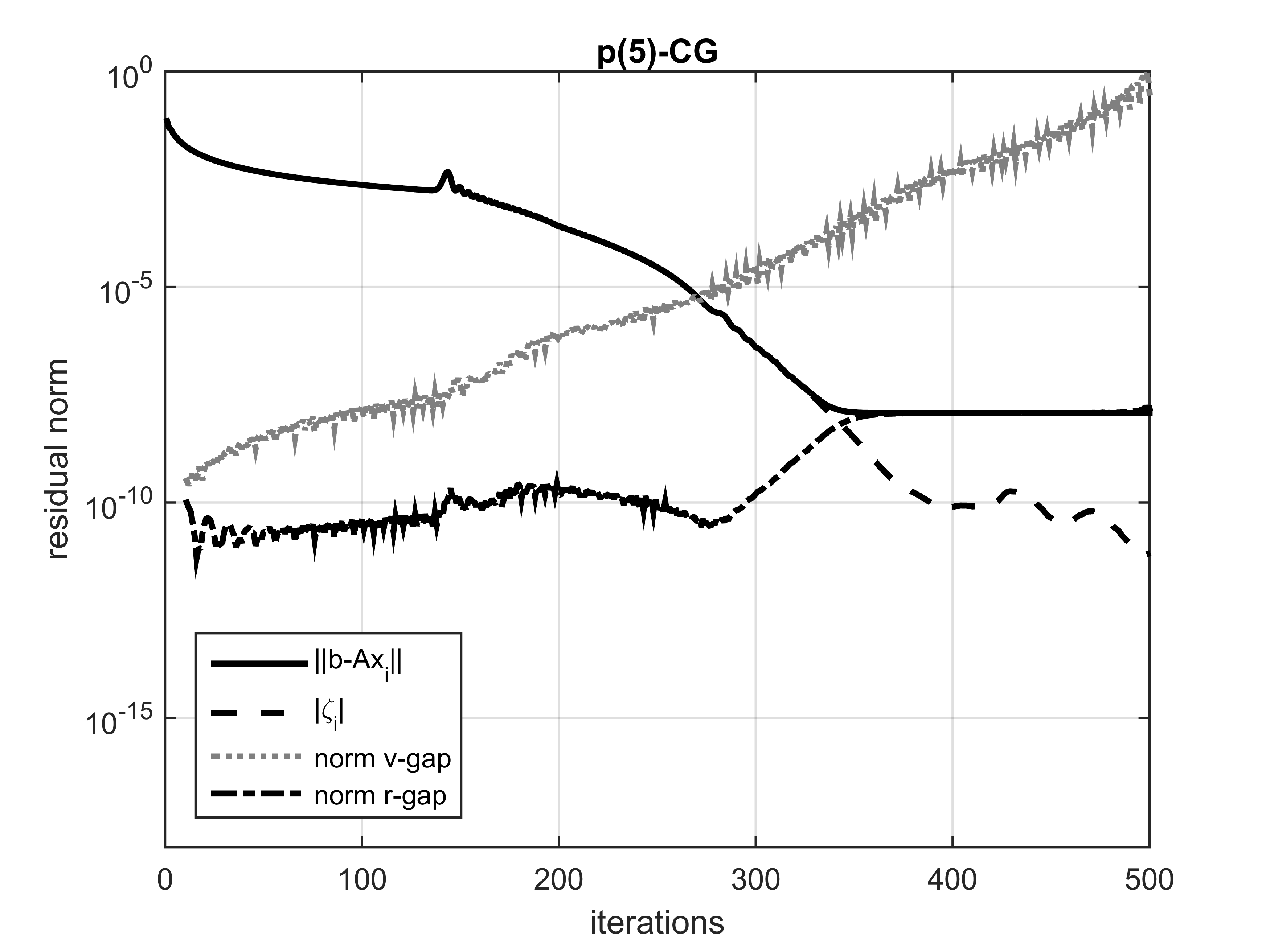} 
\end{tabular}
\end{center}
\caption{Actual residual norm $\|b-A\bar{x}_{j+1}\|$, recursively computed residual norm $\|\bar{r}_{j+1}\|$, basis gap norm $\|\bar{\bold{v}}_{j+1} - \bar{v}_{j+1}\|$ and residual gap norm $\|\bar{\bold{r}}_{j+1}-\bar{r}_{j+1}\|$ for different CG variants on a 2D Poisson problem with 200 $\times$ 200 unknowns corresponding to Fig.\,\ref{fig:residuals} (left). For CG and p-CG (top) the norm of the residual gap $(b-A\bar{x}_{j+1}) - \bar{r}_{j+1}$ is displayed, where $\bar{r}_{j+1}$ is computed using \eqref{eq:recs_xandr}. For p($\ell$)-CG with $l = 1,2,3,5$ (middle and bottom) the norms of the gaps $\bar{\bold{v}}_{j-l+1} - \bar{v}_{j-l+1}$ and $\bar{\bold{r}}_{j-l+1} - \bar{r}_{j-l+1}$ are shown, where $\bar{\bold{v}}_{j-l+1}$ satisfies \eqref{eq:v_arnoldi}, $\bar{v}_{j-l+1}$ is computed using the recurrence \eqref{eq:vbar_rec} and $\bar{\bold{r}}_{j-l+1} - \bar{r}_{j-l+1}$ is defined by \eqref{eq:extra_res5} using $\bar{\bold{v}}_{j-l+1}$ and $\bar{v}_{j-l+1}$.}
\label{fig:figure2}
\end{figure}


We now present numerical experiments regarding the numerical analysis of the p($l$)-CG method discussed in Section \ref{sec:analysis}.
Consider the $200 \times 200$ discretized 2D Poisson problem on 
$[0,1]^2$ for illustration purposes in this remainder of this section. The right-hand side is $b = A\hat{x}$ with $\hat{x} = 1/\sqrt{n}$, unless explicitly stated otherwise. 
This relatively simple problem is severely ill-conditioned and serves as an adequate tool to demonstrate the rounding error analysis from Section \ref{sec:analysis}. 

Fig.\,\ref{fig:figure2} shows the norms of the actual residual $\bar{\bold{r}}_j = b-A\bar{x}_j$ and computed residual $\bar{r}_j$ for various CG variants. The norm $\|\bar{r}_j\|$ is computed as $|\zeta_j|$ in p($\ell$)-CG, see Theorem \ref{th:resnorm}. The figure also displays the norm of the residual gap $(b-A\bar{x}_j) - \bar{r}_j$ for all methods, see expressions \eqref{eq:f_CG} and \eqref{eq:extra_res5}, and the norm of the gap $\bar{\bold{v}}_j - \bar{v}_j$ on the basis for p($\ell$)-CG, see \eqref{eq:gap_plcg2}. For p($\ell$)-CG the gap between $\bar{\bold{v}}_j$ and $\bar{v}_j$ increases dramatically as the iteration proceeds, particularly for large values of $l$, leading to significantly reduced maximal attainable accuracy as indicated by the related residual gap norms. Note that the residual gaps are not guaranteed to be monotonically increasing. The term $\|(\bar{\bold{V}}_{j+1}-\bar{V}_{j+1}) \bar{\Delta}_{j+1,j} \bar{U}^{-1}_{j} \bar{q}_j\|$ in expression \eqref{eq:extra_res5} does not necessarily increase, as error gaps on $\bar{v}_j$ may be averaged out by the linear combination. The following residuals norms $\|b-A\bar{x}_j\|$ are attained after 500 iterations: $4.47\text{e-}15$ (CG), $2.28\text{e-}11$ (p-CG), 1.27\text{e-}13 (p($1$)-CG), $2.37\text{e-}12$ (p($2$)-CG), $1.94\text{e-}09$ (p($3$)-CG), $1.19\text{e-}08$ (p($5$)-CG).

\begin{figure}
\begin{center}
\includegraphics[height=0.35\textwidth]{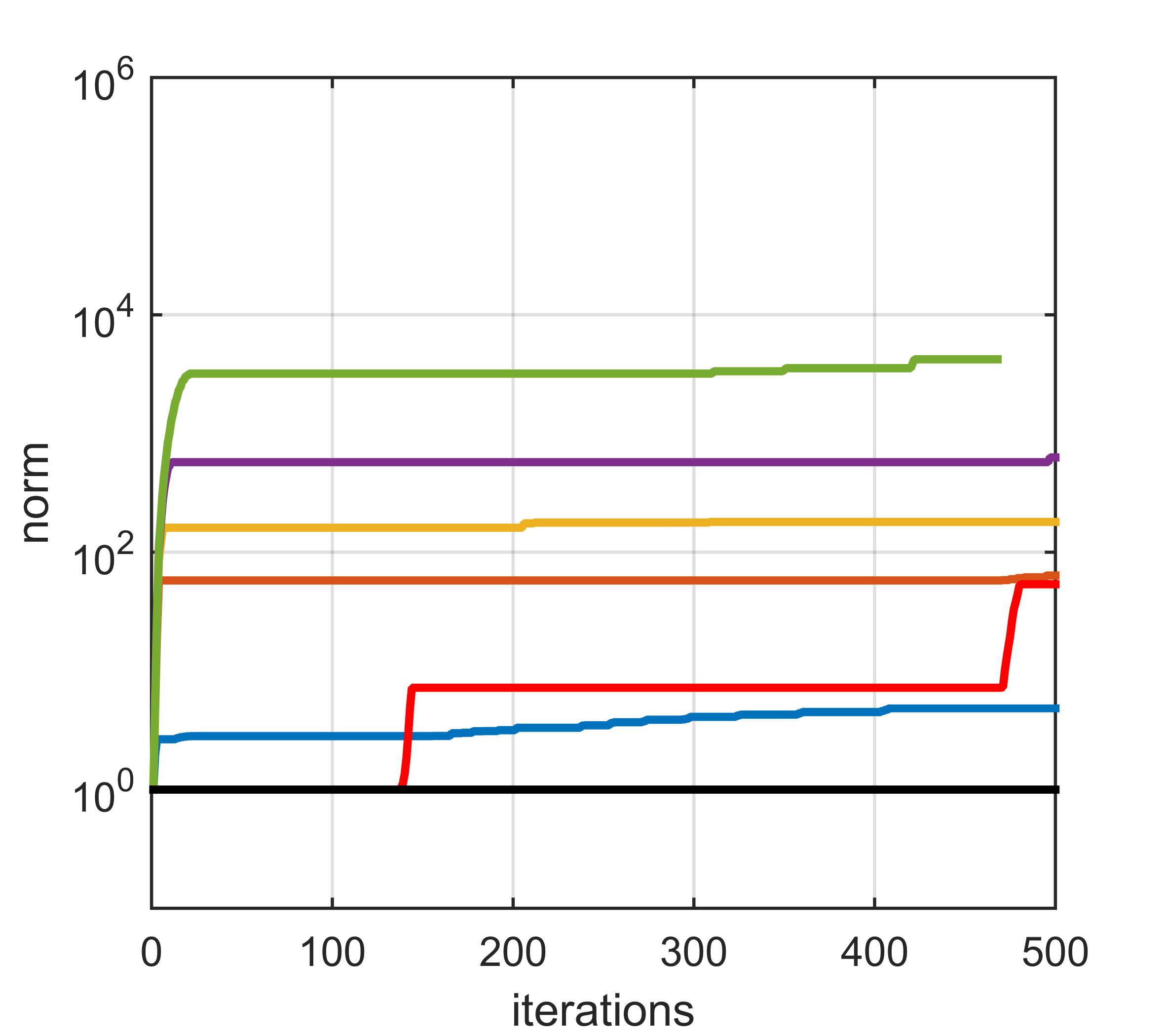} 
\includegraphics[height=0.35\textwidth]{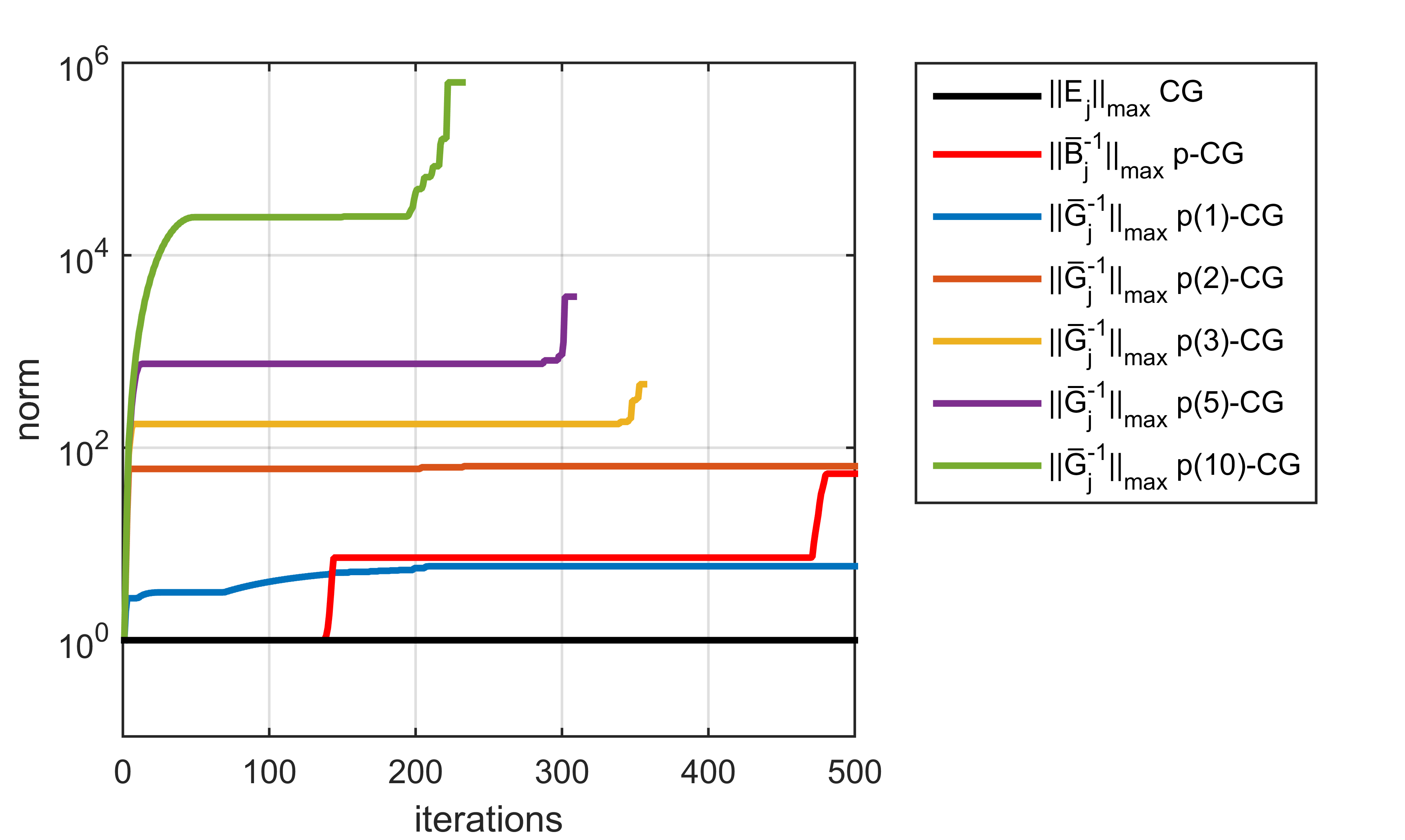} 
\end{center}
\caption{Maximum norms of the essential matrices $E_j$, $\bar{\mathcal{B}}_j^{-1}$ and $\bar{G}_j^{-1}$ involved in the local rounding error propagation for different variants of CG. See expression \eqref{eq:f_CG} for CG and reference \cite{cools2018analyzing} for p-CG. See also Appendix \ref{sec:analysis_cg} for definitions of the matrices $E_j$ and $\bar{\mathcal{B}}_j^{-1}$ that govern the residual gaps for CG and p-CG respectively. See expressions \eqref{eq:extra_res5} and \eqref{eq:gap_plcg2} for details on the residual gap in p($\ell$)-CG. Left: with optimal Chebyshev shifts on the interval [0,8], cf.~Fig.\,\ref{fig:residuals} (left). Right: with sub-optimal Chebyshev shifts on the interval [0,8*1.005], cf.~Fig.\,\ref{fig:residuals} (right).}
\label{fig:figure3}
\end{figure}

In Fig.\,\ref{fig:figure3} the maximum norm $\|\bar{G}_j^{-1}\|_{\max}$ is shown as a function of the iteration $j$ for different pipeline lengths $l$. The maximum norms $\|E_j\|_{\max}$ for CG, see \eqref{eq:f_CG}, and $\|\mathcal{B}^{-1}_j\|_{\max}$ for p-CG, see \cite{cools2018analyzing}, are also displayed as a reference, see also Appendix \ref{sec:analysis_cg} for more details. The impact of increasing pipeline lengths on numerical stability is clear from the figure. A comparison between the left panel (optimal shifts) and the right panel (sub-optimal shifts) in Fig.\,\ref{fig:figure3} illustrates the influence of the basis choice on the norm of $\bar{G}_j^{-1}$, see Section \ref{sec:bounding}. No data is plotted when the matrix $\bar{G}_j$ becomes numerically singular, which corresponds to iterations in which a square root breakdown occurs
, see Fig.\,\ref{fig:residuals}.
Relating Fig.\,\ref{fig:figure3} to the corresponding convergence histories in Fig.\,\ref{fig:residuals}, 
it is clear that the maximal attainable accuracy for p-CG is comparable to that of p(2)-CG, whereas the p(1)-CG algorithm is able to attain a better final precision. The final accuracy level at which the residual stagnates degrades significantly for longer pipelines.

Fig.\,\ref{fig:figure7} combines performance and accuracy results of several variants to the CG method into a single figure. The figure shows the relative actual residuals $\|b-A\bar{x}_j\|/\|b\|$ for the $750 \times 750$ 2D Laplace problem as a function of total time spent by the algorithm. The time spent to compute the actual residuals was not included in the reported timings, since they are in principle not computed in Alg.\,\ref{algo:plCG}. The experiment is executed on 10 of the nodes specified in \emph{Test setup 1}. 
The PETSc version used was 3.8.3 in this experiment and communication was performed using Intel MPI 2018.1.163. 
Similar to the results in Fig.\,\ref{fig:speedup1}-\ref{fig:speedup2}, the pipelined methods require significantly less overall time compared to classic CG.
The p(2)-CG algorithm outperforms the other CG variants 
in terms of time to solution, 
although it cannot reach the same maximal accuracy as the CG or p($1$)-CG methods. Note that p($3$)-CG is able to attain a residual that \emph{does} satisfy $\|b-A\bar{x}_j\|/\|b\| \leq 1.0\text{e-}13$, whereas p($1$)-CG and p($2$)-CG are not. This is due to the square root breakdown and subsequent restart of the p($3$)-CG algorithm as described in Remark \ref{remark:sqrt_breakdown}. The restart improves final attainable accuracy but delays the algorithm's convergence 
compared to other pipelined variants.

\begin{figure}[t]
\begin{center}
\includegraphics[width=0.98\textwidth]{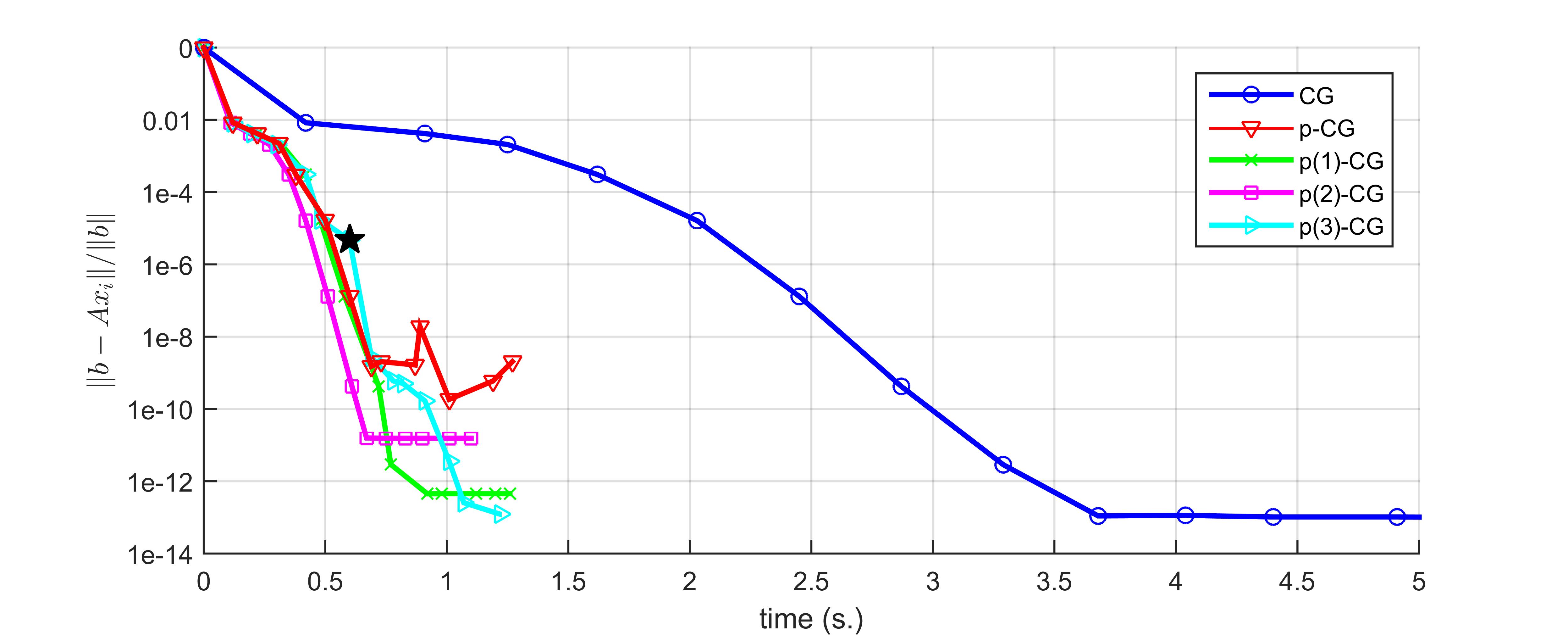}
\end{center}
\caption{Performance/accuracy experiment on 10 nodes (120 processes) for a 5-point stencil 2D Poisson problem with 562,500 unknowns. Relative residual norm $\|b - A \bar{x}_j\|/\|b\|$ as a function of total time spent by the algorithm for various (pipelined) CG variants. Square root breakdown in p(3)-CG is indicated by the $\bigstar$ symbol. 
}
\label{fig:figure7}
\end{figure}

\section{Conclusions}
\label{sec:conclusions}

As HPC hardware keeps evolving towards exascale the gap between computational performance 
and communication latency keeps increasing. Many numerical methods that 
are historically optimized towards flop performance 
now need to be revised towards also (or even: primarily) 
minimizing communication overhead. Several research teams are currently working towards 
this goal \cite{carson2013avoiding,mcinnes2014hierarchical,ghysels2013hiding,grigori2016enlarged,imberti2017varying}, 
resulting in a variety of communication reducing variants to classic Krylov subspace methods 
that feature improved scalability on massively parallel hardware.

This work reports on our efforts to extend the communication-hiding
pipe\-lined Conjugate Gradient (p-CG) method to deeper pipelines. The paper derives a variant
of CG with deep pipelines, discusses implementation issues,
comments on the numerical stability of the algorithm in finite precision, and presents
proof-of-concept scaling results.
Contrary to the p-CG method (with pipeline length one) introduced by Ghysels et al.\,in 2014
\cite{ghysels2014hiding}, the theoretical framework for the p($l$)-CG
algorithm is derived starting from the p($l$)-GMRES method
\cite{ghysels2013hiding}, rather than the original CG method \cite{hestenes1952methods}. 
The p($l$)-CG method is shown to be a simplification of the p($l$)-GMRES
variant from which it was derived in terms of computational and
storage costs, which is achieved by exploiting the 
symmetry of the system matrix and by imposing residual orthogonality.

On massively parallel machines, where the overall solution time is
dominated by global reduction latency, deep pipelined methods 
outperform the classic CG and p-CG algorithms,
as illustrated by the 
experiments 
in this work.
Initial test results show improved scalability when 
longer pipelines are used on distributed-multicore hardware.
However, contrary to many existing CG method variants
\cite{meurant1987multitasking,chronopoulos1989s,ghysels2014hiding,carson2015communication},
the p($l$)-CG algorithm may encounter square root
breakdowns. 
Implementation issues and corresponding solutions for the new CG variant are presented, including
discussions on overlapping communication latency with computational
work and on memory requirements.

It is observed that longer pipelines have an impact
on the propagation of local rounding errors and may affect the 
attainable accuracy on the solution, cf.\,\cite{carson2016numerical,cools2018analyzing}.
This observation is supported by the numerical analysis presented in the manuscript.
Practical bounds for the propagation of the local rounding errors are derived, 
leading to insights into 
the influence of the pipelined length 
and the choice of the auxiliary Krylov basis on 
attainable accuracy. 
It should be noted that the analysis in this work 
does not take into account the impact 
of loss of orthogonality due to rounding error propagation. 
We also do not be expect the analysis to be directly 
applicable to other pipelined Krylov subspace methods such as p($\ell$)-GMRES, 
although the general approach would likely show resemblances.

In summary, the main result of the paper is to show that it is indeed
possible to introduce longer pipelines in the CG algorithm and to give a
first experimental verification of the improved scalability. However,
the algorithm is rather technical to implement, requires additional storage for
the auxiliary variables and features multi-term recurrences that may affect the
numerical accuracy.
Several directions for future research are suggested by the remarks
throughout this manuscript. The robustness of the p($l$)-CG method to
rounding error propagation 
and the impact of deeper pipelines on attainable accuracy should be
improved in future work. 
An interesting technique was recently presented by Imberti et al.~\cite{imberti2017varying}
for $s$-step GMRES; however, it remains to be determined whether a similar 
idea is suitable to `stabilize' p($l$)-CG.
Finally, although the performance results reported in this
work validate the scalability of p($l$)-CG for deeper pipelines, 
it would be interesting to perform large-scale experiments on
even bigger parallel systems where very deep pipelines are expected to
be even more beneficial, cf.\,\cite{yamazaki2017improving} for
p($l$)-GMRES.

\section*{Acknowledgments}
J.\,C.\,acknowledges funding by the University of Antwerp Research Council under the University Research Fund (BOF).
S.\,C.\,gratefully acknowledges funding by the Flemish Research Foundation (FWO Flanders) under grant 12H4617N. 
The authors would like to cordially thank Pieter Ghysels (LBNL) for useful comments on previous versions 
of this manuscript and related discussions on the topic. 
Additionally, the authors gratefully acknowledge the input of the anonymous referees who aided in optimizing the contents of this paper.

\bibliographystyle{plain}
\bibliography{refs2}

\newpage

\appendix

\begin{center}{\uppercase{\textbf{The communication-hiding Conjugate Gradient method with deep pipelines} \\ \ \\ 
{\footnotesize Jeffrey Cornelis, Siegfried Cools, and Wim Vanroose} \\ \ \\ \vspace{0.2cm}
\Large{\textbf{Appendix }}}}\end{center} \vspace{0.3cm}

\noindent \hrulefill \ \\ \vspace{-0.2cm}

\noindent \textbf{This document is intended as Supplementary Materials to the SIAM Journal on Scientific Computing manuscript 
``The communication-hiding Conjugate Gradient method with deep pipelines''
by J.~Cornelis, S.~Cools and W.~Vanroose.}

\noindent \textbf{Index of Supplementary Materials} \vspace{0.2cm}

\begin{itemize}

\item[A.] An alternative characterization of the basis transformation matrix $G_j$.
\item[B.] A practical basis storage framework using sliding windows.
\item[C.] Summary of rounding error analysis for classic CG and p-CG.
\item[D.] Supplementary numerical results on maximal attainable accuracy.

\end{itemize}

\noindent \hrulefill

{\parindent0pt\section{An alternative characterization of the basis transformation matrix $G_j$}} \label{sec:lemma} \\ 

\noindent We provide an interesting alternative characterization of the basis transformation matrix $G_j$, which is defined as $G_j = V_j^T Z_j$ in the p($l$)-CG method, see Theorem \ref{zisvg}. Lemma \ref{lemma:lemma_1} relates the matrix $G_j$ to the tridiagonal Lanczos matrix $T_{j-l}$.

\begin{lemma} \label{lemma:lemma_1}
Let $j \geq l+1$ and let $V_{l+1:j} = [v_l,\ldots,v_{j-1}]$  and $Z_{l+1:j} = [z_l,\ldots,z_{j-1}]$ denote subsets of the Krylov bases $V_j$ and $Z_j$. Let $G_{l+1:j}$ be the principal submatrix of $G_{j}$ that is obtained by removing the first $l$ rows and columns of $G_{j}$. Then
\begin{equation} \label{eq:G_sub}
  G_{l+1:j} = V^T_{l+1:j} \, V_{1:j-l} \, P_l(T_{j-l}) = \left( P_l(T_{j-l}) \, V^T_{1:j-l} \, V_{l+1:j} \right)^T,
\end{equation}
where 
\begin{equation}
	V^T_{l+1:j} \, V_{1:j-l} = 
	\left(\begin{array}{ccccc} 
		0&&1&& \\
		&\ddots&&\ddots& \\
		&&\ddots&&1 \\
		&&&\ddots& \\
		&&&&0 \\
	\end{array}\right) \leftarrow (l+1)^{\text{th}}~\text{row}.
\end{equation}
\end{lemma}

\begin{proof}
The proof follows directly from the definition \eqref{eq:defz} and the Lanczos relation, which when combined imply $P_l(A)V_j = V_j P_l(T_{j,j})$, and hence
\begin{equation}
  G_{l+1:j} = V^T_{l+1:j} \, Z_{l+1:j} = V^T_{l+1:j} \, P_l(A) \, V_{1:j-l} = V^T_{l+1:j} \, V_{1:j-l} \, P_l(T_{j-l}).
\end{equation}
\end{proof}

\noindent Lemma \ref{lemma:lemma_1} effectively states that for any $j \geq l+1$ the principal submatrix $G_{l+1:j}$ of the basis transformation matrix $G_{j}$ can be obtained by shifting all entries of the matrix $P_l(H_{j,j})$ upward by $l$ places and subsequently selecting the leading $(j-l)$-by-$(j-l)$ block.

\newpage

{\parindent0pt \section{A practical basis storage framework using sliding windows}} \label{sec:sliding} \\ 

\noindent Since in each iteration $i$ of p($l$)-CG, Alg.\,\ref{algo:plCG}, only the last $l+1$ vectors $z_{i-l+1}$, \ldots, $z_{i+1}$ are required, these vectors are stored in a sliding window of $l+1$ vectors. We denote this window by $\vec{Z}_i$, where the index $i$ indicates the \emph{current iteration}. Note that this is in contrast to the indexing for the basis $Z_{i} := [z_0,\ldots,z_{i-1}]$, where the index $i$ denotes the \emph{number of vectors} in the basis. The sliding window $\vec{Z}_i$ is defined as
\begin{equation*} 
\vec{Z}_i = [\vec{Z}_i(0),\vec{Z}_i(1),\ldots,\vec{Z}_i(l)] = \left\{ 
  \begin{matrix}
	  [~~, \ldots, ~~, z_{i+1}, \ldots, z_0], & \qquad 0 \leq i < l, \\ 
	  [z_{i+1}, z_i, \ldots, z_{i-l+1}], & \qquad i \geq l.
	\end{matrix}
\right. 
\end{equation*}
Vectors $z_j$ in the sliding window are listed from highest to lowest index, i.e.\,for $i \geq l$ the vector in the first position in the sliding window is $z_{i+1}$, and the vector in the last position is $z_{i-l+1}$, see also Fig.\,\ref{fig:basis_storage}. A particular vector $z_j$ with $\max(0,i-l+1) \leq j \leq i+1$ can be accessed from $\vec{Z}_i$ as follows:
\begin{equation*} 
z_j = \left\{ 
  \begin{matrix}
	  \vec{Z}_i(l-j), & 0 \leq i < l,\\ 
	  \vec{Z}_i(i-j+1), & i \geq l.
	\end{matrix}
\right. 
\end{equation*}
For iterations $0 \leq i \leq l-1$ the window $\vec{Z}_i$ is being filled up by simply adding vectors as follows:
\begin{equation*}
\vec{Z}_0 = [~~,\ldots,~~,~~,z_1,z_0], \quad \vec{Z}_1 = [~~,\ldots,~~,z_2,z_1,z_0], \quad \ldots, \quad \vec{Z}_{l-1} = [z_l,\ldots,z_3,z_2,z_1,z_0].
\end{equation*}
From iteration $i = l$ onward the window effectively starts to \emph{slide}: the most recently computed basis vector $z_{i+1}$ is written to position $\vec{Z}_i(0)$, and the vectors $z_i,\ldots,z_{i-l}$ from the previous window $\vec{Z}_{i-1}$ are all moved one space to a higher position in the window $\vec{Z}_i$. As a result of this procedure the vector $z_{i-l}$ is dropped from the window $\vec{Z}_i$. Hence we obtain for $i \geq l$:
\begin{equation*}
\vec{Z}_l = [z_{l+1},\ldots, z_3, z_2, z_1],\quad \vec{Z}_{l+1} = [z_{l+2},\ldots, z_4, z_3, z_2], \quad \ldots, \quad \vec{Z}_{i} =  [z_{i+1}, z_i, \ldots, z_{i-l+1}].
\end{equation*}
The procedure for filling and maintaining the sliding window $\vec{Z}_{i}$ is illustrated in Fig.\,\ref{fig:slidingVZ} (left) for pipeline length $l = 2$. In this case the window $\vec{Z}_{i}$ contains $l+1 = 3$ vectors in each iteration $i \geq l-1$.

Similarly to the auxiliary basis $Z_{i+2}$, each p($l$)-CG iteration $i$ uses the last $2l+1$ basis vectors $v_{i-3l+1}$, \ldots, $v_{i-l+1} \in V_{i-l+2}$ to update the solution, see Alg.\,\ref{algo:plCG}. These basis vectors are stored using a second sliding window $\vec{V}_i$ consisting of $2l+1$ vectors, where again the index $i$ refers to the iteration. The window $\vec{V}_i$ is defined similarly to the window $\vec{Z}_i$ above, i.e.:
\begin{equation*} 
\vec{V}_i = [\vec{V}_i(0),\vec{V}_i(1),\ldots,\vec{V}_i(2l)] = \left\{ 
  \begin{matrix}
	  [~~, \ldots, ~~, v_0], & \qquad 0 \leq i < l, \\ 
	  [~~, \ldots, ~~, v_{i-l+1}, \ldots, v_0], & \qquad l \leq i < 3l, \\ 
	  [v_{i-l+1}, v_{i-l}, \ldots, v_{i-3l+1}], & \qquad i \geq 3l,
	\end{matrix}
\right. 
\end{equation*}
and a particular vector $v_j$ with $\max(0,i-3l+1) \leq j \leq i-l+1$ can be accessed from $\vec{V}_i$ as follows:
\begin{equation*} 
v_j = \left\{ 
  \begin{matrix}
	  \vec{V}_i(2l-j), & l \leq i < 3l,\\ 
	  \vec{V}_i(i-l-j+1), & i \geq 3l.
	\end{matrix}
\right. 
\end{equation*}
The window $\vec{V}_i$ for the basis $V_{i-l+2}$ only starts to get filled once the window $\vec{Z}_i$ for the auxiliary basis $Z_{i+2}$ has been completely filled. Indeed, for iterations $0$ up to $l-1$ no vectors $v_{i-l+1}$ are computed in Alg.\,\ref{algo:plCG}, i.e.:
\begin{equation*}
\vec{V}_0 = [~~,\ldots,~~,~~,~~,v_0], \quad \vec{V}_1 = [~~,\ldots,~~,~~,~~,v_0], \quad \ldots, \quad \vec{V}_{l-1} = [~~,\ldots,~~,~~,~~,v_0].
\end{equation*}
In iterations $l$ up to $3l-1$ the window $\vec{V}_i$ is gradually filled by adding one vector $v_{i-l+1}$ in each iteration $i$:
\begin{equation*}
\vec{V}_l = [~~,\ldots,~~,~~,v_1,v_0], \quad \vec{V}_{l+1} = [~~,\ldots,~~,v_2,v_1,v_0], \quad \ldots, \quad \vec{V}_{3l-1} = [v_{2l},\ldots,v_3,v_2,v_1,v_0].
\end{equation*}
The window $\vec{V}_i$ is completely filled for the first time in iteration $3l-1$.
Consequently, $\vec{V}_i$ effectively starts to slide from iteration $3l$ onwards, i.e.:
\begin{equation*}
\vec{V}_{3l} = [v_{2l+1},\ldots, v_3, v_2, v_1],\quad \vec{V}_{3l+1} = [v_{2l+2},\ldots, v_4, v_3, v_2],\quad\ldots,\quad \vec{V}_{i} =  [v_{i-l+1}, v_{i-l}, \ldots, v_{i-3l+1}].
\end{equation*}
Fig.\,\ref{fig:slidingVZ} (right) illustrates the above procedure by showing a schematic overview of the sliding window $\vec{V}_i$ in the first iterations of Alg.\,\ref{algo:plCG} for pipeline length $l = 2$. Notice how in iteration $i$ the basis $V_{i-l+2}$, characterized by the sliding window $\vec{V}_i$, contains the last updated basis vector $v_{i-l+1} = v_{i-1}$. The basis $V_{i-l+2}$ thus runs $l$ vectors behind compared to the auxiliary basis $Z_{i+2}$ represented by the sliding window $\vec{Z}_i$, which contains the most recent auxiliary basis vector $z_{i+1}$, see Fig.\,\ref{fig:slidingVZ} (left).

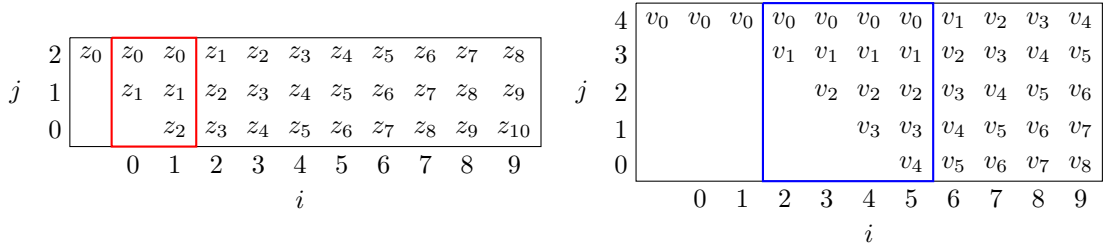
\begin{figure}\label{fig:slidingVZ}
\begin{center}
\begin{tikzpicture}[ node distance=1mm and 0mm, baseline,   proces/.style={
    anchor=center,
  }]
\matrix [matrix of nodes, nodes=proces, nodes in empty cells] (magic)
        {
          & 2 & $z_0$ & $z_0$ &  $z_0$  &  $z_1$ &  $z_2$ &  $z_3$ &  $z_4$ & $z_5$ &  $z_6$ & $z_7$ & $z_8$ \\
     $j$~ & 1 &       & $z_1$ &  $z_1$  &  $z_2$ &  $z_3$ &  $z_4$ &  $z_5$ & $z_6$ &  $z_7$ & $z_8$ & $z_9$ \\
          & 0 &       &       &  $z_2$  &  $z_3$ &  $z_4$ &  $z_5$ &  $z_6$ & $z_7$ &  $z_8$ & $z_9$ & $z_{10}$ \\
          &   &       & 0     &     1   &    2   &   3    &  4     &  5     &  6    & 7      &  8    &  9    \\
					&	  &       &       &         &        &        &   $i$   &       &       &        &       &       \\
        };
\draw[black,thin] 
(magic-1-3.north west) -| (magic-3-13.south east) -| (magic-1-3.north west);
\draw[red,thick] 
(magic-1-4.north west) -| (magic-3-5.south east) -| (magic-1-4.north west) ;
\end{tikzpicture}
\begin{tikzpicture}[ node distance=1mm and 0mm, baseline,   proces/.style={
    anchor=center,
  }]
\matrix [matrix of nodes, nodes=proces, nodes in empty cells] (magic)
        {
         & 4 & $v_0$& $v_0$&   $v_0$ & $v_0$ &   $v_0$ &  $v_0$  & $v_0$  & $v_1$ & $v_2$ & $v_3$ & $v_4$\\
         & 3 & $ $  & $ $  &   $ $   & $v_1$ &   $v_1$ &  $v_1$  & $v_1$  & $v_2$ & $v_3$ & $v_4$ & $v_5$\\
    $j$~ & 2 & $ $  & $ $  &   $ $   & $ $   &   $v_2$ &  $v_2$  & $v_2$  & $v_3$ & $v_4$ & $v_5$ & $v_6$\\
         & 1 & $ $  & $ $  &   $ $   & $ $   &   $ $   &  $v_3$  & $v_3$  & $v_4$ & $v_5$ & $v_6$ & $v_7$\\
         & 0 & $ $  & $ $  &   $ $   & $ $   &   $ $   &  $ $    & $v_4$  & $v_5$ & $v_6$ & $v_7$ & $v_8$\\
         &   &      & 0    &     1   &    2  &   3     &  4      &  5     &  6    & 7     &  8 & 9 \\
				 &   &      &      &         &       &         &   $i$   &        &       &       &    &   \\
        };
\draw[black,thin] 
(magic-1-3.north west) -| (magic-5-13.south east) -| (magic-1-3.north west) ;
\draw[blue,thick] 
(magic-1-6.north west) -| (magic-5-9.south east) -| (magic-1-6.north west) ;
\end{tikzpicture}

\vspace{-0.5cm} 

\end{center}
\caption{Schematic representation of the sliding storage windows $\vec{Z}_i(j)$ and $\vec{V}_i(j)$ in the first 10 iterations of Alg.\,\ref{algo:plCG} for pipeline length $l = 2$. \textbf{Left:} red box indicates the initial filling of the window $\vec{Z}_i$ in iterations $0$ up to $l-1$. The window $\vec{Z}_i$ effectively starts sliding from iteration $l$ onward.
\textbf{Right:} blue box indicates the initial filling of the window $\vec{V}_i$ in iterations $l$ up to $3l-1$, after the window $\vec{Z}_i$ has been completely filled (see left). The window $\vec{V}_i$ starts sliding from iteration $3l$ onward.}
\end{figure}

The concept of sliding windows can analogously be applied to the basis $\hat{Z}_i$ in the preconditioned version of the $l$-length pipelined CG method.
The sliding window for the preconditioned auxiliary basis $\hat{Z}_{i+2}$ is however limited to the last three vectors $\hat{z}_{i-1}$, $\hat{z}_{i}$ and $\hat{z}_{i+1}$, since only these vectors need to be stored in iteration $i$ of the algorithm, see Remark \ref{remark:prec_storage}.

The sliding of the window can easily be implemented in practice by re-addressing the pointers to the array elements in the windows $\vec{Z}_i$ and $\vec{V}_i$. For example, the C-code snippet \vspace{0.2cm}\\
\fbox{\parbox{0.97\textwidth}{
\begin{itemize}
\item[] \emph{\texttt{Vec *Z\_VEC, temp;}}
\item[] \emph{\texttt{temp = Z\_VEC[l];}}
\item[] \emph{\texttt{for(i = l; i>0; i--) Z\_VEC[i] = Z\_VEC[i-1];}}
\item[] \emph{\texttt{Z\_VEC[0] = temp;}}
\item[] \emph{\texttt{MatMult(A, Z\_VEC[1], Z\_VEC[0]);}}
\end{itemize} }}\vspace{0.2cm} \ \\
illustrates the sliding of the window $\vec{Z}_i$ in iteration $i \geq l$, where \emph{\texttt{Z\_VEC}} is an array of pointers to the vectors in $\vec{Z}_i$. The pointers are cycled such that $z_{i+1} =  Az_i$ can be added to the window as $\vec{Z}_i(0)$.

\newpage

{\parindent0pt\section{Summary of rounding error analysis for classic CG and p-CG}} \label{sec:analysis_cg} \\ 

\noindent We provide a brief overview of the analysis of local rounding errors in classic CG and p-CG, which was performed in detail in \cite{cools2018analyzing} and the related work \cite{carson2016numerical}. This section is intended as an easy reference to compare the numerical analysis of the p($l$)-CG method to the existing CG and p-CG methods.

\subsection{Local rounding error behavior in finite precision classic CG}

Consider the propagation of local rounding errors by the recurrence relations of classic CG, Alg.\,\ref{algo:CG}, given by expression \eqref{eq:f_CG}.
By introducing the matrix notation $\bar{\bold{R}}_{j+1} - \bar{R}_{j+1} = [\bar{\bold{r}}_0 - \bar{r}_0, \ldots, \bar{\bold{r}}_j - \bar{r}_j]$ for the residual gaps in the first $j+1$ iterations and by analogously defining $\Theta_j^{\bar{x}} = [0,-\xi_1^{\bar{x}},\ldots, -\xi_{j-1}^{\bar{x}}]$ and $\Theta_j^{\bar{r}} = [f_0,-\xi_1^{\bar{r}},\ldots, -\xi_{j-1}^{\bar{r}}]$ for the local rounding errors, expression \eqref{eq:f_CG} can be formulated as
\begin{equation*} 
  \bar{\bold{R}}_{j+1} - \bar{R}_{j+1} = (A \Theta_{j+1}^{\bar{x}} + \Theta_{j+1}^{\bar{r}}) \, E_{j+1},
\end{equation*}
where $E_{j+1}$ is a $(j+1) \times (j+1)$ upper triangular matrix of ones. 
Since all entries of $E_{j+1}$ are one, local rounding errors are merely accumulated (not amplified) in the classic CG algorithm.

\subsection{Local rounding error behavior in finite precision p-CG}

The pipelined p-CG method proposed in \cite{ghysels2014hiding}, see Alg.\,\ref{algo:PIPECG}, uses additional recurrence relations for auxiliary vector quantities defined as $w_j := A r_j$, $s_j := A p_j$ and $z_j := A s_j$. The coupling between these recursively defined variables may cause local rounding error amplification. In finite precision p-CG the following recurrence relations are computed:
\begin{align*}
	\bar{x}_{j+1} &= \bar{x}_j + \bar{\alpha}_j \bar{p}_j + \xi_{j+1}^{\bar{x}}, &
	\bar{w}_{j+1} &= \bar{w}_j - \bar{\alpha}_j \bar{z}_j + \xi_{j+1}^{\bar{w}}, &
	\bar{r}_{j+1} &= \bar{r}_j - \bar{\alpha}_j \bar{s}_j + \xi_{j+1}^{\bar{r}}, \notag \\
	\bar{s}_j 		&= \bar{w}_j + \bar{\beta}_j \bar{s}_{j-1} + \xi_j^{\bar{s}}, &
	\bar{p}_j 		&= \bar{r}_j + \bar{\beta}_j \bar{p}_{j-1} + \xi_j^{\bar{p}} , &
	\bar{z}_j     &= A\bar{w}_j + \bar{\beta}_j \bar{z}_{j-1} + \xi_j^{\bar{z}}. \label{eq:pxr_pipecg}
\end{align*}
The respective bounds for the local rounding errors $\xi_{k}^{\bar{x}},\xi_{k}^{\bar{r}},\xi_{k}^{\bar{p}},\xi_{k}^{\bar{s}},\xi_{k}^{\bar{w}}$ and $\xi_{k}^{\bar{z}}$ in these recurrence relations can be found in \cite{cools2018analyzing}, where it is also shown that the residual gap $(b-A\bar{x}_j) - \bar{r}_j$ is coupled to the gaps $A\bar{p}_j - \bar{s}_j$, $A\bar{r}_j - \bar{w}_j$ and $A\bar{s}_j - \bar{z}_j$ on the auxiliary variables in p-CG.

Let $B = [b,b,\ldots,b]$, $\bar{X}_{j+1} = [\bar{x}_0,\bar{x}_1,\ldots, \bar{x}_j]$ and $\bar{P}_{j+1} = [\bar{p}_0,\bar{p}_1,\ldots, \bar{p}_j]$. Writing 
the gaps defined in \cite{cools2018analyzing} (Section 2.3) 
in matrix notation as 
$\bar{\bold{R}}_{j+1} - \bar{R}_{j+1}$, $\bar{\bold{S}}_{j+1} - \bar{S}_{j+1}$, $\bar{\bold{W}}_{j+1} - \bar{W}_{j+1}$, $\bar{\bold{Z}}_{j+1} - \bar{Z}_{j+1}$ with actual variables that are defined as $\bar{\bold{R}}_{j+1} = B-A\bar{X}_{j+1}$, $\bar{\bold{S}}_{j+1} = A\bar{P}_{j+1}$, $\bar{\bold{W}}_{j+1} = A\bar{R}_{j+1}$ and $\bar{\bold{Z}}_{j+1} = A\bar{S}_{j+1}$,
and using the expressions for the local rounding errors on the auxiliary variables: 
$\Theta_j^{\bar{x}} = -[0,\xi_1^{\bar{x}},..., \xi_{j-1}^{\bar{x}}]$, 
$\Theta_j^{\bar{r}} = [f_0,-\xi_1^{\bar{r}},..., -\xi_{j-1}^{\bar{r}}]$, 
$\Theta_j^{\bar{p}} = [0,\xi_1^{\bar{p}},..., \xi_{j-1}^{\bar{p}}]$,  
$\Theta_j^{\bar{s}} = [g_0,-\xi_1^{\bar{s}},..., -\xi_{j-1}^{\bar{s}}]$,
$\Theta_j^{\bar{u}} = [0,\xi_1^{\bar{r}},..., \xi_{j-1}^{\bar{r}}]$,  
$\Theta_j^{\bar{w}} = [h_0,-\xi_1^{\bar{w}},..., -\xi_{j-1}^{\bar{w}}]$, 
$\Theta_j^{\bar{q}} = [0,\xi_1^{\bar{s}},..., \xi_{j-1}^{\bar{s}}]$,  
$\Theta_j^{\bar{z}} = [e_0,-\xi_1^{\bar{z}},..., -\xi_{j-1}^{\bar{z}}]$, 
the following matrix expressions for the gaps in p-CG are obtained: 
\begin{align*} 
  \bar{\bold{R}}_{j+1} - \bar{R}_{j+1} &= (A\Theta_{j+1}^{\bar{x}}+\Theta_{j+1}^{\bar{r}}) \, E_{j+1} + (\bar{\bold{S}}_{j+1} - \bar{S}_{j+1}) \bar{\mathcal{A}}_{j+1}, \\
  \bar{\bold{S}}_{j+1} - \bar{S}_{j+1} &= (A\Theta_{j+1}^{\bar{p}}+\Theta_{j+1}^{\bar{s}}) \, \bar{\mathcal{B}}^{-1}_{j+1} + (\bar{\bold{W}}_{j+1} - \bar{W}_{j+1}) \bar{\mathcal{B}^{-1}_{j+1}}, \\ 
  \bar{\bold{W}}_{j+1} - \bar{W}_{j+1} &= (A\Theta_{j+1}^{\bar{u}}+\Theta_{j+1}^{\bar{w}}) \, E_{j+1} + (\bar{\bold{Z}}_{j+1} - \bar{Z}_{j+1}) \bar{\mathcal{A}}_{j+1}, \\
  \bar{\bold{Z}}_{j+1} - \bar{Z}_{j+1} &= (A\Theta_{j+1}^{\bar{q}}+\Theta_{j+1}^{\bar{z}}) \, \bar{\mathcal{B}}^{-1}_{j+1}.
\end{align*}
By substituting these expressions we obtain the following expression for the residual gaps in p-CG:
\begin{align*}
  \bar{\bold{R}}_{j+1} - \bar{R}_{j+1} &= (A\Theta_{j+1}^{\bar{x}}+\Theta_{j+1}^{\bar{r}}) \, E_{j+1} + (A\Theta_{j+1}^{\bar{p}}+\Theta_{j+1}^{\bar{s}}) \, \bar{\mathcal{B}}^{-1}_{j+1} \bar{\mathcal{A}}_{j+1} + \ldots \notag \\
	& \quad + (A\Theta_{j+1}^{\bar{u}}+\Theta_{j+1}^{\bar{w}}) \, E_{j+1} \bar{\mathcal{B}^{-1}_{j+1}} \bar{\mathcal{A}}_{j+1} + (A\Theta_{j+1}^{\bar{q}}+\Theta_{j+1}^{\bar{z}}) \, \bar{\mathcal{B}}^{-1}_{j+1} \bar{\mathcal{A}}_{j+1} \bar{\mathcal{B}^{-1}_{j+1}} \bar{\mathcal{A}}_{j+1}, 
\end{align*}
where \vspace{-0.2cm}
\begin{equation*}
  \bar{\mathcal{A}}_{j+1} = -
  \left(\begin{array}{ccccc} 
    0&\bar{\alpha}_0&\bar{\alpha}_0&\cdots&\bar{\alpha}_0 \\
    &0&\bar{\alpha}_1&\cdots&\bar{\alpha}_1 \\
    &&\ddots&& \vdots\\
    &&&0& \bar{\alpha}_{j-1}\\
    &&&&0
  \end{array}\right),
\qquad
	\bar{\mathcal{B}}_{j+1}^{-1} = 
	\left(\begin{array}{ccccc} 
		1&\bar{\beta}_1&\bar{\beta}_1 \bar{\beta}_2&\cdots&\bar{\beta}_1 \bar{\beta}_2 \ldots \bar{\beta}_j \\
		&1&\bar{\beta}_2& \ddots &\bar{\beta}_2 \ldots \bar{\beta}_j\\
		&&\ddots&\ddots&\vdots \\
		&&&\ddots&\bar{\beta}_j \\
		&&&&1
	\end{array}\right).
\end{equation*}
Hence, the entries of the coefficient matrices $\bar{\mathcal{B}}^{-1}_{j+1}$ and $\bar{\mathcal{A}}_{j+1}$ determine the propagation of the local rounding errors in p-CG.
The entries of $\bar{\mathcal{B}}^{-1}_{j+1}$ consist of a product of the scalar coefficients $\bar{\beta_{j}}$. In exact arithmetic these coefficients equal $\beta_j = \|r_j\|^2/\|r_{j-1}\|^2$, such that
\begin{equation*}
	\beta_i \, \beta_{i+1} \, \ldots \, \beta_{j} = \frac{\|r_i\|^2}{\|r_{i-1}\|^2} \frac{\|r_{i+1}\|^2}{\|r_{i}\|^2} \ldots \frac{\|r_j\|^2}{\|r_{j-1}\|^2} = \frac{\|r_j\|^2}{\|r_{i-1}\|^2}, \quad i \leq j.
\end{equation*}
Since the residual norm in CG is not guaranteed to decrease monotonically, the factor $\|r_j\|^2/\|r_{i-1}\|^2$ may for some $i \leq j$ be much larger than one. A similar argument may be used in the finite precision framework to derive that some entries of $\bar{\mathcal{B}}^{-1}_{j+1}$ may be significantly larger than one, and may hence (possibly dramatically) amplify the corresponding local rounding errors. 
This behavior is illustrated in Section \ref{sec:experiments} by Fig.\,\ref{fig:accuracy1}, \ref{fig:figure2}, \ref{fig:figure3} and \ref{fig:figure7}, where the p-CG residual norm typically stagnates at a reduced maximal attainable accuracy level compared to classic CG.

\begin{minipage}[t]{0.46\textwidth}
\begin{algorithm}[H]
{\small
  \caption{Conjugate Gradient method (CG) \hfill \textbf{Input:} $A$, $b$, $x_0$, $m$, $\tau$} \label{algo:CG}
  \begin{algorithmic}[1]
    \State $r_0 := b - Ax_0$;
		\State $p_0 := r_0$; 
    \For{$i = 0, \dots, m$}
    \State $s_i := Ap_{i}$;
    \State $\alpha_{i} := \left( r_i, r_i \right) / \left( s_i, p_i \right)$;
		\State \textbf{end if} 
    \State $x_{i+1} := x_i + \alpha_{i} p_i$;
    \State $r_{i+1} := r_i - \alpha_{i} s_i$;
    \State $\beta_{i+1} := \left( r_{i+1}, r_{i+1} \right) / \left( r_i, r_i \right)$;
    \State $p_{i+1} := r_{i+1} + \beta_{i+1} p_i$;
    \EndFor
		\State \textbf{end for} 
  \end{algorithmic}
}
\end{algorithm}
\end{minipage}
\hfill
\begin{minipage}[t]{0.46\textwidth}
\begin{algorithm}[H]
{\small
  \caption{Pipelined Conjugate Gradient method (p-CG) \hfill \textbf{Input:} $A$, $b$, $x_0$, $m$, $\tau$}  \label{algo:PIPECG}
  \begin{algorithmic}[1]
    \State $r_0 := b - Ax_0$;
		\State $w_0 := Ar_0$;
    \For{$i = 0,\dots,m$}
    \State $\gamma_i := (r_i,r_i)$;
    \State $\delta_i := (w_i,r_i)$;
    \State $v_i := A w_i$;
		\State \textbf{end if} 
    \If{$i>0$}
    \State $\beta_i := \gamma_i/\gamma_{i-1}$; 
		\State $\alpha_i := (\delta_i/\gamma_i - \beta_i/\alpha_{i-1})^{-1}$;
    \Else
    \State $\beta_i := 0$; 
		\State $\alpha_i := \gamma_i/\delta_i$;
    \EndIf
		\State \textbf{end if} 
    \State $z_i := v_i + \beta_i z_{i-1}$;
    \State $s_i := w_i + \beta_i s_{i-1}$;
    \State $p_i := r_i + \beta_i p_{i-1}$;
    \State $x_{i+1} := x_i + \alpha_i p_i$;
    \State $r_{i+1} := r_i - \alpha_i s_i$;
    \State $w_{i+1} := w_i - \alpha_i z_i$;
    \EndFor
		\State \textbf{end for} 
  \end{algorithmic}
}
\end{algorithm}
\end{minipage}

\newpage

\begin{sidewaystable}
\section{Supplementary numerical results on maximal attainable accuracy\vspace{1.0cm}} \label{sec:supplementary}
\centering
\footnotesize
\begin{tabular}{| l | r r r r | r | r | r | r | r | r | r | r |r |}
\hline 
    Matrix   & Prec   & $\kappa(A)$&$n$& \#$nnz$     &$\|b\|_2$ & CG & p-CG & \hspace{-0.1cm}p($1$)-CG\hspace{-0.1cm} & \hspace{-0.1cm}p($2$)-CG\hspace{-0.1cm} & \hspace{-0.1cm}p($3$)-CG\hspace{-0.1cm} & \hspace{-0.1cm}p($4$)-CG\hspace{-0.1cm} & \hspace{-0.1cm}p($5$)-CG\hspace{-0.1cm} & iter \\
\hline \hline
    bcsstk14 & JAC    & 1.3e+10 & 1806   & 63,454    & 2.1e+09  & 8.2e-16  & 2.3e-12  & 2.3e-16  & 5.9e-13 & 7.3e-12  & 3.0e-12  & 1.2e-10  & 700\\
    bcsstk15 & JAC    & 8.0e+09 & 3948   & 117,816   & 4.3e+08  & 3.7e-15  & 2.4e-12  & 3.2e-14  & 2.2e-06 & 3.5e-06  & 2.1e-06  & 2.0e-06  & 780 \\
    bcsstk16 & JAC    & 65      & 4884   & 290,378   & 1.5e+08  & 3.7e-15  & 6.3e-12  & 1.1e-14  & 8.5e-12 & 5.6e-11  & 1.5e-10  & 8.5e-11  & 300 \\
    bcsstk17 & JAC    & 65      & 10,974 & 428,650   & 9.0e+07  & 1.5e-14  & 4.4e-09  & 1.4e-04  & 2.1e-06 & 3.5e-06  & 1.8e-06  & 5.3e-06  & 3600 \\
    bcsstk18 & JAC    & 65      & 11,948 & 149,090   & 2.6e+09  & 2.3e-15  & 1.2e-10  & 5.4e-11  & 3.9e-13 & 1.1e-12  & 3.0e-11  & 1.6e-11  & 2400\\
    bcsstk27 & JAC    & 7.7e+04 & 1224   & 56,126    & 1.1e+05  & 3.6e-15  & 1.8e-11  & 1.2e-14  & 2.3e-11 & 9.2e-09  & 1.1e-08  & 7.7e-09  & 350 \\
  gr\_30\_30 &  -     & 3.8e+02 & 900    & 7744      & 1.1e+00  & 2.8e-15  & 3.1e-13  & 8.9e-15  & 1.6e-14 & 1.9e-15  & 1.9e-15  & 2.1e-15  & 60\\
    nos1     & *ICC   & 2.5e+07 & 237    & 1017      & 5.7e+07  & 1.1e-14  & 4.2e-10  & 4.3e-11  & 1.3e-05 & 7.9e-05  & 6.0e-03  & 6.6e-05  & 350\\
    nos2     & *ICC   & 6.3e+09 & 957    & 4137      & 1.8e+09  & 8.3e-14  & 1.0e-07  & 4.4e-06  & 1.4e-05 & 1.4e-05  & 9.7e-06  & 1.0e-05  & 3180\\
    nos3     & ICC    & 7.3e+04 & 960    & 15,844    & 1.0e+01  & 9.6e-15  & 1.3e-12  & 2.4e-14  & 4.1e-14 & 2.4e-14  & 9.8e-15  & 5.4e-13  & 65\\
    nos4     & ICC    & 2.7e+03 & 100    & 594       & 5.2e-02  & 1.9e-15  & 3.5e-14  & 7.2e-16  & 7.6e-16 & 3.6e-15  & 4.0e-15  & 6.0e-15  & 33\\
    nos5     & ICC    & 2.9e+04 & 468    & 5172      & 2.8e+05  & 3.1e-16  & 6.7e-14  & 2.8e-16  & 1.8e-16 & 2.3e-16  & 6.8e-16  & 2.7e-16  & 63\\
    nos6     & ICC    & 8.0e+06 & 675    & 3255      & 8.6e+04  & 5.0e-15  & 4.1e-11  & 4.8e-14  & 6.1e-09 & 5.0e-08  & 4.3e-08  & 2.4e-08  & 34\\
    nos7     & ICC    & 4.1e+09 & 729    & 4617      & 8.6e-03  & 3.1e-08  & 1.1e-07  & 5.4e-08  & 9.9e-08 & 1.6e-07  & 2.7e-05  & 5.9e-04  & 31\\
    s1rmq4m1 & ICC    & 1.8e+06 & 5489   & 262,411   & 1.5e+04  & 4.7e-15  & 5.5e-12  & 8.5e-15  & 8.1e-14 & 4.9e-15  & 3.1e-15  & 2.3e-14  & 135\\
    s1rmt3m1 & ICC    & 2.5e+06 & 5489   & 217,651   & 1.5e+04  & 8.9e-15  & 4.1e-11  & 3.0e-15  & 2.7e-13 & 3.2e-13  & 2.7e-12  & 4.0e-13  & 245\\
    s2rmq4m1 & *ICC   & 1.8e+08 & 5489   & 263,351   & 1.5e+03  & 7.1e-15  & 3.0e-10  & 3.4e-13  & 6.0e-11 & 6.8e-06  & 1.4e-05  & 1.1e-05  & 370\\
    s2rmt3m1 & ICC    & 2.5e+08 & 5489   & 217,681   & 1.5e+03  & 2.3e-14  & 7.4e-10  & 1.2e-12  & 4.9e-12 & 9.6e-06  & 7.1e-06  & 6.9e-06  & 265\\
    s3rmq4m1 & *ICC   & 1.8e+10 & 5489   & 262,943   & 1.5e+02  & 1.5e-14  & 2.9e-08  & 1.9e-06  & 1.8e-06 & 2.5e-05  & 9.5e-07  & 3.9e-07  & 1650\\
    s3rmt3m1 & *ICC   & 2.5e+10 & 5489   & 217,669   & 1.5e+02  & 2.9e-14  & 1.0e-07  & 2.3e-09  & 2.3e-07 & 3.5e-07  & 4.9e-07  & 5.2e-07  & 2282\\
    s3rmt3m3 & *ICC   & 2.4e+10 & 5357   & 207,123   & 1.3e+02  & 3.2e-14  & 2.4e-07  & 1.4e-07  & 1.7e-07 & 4.8e-06  & 4.9e-06  & 1.6e-05  & 2862\\
\hline
\end{tabular}
\caption{A random selection of real, non-diagonal and symmetric positive definite matrices from Matrix Market, listed with their respective condition number $\kappa(A)$, number of rows/columns $n$ and total number of nonzeros \#$nnz$. A linear system with right-hand side $b = A\hat{x}$ where $\hat{x}_i = 1/\sqrt{n}$ is solved with each matrix. The initial guess is all-zero $\bar{x}_0 = 0$. Preconditioners Jacobi (JAC), Incomplete Cholesky factorization (ICC) and compensated Incomplete Cholesky with global diagonal shift (ICC*) are included where required. For each problem the number of iterations (and thus the number of \textsc{spmv} and preconditioner applications) performed is fixed for all methods. The number of iterations performed is fixed per problem for all methods and is based on classic CG reaching maximal attainable accuracy (stagnation point). The relative residuals $\|b-A \bar{x}_i\|_2/\|b\|_2$ are shown for all methods. Pipelined methods generally reach a lower precision compared to classic CG for the same number of \textsc{spmv}s (iterations) and the loss of attainable accuracy is more pronounced for longer pipelines, cf. Section~\ref{sec:analysis}.}
\label{tab:matrix_market}
\end{sidewaystable}

\newpage

\end{document}